\documentclass[12pt]{article}
\usepackage{amsmath}
\usepackage{amsfonts}
\usepackage{graphicx,psfrag,epsf}
\usepackage{enumerate}
\usepackage[usenames, dvipsnames]{color}
\usepackage{url} 
\usepackage[noblocks]{authblk}

\usepackage{adjustbox}
\usepackage{mdframed}
\usepackage{enumitem}
\usepackage{caption}
\usepackage{subcaption}
\usepackage{siunitx}
\usepackage{float}
\usepackage{booktabs}
\usepackage[normalem]{ulem}
\newcommand{\stkout}[1]{\ifmmode\text{\sout{\ensuremath{#1}}}\else\sout{#1}\fi}
\usepackage{mathtools}
\usepackage{mdwlist}
\usepackage{multirow} 
\usepackage{comment}

\usepackage{inputenc}
\usepackage{csquotes}
\usepackage[round]{natbib}
{
\begin{basedescript}{
\desclabelwidth{2em}}}
{
\end{basedescript}
}

\numberwithin{table}{section}

\usepackage{verbatim}
\usepackage{diagbox}
\usepackage{bbm}
\usepackage{algorithm}
\usepackage[noend]{algpseudocode}
\algrenewcommand\algorithmicrequire{\textbf{Input:}}

\algrenewcommand\algorithmicensure{\textbf{Output:}}

\usepackage[english]{babel}
\usepackage{amsthm}
\usepackage{mathtools}
\newtheorem{theorem}{Theorem}[section]

\usepackage{xr-hyper}
\usepackage{hyperref}
\hypersetup{
    colorlinks=true,
    linkcolor=blue,
    citecolor=blue,
    filecolor=magenta,      
    urlcolor=blue,
}

\newcommand\bzero{\mbox{\boldmath${0}$}}

\newcommand\btheta{\mbox{\boldmath${\theta}$}}

\newcommand\bC{{\bf C}}

\newcommand\bI{{\bf I}}

\newcommand\bR{{\bf R}}

\newcommand\bs{{\boldsymbol s}}

\newcommand\bT{{\bf T}}
\newcommand\bU{{\bf U}}
\newcommand\bu{{\bf u}}

\newcommand\bZ{{\bf Z}}
\newcommand\bz{{\bf z}}

\newcommand\mbR{{\mathbb R}}

\usepackage{subcaption}

\newcommand{\lowcase}[4]{
\begin{subfigure}{0.48\textwidth}
\centering
\includegraphics[width=\linewidth]{Figures/boxplots/low_dim/plots/case_#1.pdf}
\caption{$\alpha=#2$, $\nu=#3$, $\rho=#4$}
\end{subfigure}
}

\newcommand{\highcase}[4]{
\begin{subfigure}{0.48\textwidth}
\centering
\includegraphics[width=\linewidth]{Figures/boxplots/high_dim/plots/case_#1.pdf}
\caption{$\alpha=#2$, $n=#3$, $T=#4$}
\end{subfigure}
}

\usepackage{setspace}
\usepackage{tocloft}

\cftsetpnumwidth{4em}
\cftsetrmarg{5em}
\cftsetindents{subsection}{3em}{3.5em}

\AtBeginDocument{\setstretch{1.68}}

\makeatletter
\renewcommand*\env@matrix[1][\arraystretch]{%
  \edef\arraystretch{#1}%
  \hskip -\arraycolsep
  \let\@ifnextchar\new@ifnextchar
  \array{*\c@MaxMatrixCols c}}
\makeatother

\usepackage{etoolbox}
\makeatletter
\g@addto@macro\normalsize{%
  \setlength\abovedisplayskip{3.5pt}
  \setlength\belowdisplayskip{3pt}
  \setlength\abovedisplayshortskip{0pt}
  \setlength\belowdisplayshortskip{0pt}
}
\makeatother

\begin{document}

\def\spacingset#1{\renewcommand{\baselinestretch}%
{#1}\small\normalsize} \spacingset{1}

\newcommand{\blind}{0}  

\title{\bf Spatial Extremes at Scale: A Case Study of Surface Skin Temperature and Heat Risk in the United States}

\if\blind1
\author{}
\else
\author[1]{Ben Seiyon Lee}
\author[2]{Reetam Majumder}
\author[3]{Jordan Richards}
\author[4]{Emma S.\ Simpson}
\author[5]{Likun Zhang}
\affil[1]{Department of Statistics, George Mason University}
\affil[2]{Department of Mathematical Sciences, University of Arkansas}
\affil[3]{School of Mathematics and Maxwell Institute for Mathematical Sciences, University of Edinburgh}
\affil[4]{Department of Statistical Science, University College London}
\affil[5]{Department of Statistics, University of Missouri}
\fi

\date{}

\spacingset{1}

\maketitle
\begin{abstract}
Understanding and mapping extreme heat is critical for risk management and public health planning, particularly in regions with complex terrain and heterogeneous climate. We present a case study of extreme heat in the Four Corners region of the United States, using high-resolution surface skin temperature data from the North American Land Data Assimilation System to characterize spatially heterogeneous and seasonally varying extremes across complex terrain, and to assess their implications for heat-related public health risks. Spatial extremes exhibit complex dependencies across geographic regions, which require sophisticated statistical models to capture. While recent advances in spatial extreme value modeling provide flexible representations of joint tail dependencies, statistical inference remains computationally demanding, especially for datasets with a large number of locations. To address this, we propose a random scale mixture process that facilitates Bayesian inference of spatial extremes, and develop scalable inference strategies that leverage advances in spatial modeling and amortized learning. We evaluate the proposed inference methods through large-scale simulation studies, representing the first such extensive study in spatial extremes, and a high-resolution surface skin temperature application in the Four Corners region. Surface skin temperature is particularly useful as a predictor for air temperature, for studying heatwaves and related environmental phenomena, and to calculate heat indices reflecting downstream health risks at any location. Our findings provide insights into efficient, data-driven approaches for modeling spatial extremes, and serve as guidelines for practitioners in the fields of climate science, environmental risk assessment, and beyond. 
\end{abstract}

\noindent%

\newpage
\doublespacing
\setcounter{page}{1}
\section{Introduction}\label{s:intro}

Extreme temperatures play a central role in environmental risk assessment and serve as a key input to widely used risk indices for droughts \citep[][]{keetch1968drought}, fire weather \citep[][]{Wagner1987}, and health risks \citep{steadman1979assessment}. Modeling extreme temperatures is also critical for assessing climate change impacts \citep{dore2005climate}, as record-breaking heat and freezing events can trigger floods, droughts, wildfires, agricultural frost damage, and other hazards that affect property, infrastructure, and ecosystems \citep{change2013physical,easterling2016detection}.  Studying spatial extreme temperature events can potentially strengthen early warning systems \citep{stott2004human}, improve the allocation of critical resources \citep{kundzewicz2008implications}, and guide policy development aimed at safeguarding both communities and ecosystems. However, extreme temperature events often exhibit strong and complex spatial dependence, where high or low temperatures occur simultaneously over large geographic regions \citep{seneviratne2012changes,easterling2016detection}. This is difficult to model with standard statistical models. 

Our case study focuses on the accurate spatial interpolation of temperature at \textit{any} spatial location within the study domain, while \textit{faithfully characterizing} extremal dependence. A key downstream application is heat stress assessment, the leading cause of weather-related fatalities in the United States \citep{NOAA2023HazStat}. Heat indices \citep{steadman1979assessment} are a key component for this purpose, combining temperature and humidity into a single measure of perceived heat stress that can guide public health advisories and risk communication \citep{anderson2013methods}. 
Since the heat index is highly sensitive to extreme temperatures, their accurate modeling is essential for generating reliable heat index estimates \citep{steadman1979assessment}. Gridded products, such as the North American Land Data Assimilation System \citep[NLDAS;][]{xia2012continental}, provide spatially indexed surface skin temperature fields suitable for extreme value analysis, which can be combined with humidity measures to project heat indices at unobserved locations. However, their high spatial resolution and long temporal record make analysis computationally demanding, necessitating scalable statistical inference and prediction methods.

Gaussian processes (GPs) are the cornerstone of geostatistical spatial modeling \citep{cressie2011statistics}, but are not amenable to heavy-tailed marginal behavior or strong extremal dependence \citep{davison2013geostatistics}. Crucially, dependence in the extremes and the bulk may exhibit markedly different behavior \citep{huser2017bridging}. That is, extremal spatial dependence may weaken as events become more severe and eventually vanish (commonly termed \emph{asymptotic independence}) or, conversely, persist even in the far tail (\emph{asymptotic dependence}), so that extremes occur concurrently across locations with non-negligible probability. 

For spatial extremes, max-stable processes are a standard class of models for site-wise block maxima. Under max-stability, the extremal dependence structure remains effectively unchanged with event severity. Therefore, they are restricted to asymptotic dependence and cannot capture weakening dependence with increasing severity. As such, recent work advocates the use of more flexible alternatives for modeling spatial extremes \citep{huser2025modeling}. In addition, {it is well-documented that} inference for max-stable processes can be computationally prohibitive for even a small number of locations ($\sim 20$) \citep[see, e.g.,][]{ castruccio2016high, huser2019full}. Random scale mixture processes \citep[e.g.,][]{opitz2016modeling,huser2017bridging,huser2019modeling} provide a unified framework for spatial extremes, accommodating both asymptotic dependence and asymptotic independence. In this case study, we propose a novel random scale mixture process that enables a smooth transition between extremal dependence classes similarly to the \citet{huser2019modeling} model, and is amenable to high-dimensional Bayesian inference.

Many inference approaches for spatial extremes have been developed \citep[see the review by][]{huser2022}, but systematic comparisons of competing scalable methods remain limited. This is particularly notable given that computational cost is a primary bottleneck. Many approaches lack tractable Bayesian formulations \citep{zhang-2022a}, and performance varies across extremal dependence regimes. Although a wide range of spatial extremes models is available, the practical uptake of spatial extremal inference methods remains slow and inefficient. A few important exceptions exist. For example, \citet{heaton2019case} provide a landmark large-scale comparison study, though within the Gaussian process setting, while \citet{shi2026log} show that a log-Laplace nugget can reduce the computational burden of threshold-based spatial extremes inference to essentially standard covariance-matrix operations.

Through a detailed real-world case study, we provide the first systematic comparison of scalable Bayesian inference methods for spatial extremes, evaluating computational efficiency and the recovery of extremal dependence across asymptotic dependence regimes. We exploit recent advances in scalable Gaussian process modeling, adapting these approaches for our {newly-proposed} random scale mixture model. This results in four previously unexplored inference approaches for studying spatial extremes, exploiting methods that approximate or replace the full likelihood function, including Vecchia approximations \citep{vecchia1988estimation, katzfuss2020vecchia}, covariance tapering \citep{furrer2006covariance}, and low-rank basis representations \citep{cressie2022basis, wikle2010low}, as well as simulation-based inference via neural Bayes estimators \citep[NBEs;][]{sainsbury2022fast}. We evaluate the performance of these approaches through two extensive simulation studies, assessing how each method recovers spatial extremal dependence structures and risk-relevant model parameters, across a range of (extremal) dependence settings. This comprehensive simulation study provides a convincing justification for the use of our model to perform inference in our case study, which examines a 46-year record of surface skin temperatures over the climatologically diverse Four Corners region of the Southwest United States (US). This case study demonstrates how accurate interpolation of surface skin temperature improves the estimation of near-surface air temperature and, consequently, improves heat index calculations for public health warnings, and helps identify vulnerable communities with elevated heat risk.

The remainder of the paper is organized as follows. {Section~\ref{s:motivation} introduces our dataset of surface skin temperatures and highlights current scientific challenges within the context of climate risk management and decision making. In Section~\ref{s:modeling_framework}, we introduce key concepts from spatial extreme value analysis, and 
propose our new modeling framework.} Section~\ref{s:comp_strategies} describes the computational inference strategies that we compare and validate in this study. Section~\ref{s:simulation_study} presents results from our comparative simulation studies, examining performance in both low-dimensional (fewer locations and time points) and high-dimensional settings (large spatial grids with long temporal records). Section~\ref{s:application} provides the results of our application to the NLDAS skin temperature extremes across the Four Corners region. Section~\ref{s:discussion} concludes with a discussion and recommendations for practitioners.

\section{Scientific Motivation and Case Study Data}\label{s:motivation}
This section outlines the scientific motivation for our study and introduces the case study dataset. We discuss the spatial heterogeneity and seasonal patterns of heat extremes, with a focus on surface skin temperatures. We then introduce the NLDAS dataset and present an exploratory data analysis, including the preprocessing steps used to prepare the data for spatial extremal dependence modeling.

\subsection{Background}\label{s:motivation:overview}
The frequency and intensity of heat extremes have increased globally since 1950, with changes in extreme temperatures often exceeding those in mean surface temperature \citep{seneviratne2014no}. In the western United States, heat increasingly drives drought and wildfire activity \citep{Westerling2006,Westerling2011} and poses risks to critical infrastructure \citep{NCA5Ch5,NCA5Ch13}. These trends highlight the need to quantify extreme heat risk across space and seasons. 

Despite this need, spatial coverage of data from weather stations remains sparse, particularly across complex terrain \citep{minder2010surface} and in urban settings \citep{lee2016study}. Reanalysis products, such as MERRA-2 \citep{MERRA-2} and ERA5 \citep{ERA5}, provide spatially complete fields, but these are at relatively coarse scales ($\sim$ 30--50 km). In complex terrain or urban environments, $T_{2m}$ exhibits fine-scale variation that coarse grids cannot resolve. As a result, direct interpolation of $T_{2m}$ reanalysis fields can miss key physical factors that affect temperature variability \citep{oyler2016remotely}.

Our study focuses on daily surface skin temperature \citep[$T_{skin}$;][]{norman1995terminology}, which provides a physically consistent representation of near‐surface thermal conditions \citep{jin2010land}. Following common practice in remote sensing, we use $T_{skin}$ as an intermediate variable and infer near-surface air temperature at 2-m ($T_{2m}$) via supervised learning \citep{oyler2016remotely}. This approach incorporates physically meaningful predictors like wind, elevation, and the normalized difference vegetation index (NDVI) that mediate the $T_{skin}$--$T_{2m}$ gradient.

Several gridded products, like NLDAS and MODIS \citep{Wan01092008}, provide high-resolution $T_{skin}$ data. NLDAS captures thermal signatures (asphalt vs.\ vegetation, irrigated vs.\ dry surfaces) that drive local air temperature variations. However, the NLDAS resolves meteorological fields on the order of $\sim$12 km, which is still not ideal for examining sub-grid variability relevant for localized heat extremes (i.e., $<12$ km). On the other hand, thermal-infrared satellite $T_{skin}$ products such as MODIS offer much finer spatial detail (e.g., $\sim$1 km), but are temporally sparse because they only observe a given location at specific overpass times \citep{wan2008new}. These limitations motivate the development of statistical models that can deliver physically sensible and high-resolution $T_{skin}$ predictions at unobserved locations, which can then inform point-referenced $T_{2m}$ estimation (instead of areal averages) and localized extreme-event analysis. 

\subsection{Study objectives and scientific questions}\label{s:motivation:impact}
Our primary downstream use of spatial predictions of $T_{skin}$ is improved estimation of $T_{2m}$ in topographically complex, semi-arid/arid, or poorly instrumented regions. The two temperature variables are complementary in their contribution of valuable information to the study of climate change \citep{jin2010land}. Converting $T_{skin}$ to $T_{2m}$ using a locally calibrated model that incorporates auxiliary covariates (e.g., NDVI, wind, soil moisture, and time of day/season) can outperform direct interpolation of $T_{2m}$ because $T_{skin}$ better captures fine-scale thermal heterogeneity and microclimatic gradients \citep{oyler2016remotely,stoyanova2019spatial}. We also validate competing $T_{2m}$ interpolation approaches against independent observations from the Global Historical Climatology Network \citep[GHCN;][]{menne2012overview}, which serve as sub-grid evaluation sites (i.e., away from reanalysis grid points).

$T_{skin}$ can help quantify heat stress, a major public health concern. Using $T_{2m}$ estimated from $T_{skin}$ and coincident 2-m dew point temperature, we compute the heat index \citep{rothfusz1990heat} using the \texttt{weathermetrics} \texttt{R} package. The heat index reflects the body’s ability to cool via sweat evaporation and is used to issue operational heat advisories \citep{noaa_heat_safety}. It also provides a consistent measure of heat extremes across space, since identical temperatures can correspond to different levels of heat stress depending on humidity \citep{steadman1979assessment,rothfusz1990heat}.

\textbf{Scientific questions.} (1) What are the spatial and seasonal patterns of extreme surface skin temperature in the Four Corners region, and what do they imply about extremal dependence? (2) Can we accurately predict extreme temperatures and downstream heat risk at sub-grid locations, and where do elevated heat risk and socioeconomic vulnerability overlap? (3) Which scalable inference approach best answers these questions while remaining computationally feasible?

\subsection{Data acquisition and data structure}\label{s:motivation:data}
The NLDAS Phase~2 dataset \citep{xia2012continental} provides hourly, gridded $T_{skin}$ values on a $0.125^{\circ}$ resolution grid ($\sim$ 12--14 km). We analyze a 46-year record (1979--2024) of $T_{skin}$ over the Four Corners region ($106.4^\circ$--$111.2^\circ\mathrm{W}$, $34.44^\circ$--$38.19^\circ\mathrm{N}$;  Figure~\ref{fig:seasonal_avg_max}) at the intersection of the US states of Colorado, New Mexico, Arizona, and Utah. This study region comprises diverse terrain ranging from desert plains to mountain ranges that produce hot, dry summers and cold, snowy winters. The study region spans approximately $430.6$ km and $417.5$ km in the longitudinal and latitudinal directions respectively, with a maximum diagonal distance of about $600$ km across $n=1{,}209$ grid cells. For our analyses, we hold out 130 grid cells for validation; Figure~\ref{fig:gev_params} shows the spatial distribution of these holdout grid points.

We partition the records by season as temperature extremes differ substantially across the year, both in their statistical properties (e.g., marginal distributions and spatial extent) and their downstream environmental impacts. For instance, winter heat directly threatens water security and drives snow drought as it causes precipitation to fall as rain rather than snow and accelerate the melt of mountain snowpack \citep{siirila2021low}. Summer extremes in arid regions exhibit heavier upper-tails, associated with extended periods of high temperatures that drive heat-related mortality \citep{NOAA2023HazStat}, increase energy demand \citep{NCA5Ch5}, and reduce agricultural yields \citep{zhao2017temperature}. To illustrate seasonal differences, Figure~\ref{fig:seasonal_avg_max} displays the spatially-averaged seasonal maxima of $T_{skin}$ over the time record, and Figure~\ref{fig:raw_vs_standardized_maps} presents maps of the seasonal block maxima of $T_{skin}$ on the original scale and on a marginally transformed (uniform) scale. The latter is used as a preprocessing step to remove seasonality and to help to visualize spatial extremal dependence (see Section~\ref{s:exploratory} for details). We observe marked seasonal differences in the marginal extremal behavior and spatial extremal dependence. Thus, we opt to separately study the four seasonal records: DJF (December--February), MAM (March--May), JJA (June--August), and SON (September--November). 

\begin{figure}[!t]
    \centering
    \includegraphics[height=0.295\textwidth]{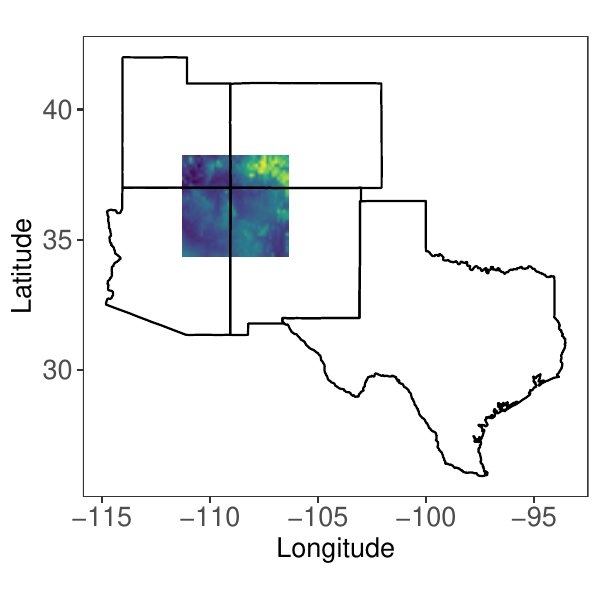}
    \hspace*{-0.1cm}\includegraphics[height=0.295\textwidth]{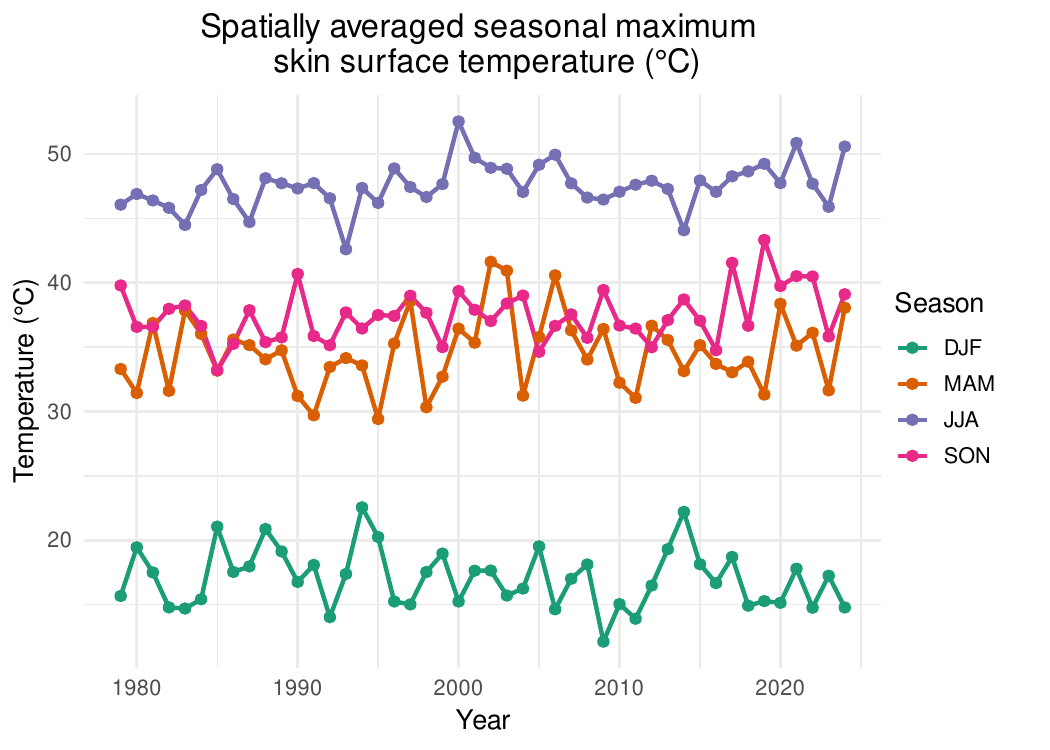}
    \hspace*{-0.3cm}\includegraphics[height=0.295\textwidth]{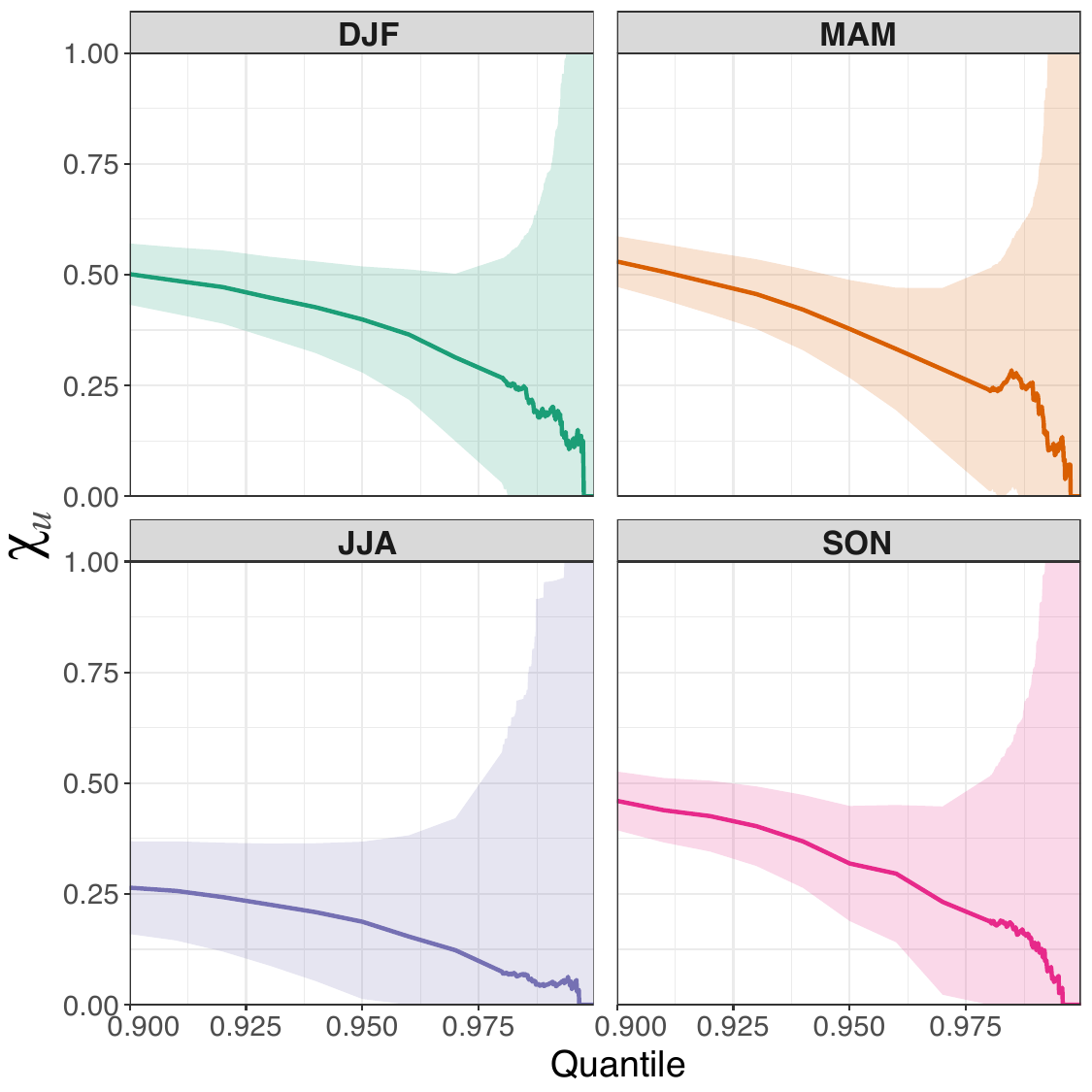}
    \caption{{\bf Left:} Location of the study region within the broader south-central and southwestern United States. The colored raster cells mark the NLDAS grid points over the Four Corners region and their elevations. {\bf Middle:} Spatially averaged seasonal maximum surface skin temperature ($T_{skin}$) over 1979--2024. For each year and season (DJF, MAM, JJA, SON), 
    seasonal block maxima are computed at each grid location and then averaged across all spatial locations. {\bf Right:} Season-wise empirical estimates of the extremal dependence function $\chi_u(h)$ (see Equation~\eqref{eqn:chi_definition}), at spatial lag $h=0.177$ ($\sim 75$ km). Point-wise $95\%$ confidence bands are provided as shaded regions.}
    \label{fig:seasonal_avg_max}
\end{figure}

\begin{figure}[!t]
    \centering
    \includegraphics[width=\linewidth]{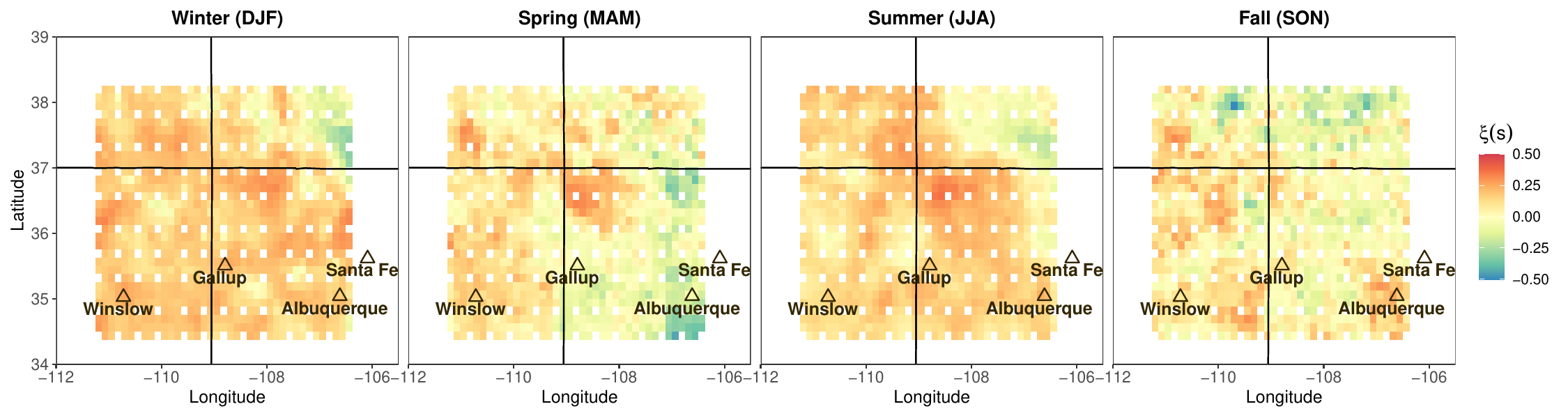}
    \caption{Seasonal maps of the maximum \textit{a posteriori} (MAP) estimate of the marginal GEV shape field $\{\xi(\bs)\}$ over the study region for each season. The white points denote the held-out NLDAS grid cells excluded from model fitting. The four triangle markers identify the GHCN stations---Winslow, Gallup, Albuquerque, and Santa Fe---used for subsequent station-based validation. State boundaries are overlaid for geographic reference.}
    \label{fig:gev_params}
\end{figure}
 
\subsection{Exploratory data analysis}\label{s:exploratory}
\vspace*{-0cm}\paragraph{Marginal modeling:} For each spatial location,  seasonal maxima are extracted to form approximately stationary series of $T_{skin}$ data suitable for extreme value analysis. The natural candidate for modeling block (seasonal) maxima time series is the generalized extreme value (GEV) distribution, having distribution function
\begin{align*}
    G(z;\mu,\sigma,\xi) = \begin{cases}
        \exp\left[-\left\{1+\xi\left(\frac{z-\mu}{\sigma}\right)\right\}_+^{-1/\xi}\right], &\xi\neq 0,\\
        \exp\left[-\exp\left\{-\left(\frac{z-\mu}{\sigma}\right)\right\}\right], &\xi=0,
    \end{cases}
\end{align*}
for $x_+=\max(0,x)$, $\mu,\, \xi\in\mathbb{R},$ and $\sigma>0$. To account for measurement error in the observations, we model season-wise maxima via a latent spatial GEV model \citep[see, e.g.,][]{cooley2007bayesian, risser2019probabilistic}, which is augmented with a Gaussian nugget term to mitigate small-sample noise. A complete description of the spatial GEV model and priors is provided in Section~\ref{sm:NLDAS_process} of the supplementary material. 

In short, for each spatial location $\bs$, the GEV parameter fields $\mu(\bs)$, $\sigma(\bs)$, and $\xi(\bs)$ are each modeled as a Mat\'ern Gaussian process (see Section~\ref{s:preliminaries} of the supplementary material) with elevation included as a covariate in the mean function. For inference, the maximum \textit{a posteriori} (MAP) estimate is obtained (see Figures~\ref{fig:gev_params} and~\ref{fig:gev_params2}), bypassing costly MCMC-based posterior sampling, and then used to transform the observed seasonal monthly maxima to a uniform scale via the probability integral transform. Figure~\ref{fig:gev_params} displays estimates of spatially varying $\xi$ values, indicating both heavy right and left tails depending on season and location. The corresponding estimates of the location and scale fields, $\{\mu(\bs)\}$ and $\{\sigma(\bs)\}$, are shown in Figure~\ref{fig:gev_params2} of the supplementary material, along with Anderson--Darling goodness-of-fit $p$-values for the fitted GEV distribution at each grid cell, providing evidence of widespread good fit across the domain.

The data on uniform scale serve as pseudo-observations for a subsequent analysis of the spatial extremal dependence. Two-step inference remains standard practice for spatial extremes \citep[see][for an overview]{huser2022}; adopting this framework provides a practical and methodologically clean way to isolate and evaluate dependence modeling performance, while leaving full joint Bayesian modeling to future work. We also note that, within the two-step paradigm, alternative approaches to marginal GEV modeling are available, including generalized additive models \citep{youngman2022evgam} and machine learning \citep{Richards2024}.

\begin{figure}[!t]
    \centering
    \includegraphics[width=\textwidth]{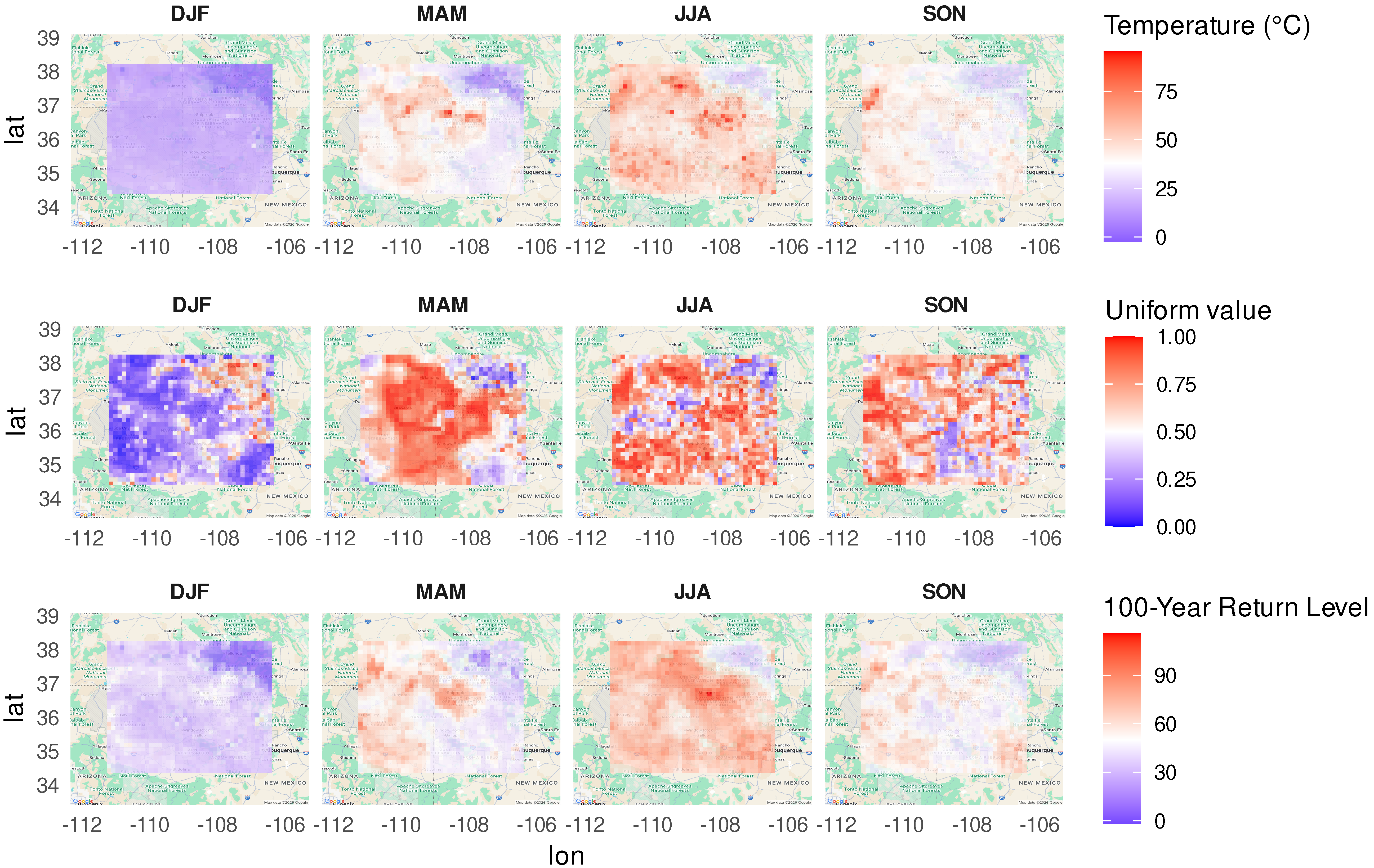}
    \caption{Spatial maps of seasonal maxima for the Four Corners region in 2024.
    \textbf{Top row:} Raw seasonal $T_{skin}$ from the NLDAS dataset, shown for DJF, MAM, JJA, and SON.
        \textbf{Middle row:} Uniform-scale seasonal maxima obtained by transforming the raw maxima using the fitted spatial GEV marginal distributions.
    \textbf{Bottom row:} 100-year return level ($^\circ C$) for $T_{skin}$ based on the fitted spatial GEV marginal distributions. }   \label{fig:raw_vs_standardized_maps}
\end{figure}

\vspace*{-0.3cm}\paragraph{Classifying extremal dependence:} When studying spatial extremes, various metrics can be used to formally classify the extremal dependence type, that is, the aforementioned notions of \emph{asymptotic dependence} and \emph{asymptotic independence}. This distinction is critical, as it governs whether extremes co-occur across locations or remain spatially localized, with direct implications for regional risk assessment and resource allocation. Letting $\{U(\bs):\bs\in\mathcal{S}\}$ denote a spatial process on standard uniform margins, for a continuous spatial domain $\mathcal{S}\subset\mathbb R^2$, the pairwise upper-tail index \citep{joe1997multivariate} for the pair of locations $(\bs,\bs')$ is defined as $\chi\left(\bs,\bs'\right)=\lim_{u\rightarrow 1}\chi_u\left(\bs,\bs'\right)\in[0,1]$, with
\begin{align}
    \chi_u\left(\bs,\bs'\right) = \frac{\Pr\left\{U(\bs)>u,U(\bs')>u\right\}}{1-u}, \qquad u\in(0,1).
\label{eqn:chi_definition}
\end{align}
The complementary coefficient of tail dependence $\eta\left(\bs,\bs'\right)\in(0,1]$ \citep{Ledford1996} is defined via the equation
\begin{align}
\Pr\left\{U(\bs)>u,U(\bs')>u\right\} = \mathcal{L}(1-u)(1-u)^{1/\eta\left(\bs,\bs'\right)},\qquad \text{for $u\rightarrow 1$},   
\label{eqn:eta_definition}    
\end{align}
where $\mathcal{L}$ is a slowly varying function at zero. When the process $\{U(\bs):\bs\in\mathcal{S}\}$ is stationary and isotropic, the coefficients $\chi(\bs,\bs')$ and $\eta(\bs,\bs')$ are functions of pairwise distance $h=\|\bs-\bs'\|$ only, for which we write $\chi(h)$ and $\eta(h)$.

When $\chi\left(\bs,\bs'\right)>0$ and $\eta\left(\bs,\bs'\right)=1$, the spatial process $U(\cdot)$ exhibits asymptotic dependence, which in practice means that the most extreme events can occur simultaneously at locations $\bs$ and $\bs'$. When $\chi\left(\bs,\bs'\right)=0$, we are in the setting of asymptotic independence, and the value of $\eta\left(\bs,\bs'\right)$ quantifies the corresponding strength of ``residual extremal dependence''. In general, $\eta\left(\bs,\bs'\right)<1$ always signifies asymptotic independence, while the boundary case of $\eta\left(\bs,\bs'\right)=1$ is more subtle; see, e.g.,  \citet{huser2022} for further detail. Sub-asymptotic behavior describes the dependence structure observed at high but non-asymptotic quantiles, that is, for $u$ close (but not equal) to 1.

\vspace*{-0.3cm}\paragraph{Diagnostics for max-stability:} Spatial models that assume asymptotic dependence (e.g. max-stable processes) require the tail dependence coefficient \(\chi_u(h)\) to converge linearly to a positive limit as \(u \nearrow 1\) . This assumption is restrictive, as it precludes asymptotic independence and forces a fixed extremal dependence class regardless of the data \citep[see, e.g.,][]{de1984spectral,ribatet2013spatial}. To evaluate whether max-stability is appropriate for the NLDAS data, we examine empirical estimates of \(\chi_u(h)\) and apply the bootstrap max-stability test of \citet{koh2024space}.

The right panel of Figure~\ref{fig:seasonal_avg_max} shows that, at a fixed distance \(h=0.177\), the seasonal \(\chi_u(h)\) curves depart from the linear stabilizing pattern expected under max-stability as \(u \nearrow 1\). On the unit-square scale, this lag corresponds to approximately $75\,\text{km}$ in the original domain. This suggests that max-stable models are too rigid for our data and motivates the use of more flexible random scale mixture models that can capture both extremal dependence classes. Results from the bootstrap max-stability test are provided in Section~\ref{sec:maxstab-tests} of the supplementary material, showing that max-stability is strongly rejected over the full Four Corners domain and remains widely rejected in smaller subregions for MAM, JJA, and SON. In contrast, DJF shows more frequent local non-rejection, suggesting that max-stable behavior may be plausible only in some subregions during winter, a finding consistent with the strong asymptotic dependence recovered by our proposed model in Section~\ref{s:application}.

\section{Modeling Framework for Spatial Extremes}\label{s:modeling_framework}

In this section, we review random scale mixture processes and introduce our novel random scale mixture model, which provides a flexible representation of extremal dependence and permits computationally efficient Bayesian inference.

\subsection{Random scale mixture processes}\label{s:rsm} 
Max-stable processes can be seen as a product mixture of a stochastic process $\{W(\cdot)\}$, with light upper-tailed marginals, and a random scaling $R$ with a heavy upper-tail (see Equation~\eqref{eq_MSP} of Section~\ref{s:preliminaries} of the supplementary material). The random scale having heavier upper-tails than the margins of $\{W(\cdot)\}$ induces asymptotic dependence in the resulting mixture $\{X(\cdot)\}$. Similar constructions, with modifications of $R$ and/or $\{W(\cdot)\}$, can result in a process that exhibits asymptotic independence \citep{engelke2019extremal}. This motivates the use of parametric models that allow for interpolation between the two extremal dependence classes, circumventing the need to specify an extremal dependence class prior to modeling. Examples include Gaussian scale \citep{huser2017bridging} and location \citep{krupskii2018factor} mixtures, but here we focus on the class of random scale mixture processes of \cite{huser2019modeling}, which allow for a transition between asymptotic dependence and asymptotic independence within the interior of the parameter space.

Specifically, \citet{huser2019modeling} define a random scale mixture process 
\begin{equation}
\label{eq:RSM}
\{X(\bs):\bs\in\mathcal{S}\}=R^\delta \{W(\bs)^{1-\delta}:\bs\in\mathcal{S}\},
\end{equation}
for $\delta \in [0,1]$, and where $R \geq 1$ is a unit Pareto random variable and $\{W(\bs):\bs \in \mathcal{S}\}$ is an asymptotically independent process with unit Pareto margins, i.e., for all $\bs$, $\Pr\{W(\bs) > w\} = w^{-1}$ for $w\geq 1$. The extremal dependence class of $\{X(\bs)\}$ is determined by the parameter $\delta$: if $\delta < 1/2$, then $\{X(\bs)\}$ exhibits asymptotic independence; and if $\delta \geq 1/2$, $R^\delta$ is heavier-tailed than the margins of $\{W(\bs)^{1-\delta}\}$, and so $\{X(\bs)\}$ is asymptotically dependent. In applications of model \eqref{eq:RSM}, $\{W(\bs)\}$ is typically taken to be a Gaussian process (GP), and we also adopt this approach.

\subsection{L\'evy random scale mixture model}\label{subsec:lrsm}
We tailor the process in \eqref{eq:RSM} to enable scalable Bayesian inference, motivated by the inferential and computational challenges in the heat extremes case study. Conditional on $R$, the process $\{X(\bs)\}$ in~\eqref{eq:RSM} is a GP (up to marginal standardization). This suggests
that one could exploit techniques from the classical geostatistics literature (based on GPs) to perform scalable Bayesian inference on spatial extremal dependence. However, both components of the random scale mixture process of \citet{huser2019modeling} have marginal support on $(1,\infty)$. As highlighted by \citet{zhang-2022a}, this support constraint leads to computational instability when updating $R$ via Markov chain Monte Carlo (MCMC), because large proposals for $R$ may produce $X(\bs)/R^{\alpha} < 1$ for some $\bs \in \mathcal{S}$, thereby violating the model support and yielding near-zero acceptance probabilities. Allowing component-wise support on $(0,\infty)$ resolves this support mismatch and substantially improves mixing. 

Our proposed model, which we term the L\'evy Random Scale Mixture (LRSM), specifies the data-level spatial process as
\begin{align}\label{eqn:lrsm}
    \{X (\bs):\bs\in\mathcal{S}\}&:= R^{\alpha}\left\{g(Z(\bs)):\bs\in\mathcal{S}\right\},
\end{align}
where $\alpha \geq 0$ and $g(\cdot):= \{1-\Phi(\cdot)\}^{-1}-1$. The L\'evy random scale $R\sim \text{L\'{e}vy}(m,c)$ has location and scale parameters $m\in \mbR$ and $c>0$, respectively, and the latent process $\{Z(\bs)\}$ is a Gaussian process with standard Gaussian marginals, i.e., $Z(\bs)\sim N(0,1)$ for all $\bs \in \mathcal{S}$.

Note that $R\in (0, \infty)$, and the function $g(\cdot) $ ensures that $\{g(Z(\bs))\}$ has shifted unit Pareto margins, i.e., for all $\bs$, $\Pr\{g(Z(\bs)) > z\}=(1+z)^{-1}$ for $z \geq 0$, and thus also has support on $(0,\infty)$. Despite this adaptation, the LRSM model \eqref{eqn:lrsm} has similar desirable extremal dependence properties to the model proposed by \citet{huser2019modeling}. The L\'evy distribution is regularly-varying with tail index $c$ \citep[Chapter 1.5,][]{nolan2020univariate}, yielding Pareto-type upper-tails. The component $g(Z(\bs))$ is also regularly-varying but with unit tail index. The parameter $\alpha >0$ controls the extremal dependence class of $\{X (\bs)\}$, inducing asymptotic independence when $\alpha\leq c$ and asymptotic dependence when $\alpha>c$. In Section~\ref{sm:chi_eta} of the supplementary material, we present the theoretical properties of pairwise tail dependence for process~\eqref{eqn:lrsm}.

\subsection{Modeling choices}\label{subsec:modelingChoices}
To use the LRSM framework~\eqref{eqn:lrsm} in practice, one must specify the distribution of $R$ (through the parameters $m$ and $c$) and the form of the Gaussian process $\{Z(\bs)\}$ (e.g., via its correlation function). The choices described below will be used in the remainder of the paper.

The L\'evy scale parameters, $m$ and $c$, do not directly impact the extremal dependence structure of the LRSM; instead it is the ratio $\alpha/c \in(0,\infty)$ that determines the strength and class of extremal dependence. Moreover, while $m$ and $c$ do appear in the marginal distribution of the LRSM (see Section~\ref{sm:chi_eta} of the supplementary material), their values are of no practical consequence. Similarly to \citet{huser2019modeling}, we use the LRSM  to define an extremal dependence copula (see Section~\ref{sec:bayes_inference}) and perform a separate marginal transformation of the data. In our application, without loss of generality, we fix the location and scale parameters of the L\'evy random variable to $m=0$ and $c=1/2$. The interpretation of $\alpha$ in the LRSM thus shares the same appealing interpretation as $\delta$ in \eqref{eq:RSM}, with $\alpha < 1/2$ implying asymptotic independence and $\alpha \geq 1/2$ inducing asymptotic dependence.

For the Gaussian process $\{Z(\bs)\}$, we use the popular Mat\'ern class with correlation function $C_{\boldsymbol{\theta}}:\mathcal{S}\times\mathcal{S}\rightarrow\mathbb{R}$ defined in Equation~\eqref{eq:MATERN} of the supplementary material. Here, the correlation function $C_{\boldsymbol{\theta}}$ is parameterized by a spatial range $\rho>0$ and smoothness $\nu>0$, i.e., with parameters $\boldsymbol{\theta}:=(\rho,\nu)$. The marginal variance of the Gaussian process is equal to one, so inference focuses solely on $\rho$ and $\nu$. 

We choose not to vary our model parameters across space, and the Mat\'ern correlation function depends only on the Euclidean distance between locations. These choices impose spatial stationarity and isotropy in our (extremal) dependence structure. However, the presented framework could be extended to account for nonstationarity and anisotropy, if required. For examples of handling such nonstationarity in spatial extremes see, e.g., \cite{shao2025flexible} and \cite{shi2024spatial}.

 \section{Scalable Inference for Spatial Extremal Dependence}\label{s:comp_strategies}
Likelihood-based inference for the LRSM model~\eqref{eqn:lrsm} is computationally demanding, involving repeated evaluation of high-dimensional Gaussian densities and nonlinear transformations requiring numerical integration. This section outlines several strategies we propose to make inference feasible at the scale of our application. These include Bayesian hierarchical models using scalable Gaussian likelihood approximations (Section~\ref{sec:bayes_inference}) and likelihood-free neural Bayes estimators that map data to parameter point estimates (Section~\ref{sec:nbe}).

\subsection{Bayesian inference}
\label{sec:bayes_inference}

The marginal distribution and density functions of $X(\cdot)$, denoted by $F_X(\cdot; \alpha)$ and $f_X(\cdot; \alpha)$, respectively, depend on the parameter $\alpha$ (see Section~\ref{sm:chi_eta} of the supplementary material), which is an unknown quantity to be estimated. Inference in our LRSM framework proceeds by considering the process model as a copula and working with the marginally standardized uniform random variables $\left\{U(\bs)=F_{X}(X(\bs);\alpha): \bs \in \mathcal{S}\right\}.$

For inferential purposes, we assume that we have $T$ temporal replicates observed at $n$ spatial locations, updating the notation in~\eqref{eqn:lrsm} to 
\[
X_t(\bs_i) = R_t^{\alpha} g(Z_t(\bs_i)), \qquad t = 1,\ldots,T,\qquad i = 1,\ldots,n,
\]
where $Z_t(\bs_i)$ are realizations from a Mat\'ern Gaussian process with correlation parameters $\boldsymbol{\theta} = (\rho,\nu)$, and the replicates are assumed independent across $t$, conditional on $(\alpha, R_t, \boldsymbol{\theta})$. The corresponding variables on uniform scale are denoted $U_t(\bs_i)$, $t=1,\dots,T$, $i=1,\dots,n$. We define the random vectors $\bU_t := (U_t(\bs_1),\ldots,U_t(\bs_n))^\top$ and $\bZ_t := (Z_t(\bs_1),\ldots,Z_t(\bs_n))^\top$, with corresponding stacked vectors $\bU$ and $\bZ$, and let $\bR := (R_1,\ldots,R_T)^\top$. Then the Bayesian hierarchical construction for the LRSM model can be written as
\begin{align*}
\mbox{Data Model:}\qquad & \{U_t(\bs):\bs \in \mathcal{S}\}_{t=1}^T|(\mathbf{R}, \alpha, {\boldsymbol{\theta}}) \sim F_{\bU}(\cdot;\mathbf{R},\alpha, {\boldsymbol{\theta}}),\\[-0.3em]
\mbox{Process Model:}\qquad & R_t\sim \mbox{L\'evy} (0,1/2),\qquad t=1,\dots,T,\\[-0.3em]
\mbox{Parameter Model:}\qquad & {\boldsymbol{\theta}} \sim p({\boldsymbol{\theta}}), \alpha\sim p(\alpha),
\end{align*}
\noindent where $F_{\mathbf{U}}(\cdot ;\mathbf{R},\alpha, {\boldsymbol{\theta}})$, the conditional distribution function of $\bU$ given $(\bR,\alpha,\btheta)$, is provided in Section~\ref{sm:likelihood} of the supplementary material, along with the functions to transform $U_t(\bs)$ to $Z_t(\bs)$ (and vice versa) given $\mathbf{R}$. We specify weakly informative independent priors for the dependence parameters, 
with $\alpha \sim \text{Uniform}(0,1)$ and $\rho \sim \text{Uniform}(0,0.5)$. Also see Section~\ref{sm:likelihood} for additional computational details. Figure~\ref{fig:flowchart} illustrates the modeling framework, linking the latent L\'evy random scale mixture process \(X_t(\bs)\) to the observations \(Y_t(\bs)\). 

\begin{figure}[!t]
    \centering
    \includegraphics[width=0.81\linewidth]{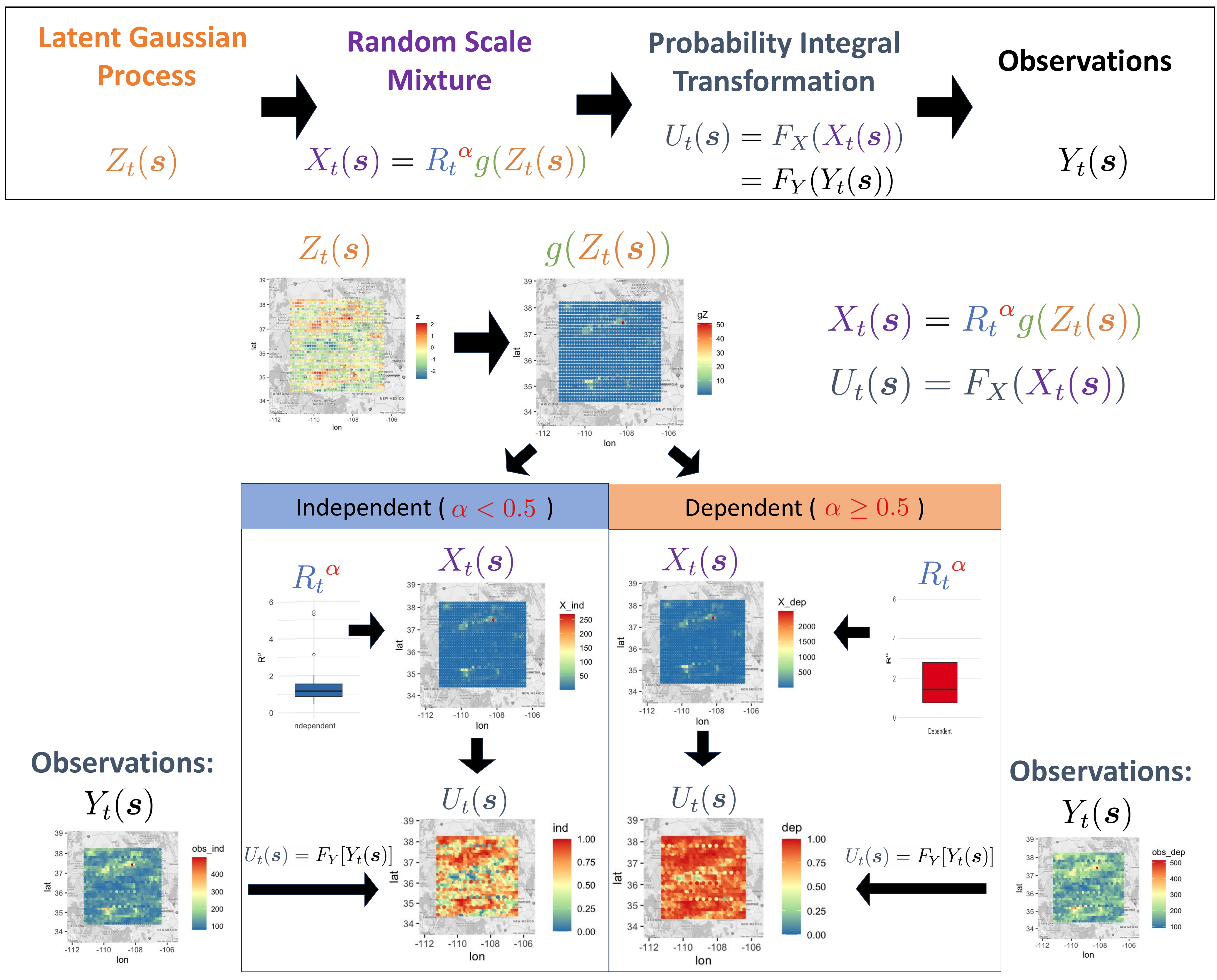}
    \vskip -0.1cm
    \caption{Schematic of the modeling framework linking the latent L\'evy random scale mixture process $X_t(\bs)$ to the observed data $Y_t(\bs)$ through a copula representation. A marginal spatially-varying GEV model is used to transform observations to a common scale, while the random scale mixture captures spatial extremal dependence. }
    \label{fig:flowchart}
\end{figure}

Repeated evaluation of the corresponding conditional density $f_{\bU}(\cdot; \bR, \alpha,\btheta)$ is computationally intensive, as it requires large matrix operations involving $\bC_{\boldsymbol{\theta}}$, which is the matrix containing evaluations of the correlation function $C_{\boldsymbol{\theta}}(\cdot,\cdot)$ for all pairs of sampling locations $(\bs_i,\bs_j),i,j=1,\dots,n$. Each Cholesky decomposition of $\bC_{\boldsymbol{\theta}}$ costs $\mathcal{O}(n^3)$ operations. The scalable approaches we consider target computational bottlenecks associated with $\bC_{\boldsymbol{\theta}}$. Vecchia approximations \citep{vecchia1988estimation, stein2004approximating} exploit sparse representations of the Cholesky factor of $\bC_{\boldsymbol{\theta}}^{-1}$ through conditional independence assumptions. Covariance tapering \citep{furrer2006covariance} induces sparsity directly in $\bC_{\boldsymbol{\theta}}$ by shrinking long-range covariances to zero while preserving positive definiteness. Low-rank methods approximate $\bC_{\boldsymbol{\theta}}$ using the outer product of $k \ll n$ basis functions. Additional details are provided in Section~\ref{sec:scalable} of the supplementary material.

\subsection{Neural Bayes estimators}
\label{sec:nbe}

 Neural Bayes estimators (NBEs) have shown promise as an efficient alternative to standard likelihood-based techniques for inference with models for marginal \citep{rai2024fast,richards2023modern} and spatial \citep{lenzi2023neural} extremes. They use pre-trained neural networks to efficiently map data to point estimates of posterior summary statistics. We follow \citet{sainsbury2022fast,sainsbury2023neural} and construct neural Bayes estimators for posterior means and posterior credible intervals, using graph neural networks \citep[GNNs;][]{Zhou_2020_GNN_review} to perform efficient inference with irregular spatial data. Once trained, the estimator can be queried for negligible computational cost and is amortized with respect to the number of samples, $T$, and the number $n$ and configuration of the sample locations. That is, we train the estimator only once and then reuse it for point estimation throughout the application and simulation study. Note that the NBEs cannot be used for out-of-sample prediction as they do not provide access to the full joint posterior and there is no closed form expression for the conditional distribution of the LRSM.

  Details on the implementation and training of the NBEs are deferred to Section~\ref{sec:NBE_discuss} of the supplementary material. For prior usage of NBEs for random scale mixture models, see 
\citet{richards2024neural}, \citet{dell2025flexible}, and \cite{andre2025neural}.

\section{Simulation Study and Performance Comparison}
\label{s:simulation_study}

We conduct three simulation studies to evaluate the inferential, predictive, and computational performance of the proposed LRSM inference methods. The first considers moderately sized datasets, focusing on parameter estimation and out-of-sample prediction.  The second examines scalability for larger datasets. Section~\ref{sec:ms_sim_study} of the supplementary material presents results under model misspecification using max-stable and asymptotically independent inverted max-stable processes \citep{wadsworth2012dependence}.

\subsection{Simulation design}
For all simulation studies, we generate $n$ irregularly spaced locations, $\bs_1,\dots,\bs_n,$  uniformly at random within the unit square, i.e., $\mathcal{S}= [0,1]^2$. At these sites, we generate $T$ independent replicates of the LRSM process \eqref{eqn:lrsm}. The true values of $\rho$, $\nu$, and $\alpha$, and the dimensions of the datasets ($n$ and $T$) vary across the simulation studies. The \emph{low-dimensional study} fixes $(n,T)=(500,100)$ and explores $\rho\in\{0.05,0.25\}$ for $\nu=0.5$, and $\rho\in\{0.0316,0.1581\}$ for $\nu=1.5$, chosen so that the corresponding effective ranges are 0.15 and 0.75. It also examines the extremal dependence class over a representative range of values $\alpha\in \{0.05, 0.3, 0.45, 0.55, 0.7\}$. By contrast, the \emph{high-dimensional study} fixes $(\rho,\nu)=(0.05,0.5)$ and varies $\alpha \in \{0.3, 0.7\}$,  $n \in \{1000, 5000\}$ and  $T \in \{10, 50\}$. Full details of the simulation design are provided in Section~\ref{sec:simulation_design_supplement} of the supplementary material.

For each parameter combination, we report results from various inference methods corresponding to repetitions over 100 experiments. Concretely, for each of the 100 simulated datasets, we randomly hold out 25\% locations for validation in the low-dimensional studies (15\% for high-dimensional studies), and compare five inference procedures: (i) MCMC using the full GP likelihood (i.e., the ``gold standard''); (ii) MCMC using the Vecchia approximation with conditioning sets of sizes $m=3,\,5,\,10,\,15$; (iii) MCMC using the low-rank basis approach with 100 and 200 spatial basis functions; (iv) MCMC using covariance tapering with a spherical taper yielding 90\% and 20\% covariance sparsity; and (v) the Neural Bayes estimators. All tuning was performed on the training data, and predictive metrics were computed on a held-out validation set. We report the empirical 95\% credible intervals for $\alpha$ and $\rho$, as well as corresponding interval scores \citep{gneiting2007strictly}. Predictive performance was assessed using three tail-weighted versions of the continuous ranked probability score (CRPS) with results reported in Sections~\ref{tab:lowdim_results} and~\ref{tab:highdim_results} of the supplement. Additional details on the MCMC-based inference procedures, the NBE approach, and the validation metrics are provided in Sections~\ref{sec:implementation} and~\ref{sec:validation_metrics} of the supplement.

\subsection{Results}\label{s:sim_results}
Across both simulation designs, methods that preserve local spatial dependence provide the most reliable inference and predictive performance.  In particular, full GP and higher-order Vecchia approximations are the most accurate, while the NBE approach offers substantial computational gains with modest losses in statistical efficiency but does not support out-of-sample prediction. Performance for most scalable methods deteriorates under smoother spatial fields and stronger dependence.

\vspace*{-0.3cm}\paragraph{Low-dimensional setting.}
Representative results are shown in Figure~\ref{fig:simulation-panels}, with full results reported in Tables~\ref{tab:sim-results-1}--\ref{tab:sim-results-20} and Figures~\ref{fig:lowdim-1}--\ref{fig:lowdim-5} of the supplementary material. The full GP achieved near-nominal coverage and the lowest interval scores, serving as a benchmark. Vecchia approximations performed similarly to the full GP, with accuracy improving as the conditioning size $m$ increases. Results were nearly indistinguishable from the full GP for moderate $m$, while smaller $m$ led to under-coverage, particularly under stronger long-range spatial dependence. Coverage was highest for rougher, short-range fields, and declined under smoother, longer-range spatial dependence, as well as larger $\alpha$. Predictive accuracy was comparable between both methods across all settings.

\begin{figure}[!t]
\centering
\begin{subfigure}{0.49\textwidth}
    \centering
    \includegraphics[width=\linewidth]{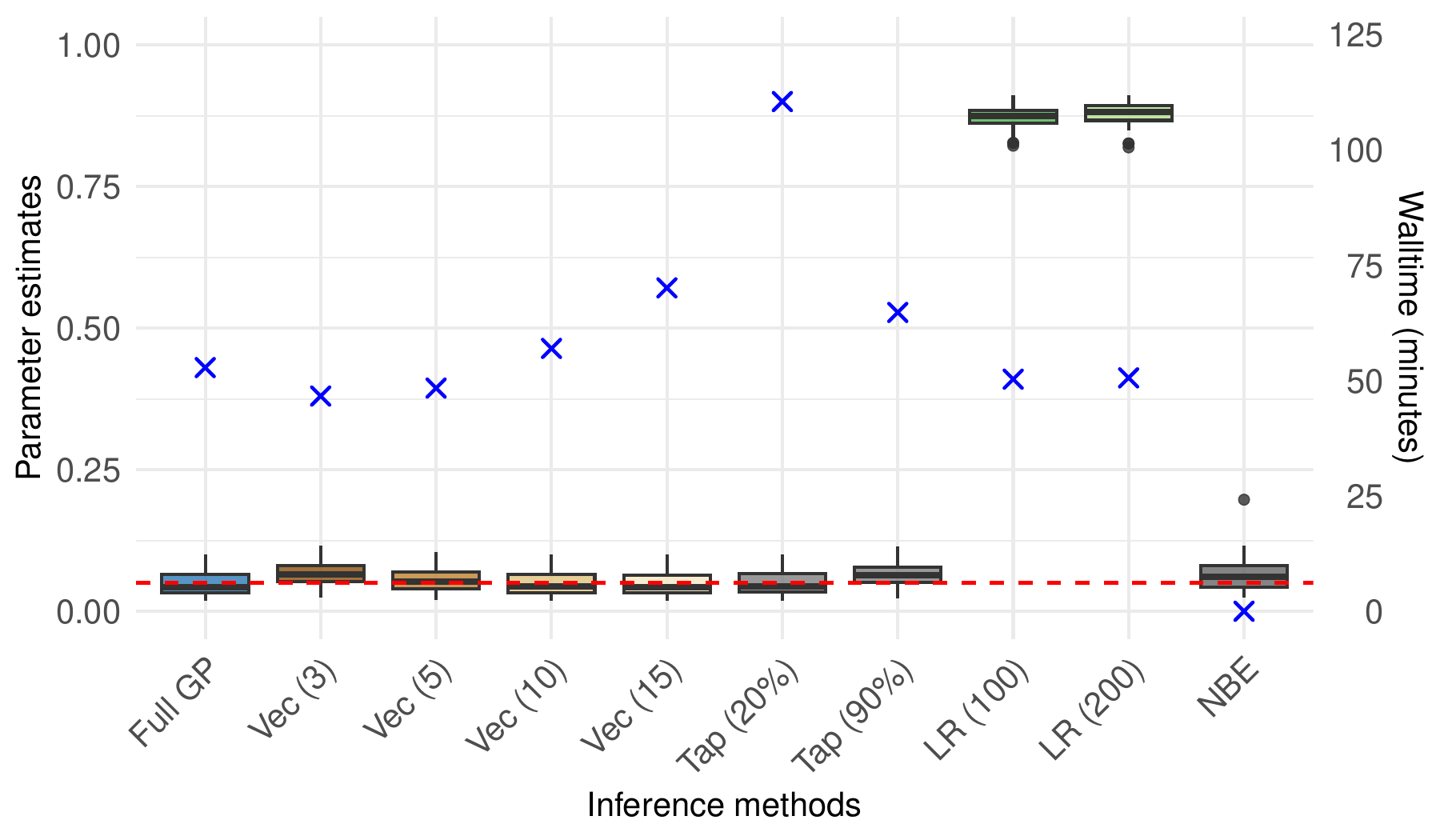}
    \caption{$\alpha=0.05$, $\nu=0.5$, $\rho=0.05$}
\end{subfigure}
\hfill
\begin{subfigure}{0.49\textwidth}
    \centering
    \includegraphics[width=\linewidth]{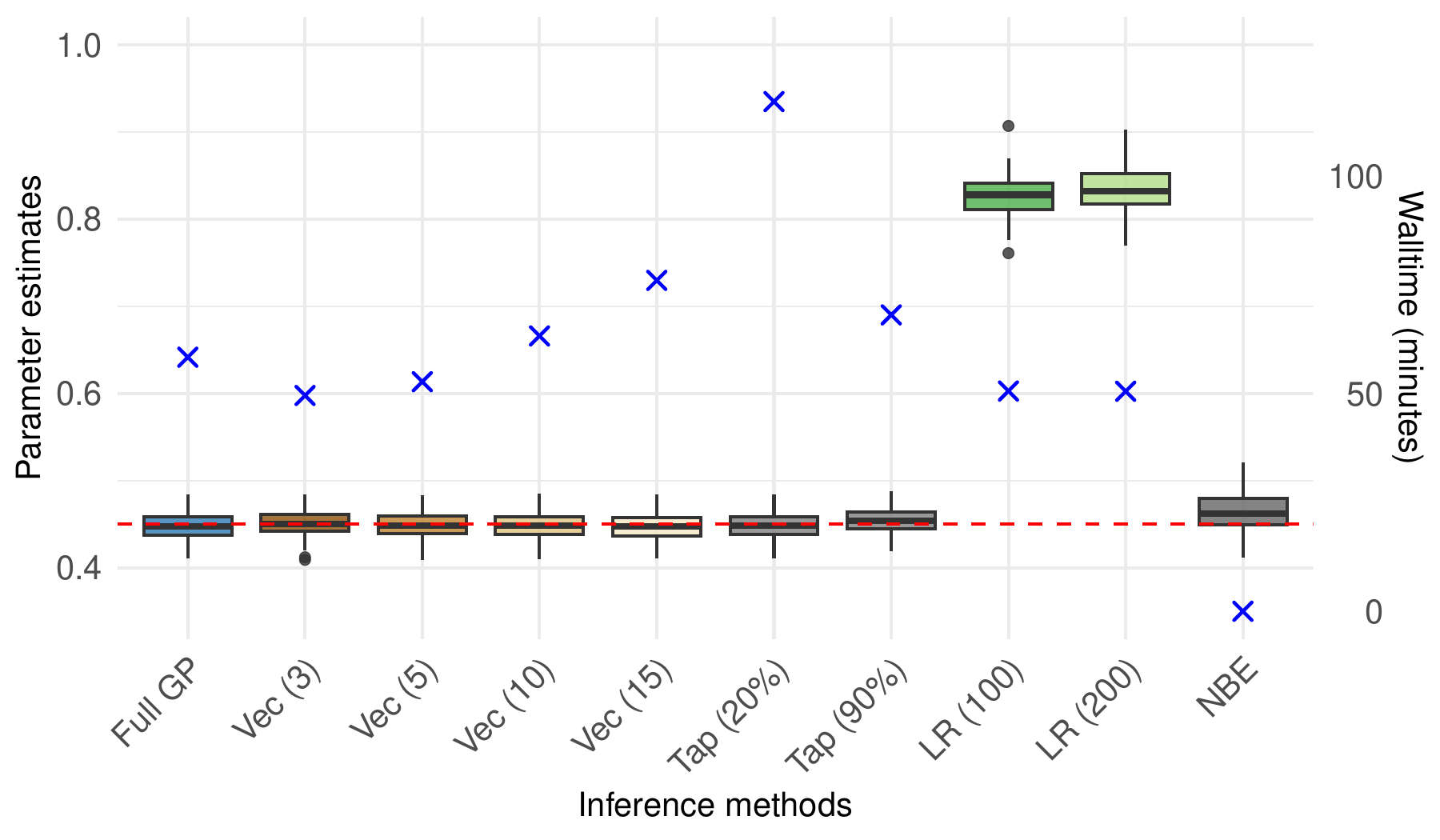}
    \caption{$\alpha=0.45$, $\nu=0.5$, $\rho=0.05$}
\end{subfigure}

\vspace{0.4cm}

\begin{subfigure}{0.49\textwidth}
    \centering
    \includegraphics[width=\linewidth]{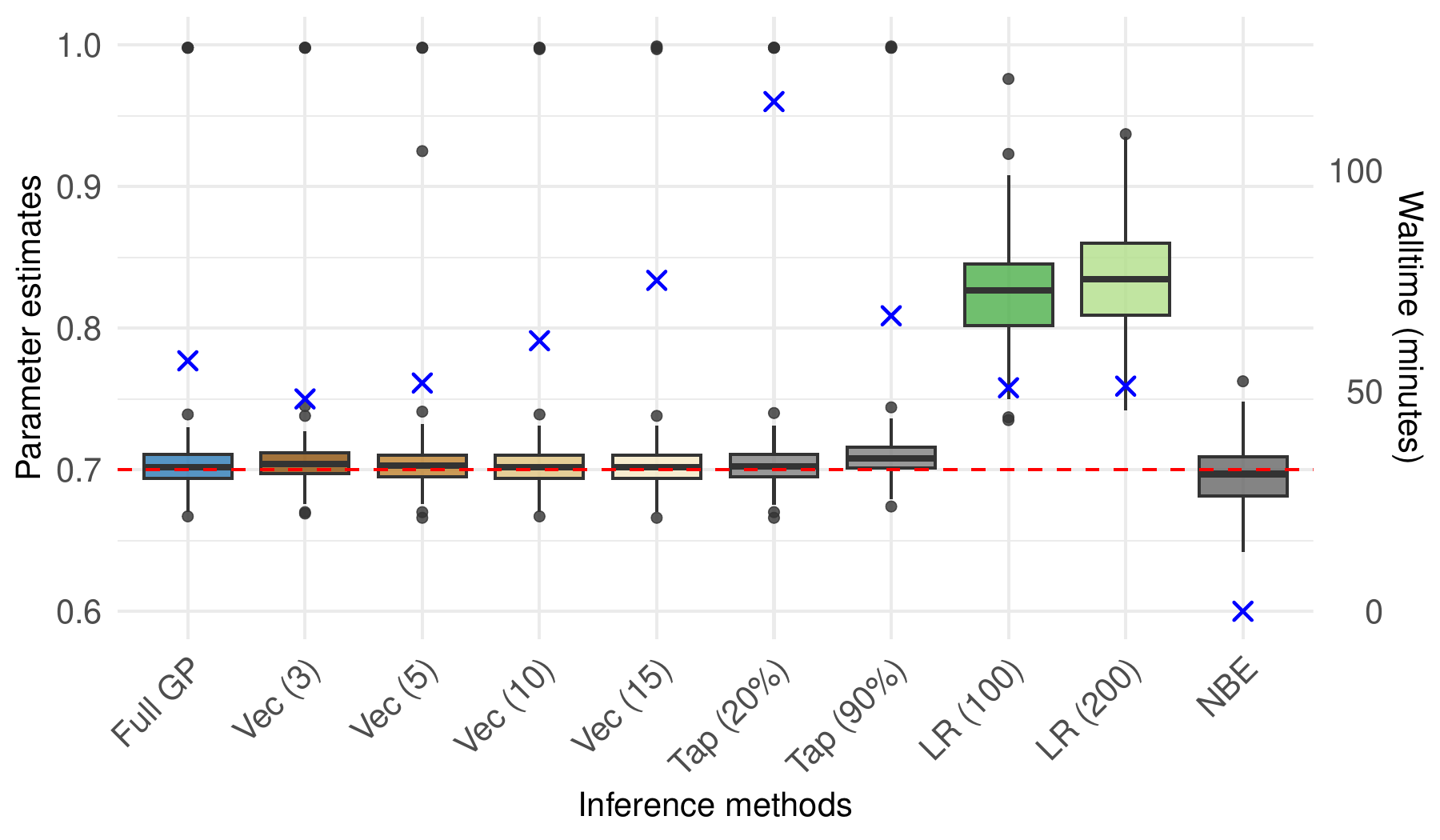}
    \caption{$\alpha=0.70$, $\nu=0.5$, $\rho=0.05$}
\end{subfigure}
\hfill
\begin{subfigure}{0.49\textwidth}
    \centering
    \includegraphics[width=\linewidth]{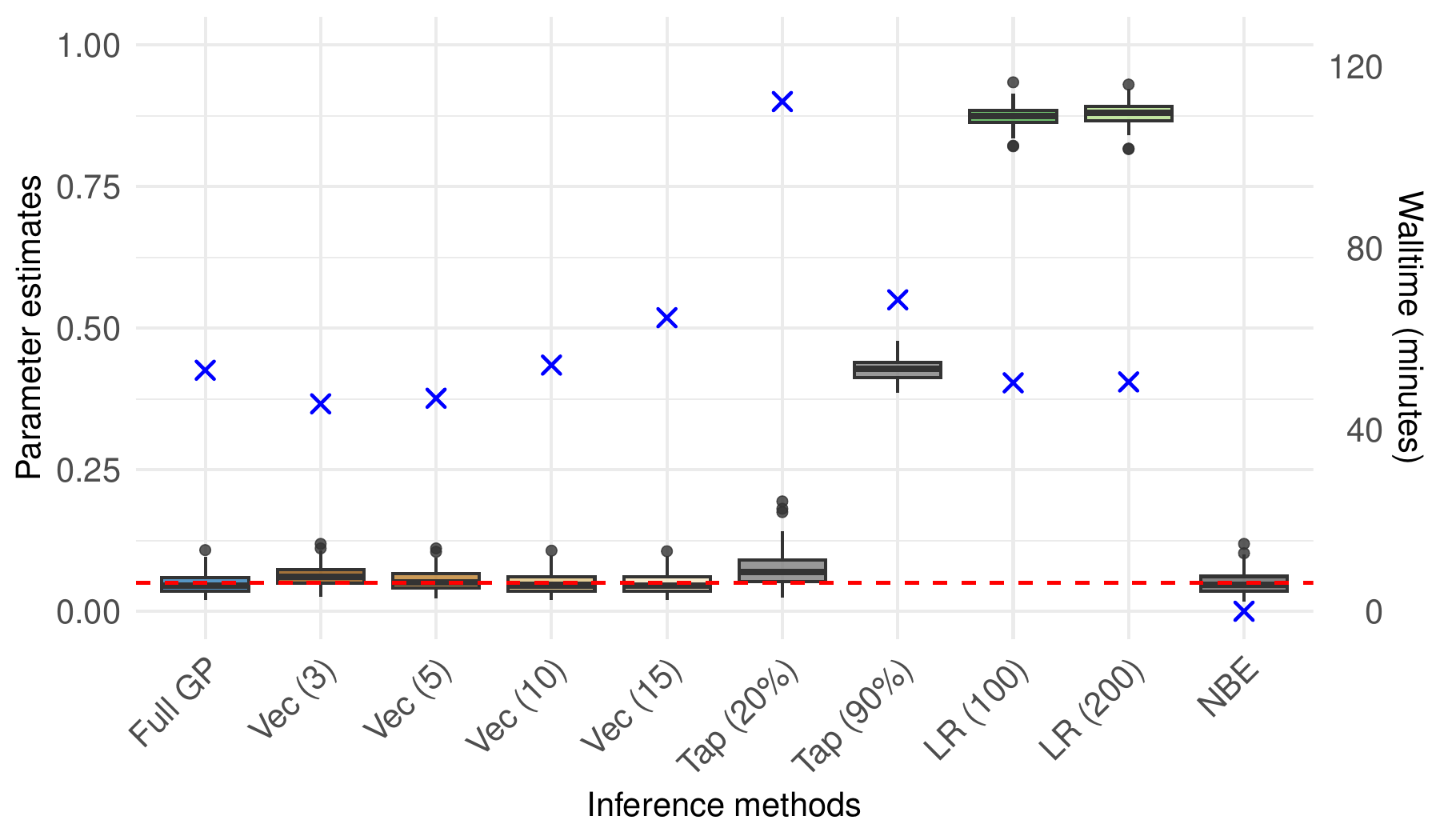}
    \caption{$\alpha=0.05$, $\nu=1.5$, $\rho=0.032$}
\end{subfigure}

\caption{
{\bf Low-dimensional study.} Comparison of parameter estimation accuracy and computational cost across inference methods under four representative simulation settings. Here, Vec$(m)$ denotes the Vecchia approximation with conditioning set size $m$, Tap$(p\%)$ denotes covariance tapering with sparsity level $p\%$, and LR denotes the low-rank method. Boxplots show posterior median estimates of the dependence parameter $\alpha$ obtained from 100 simulated datasets, with the true parameter value (dashed red line). Blue crosses denote the walltime (in minutes) required for each method. Panels represent: 
(a) weak dependence ($\alpha=0.05$, $\nu=0.5$, $\rho=0.05$); (b) moderate dependence ($\alpha=0.45$, $\nu=0.5$, $\rho=0.05$); 
(c) strong dependence ($\alpha=0.70$, $\nu=0.5$, $\rho=0.05$); and (d) a smooth latent field ($\alpha=0.05$, $\nu=1.5$, $\rho=0.032$).}

\label{fig:simulation-panels}
\end{figure}

Covariance tapering performs well at low sparsity levels but produces inaccurate inference as sparsity increases. This pattern likely reflects bias induced by truncating long-range extremal dependence, leading to misrepresentation of the underlying spatial correlation structure. Low-rank methods exhibit poor coverage and interval scores, indicating that global basis representations fail to capture local spatial structure. 

The NBE achieves near-nominal coverage with substantially reduced computational costs compared to the likelihood-based approaches, though at the cost of modestly higher interval scores. Practical considerations for NBE are a lengthy training stage and inability to make predictions at unobserved locations. Overall, methods that preserve local spatial extremal dependence provide the most reliable inference, with Vecchia approximations offering the strongest balance between accuracy and computational efficiency.

\vspace*{-0.3cm}\paragraph{High-dimensional setting.}
Figure~\ref{fig:highdim-panels} highlights representative high-dimensional scenarios, while complete results are reported in Tables~\ref{tab:sim-results-high1}--\ref{tab:sim-results-high8} and Figures~\ref{fig:highdim-1}--\ref{fig:highdim-2} of the supplementary material. Performance is broadly consistent with the low-dimensional setting. When computationally feasible ($n=1000$), the full GP provides strong inferential performance and serves as a benchmark. Vecchia approximations remain the most reliable scalable approach, achieving near-nominal coverage and interval scores comparable to the full GP while remaining computationally feasible at larger scales.

\begin{figure}[!t]
\centering
\begin{subfigure}{0.49\textwidth}
    \centering
    \includegraphics[width=\linewidth]{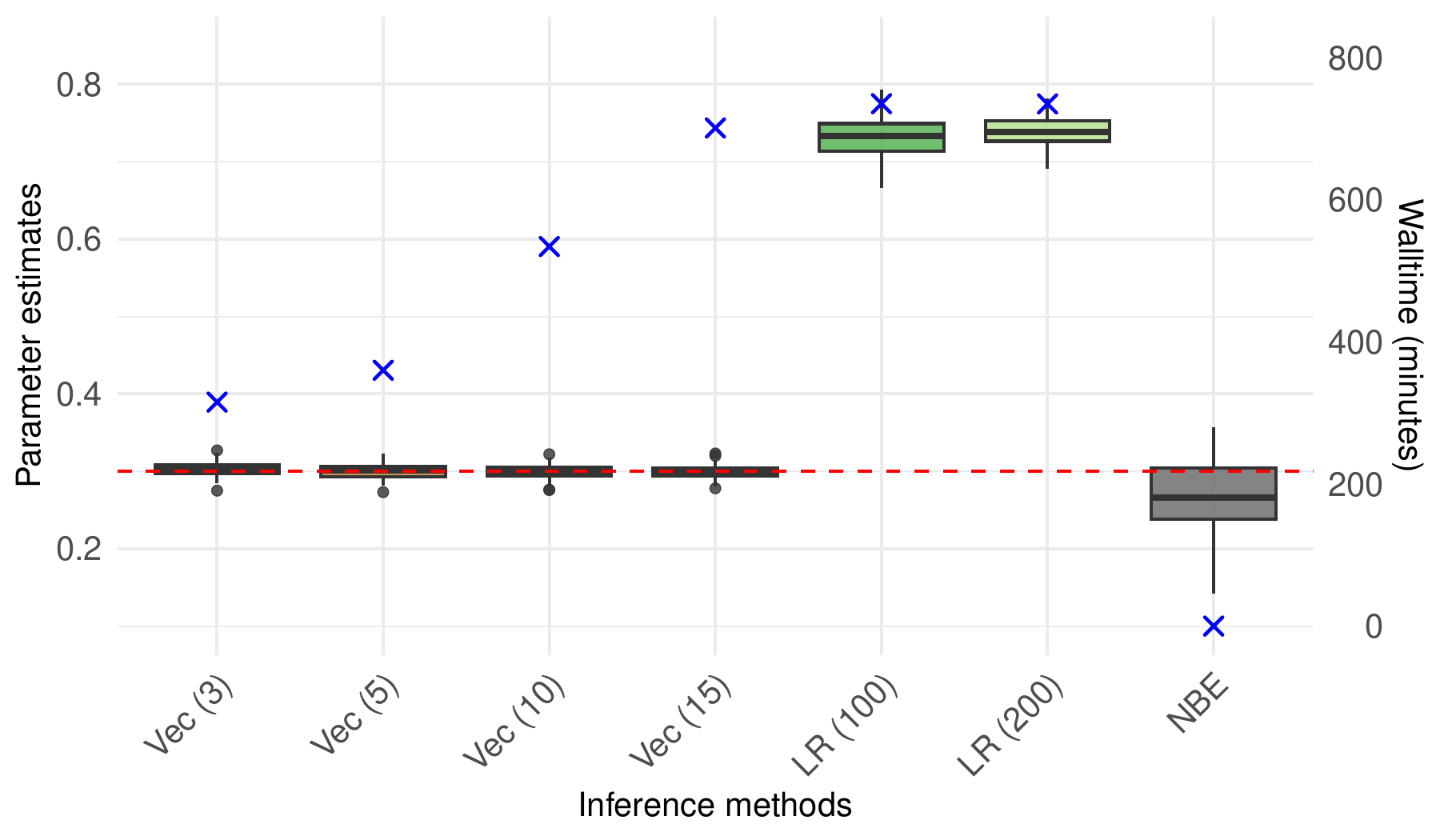}
    \caption{Scenario 1: $\alpha=0.3$, $n=5000$, $T=50$}
\end{subfigure}
\hfill
\begin{subfigure}{0.49\textwidth}
    \centering
    \includegraphics[width=\linewidth]{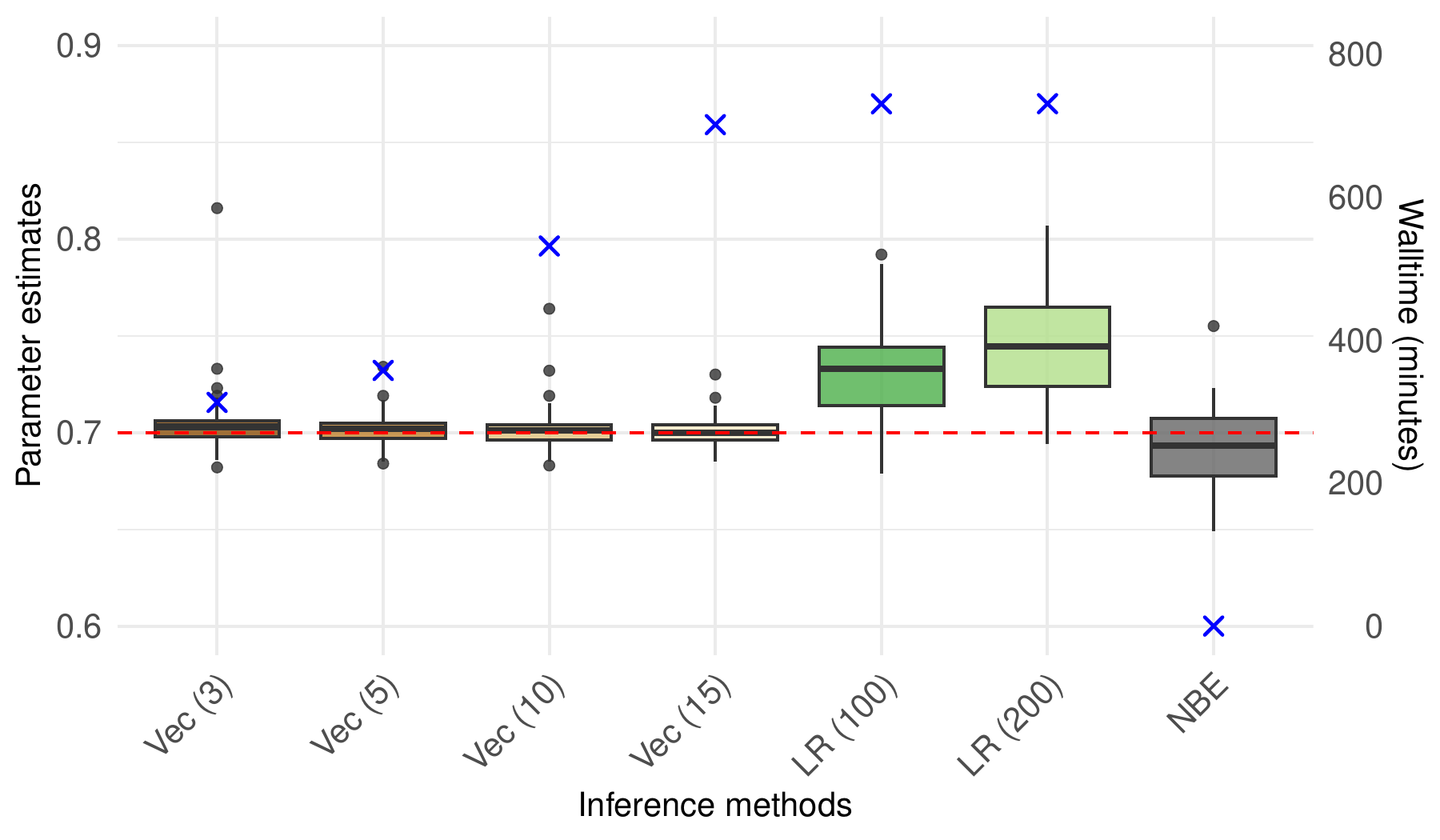}
    \caption{Scenario 2: $\alpha=0.7$, $n=5000$, $T=50$}
\end{subfigure}

\vspace{0.35cm}

\begin{subfigure}{0.49\textwidth}
    \centering
    \includegraphics[width=\linewidth]{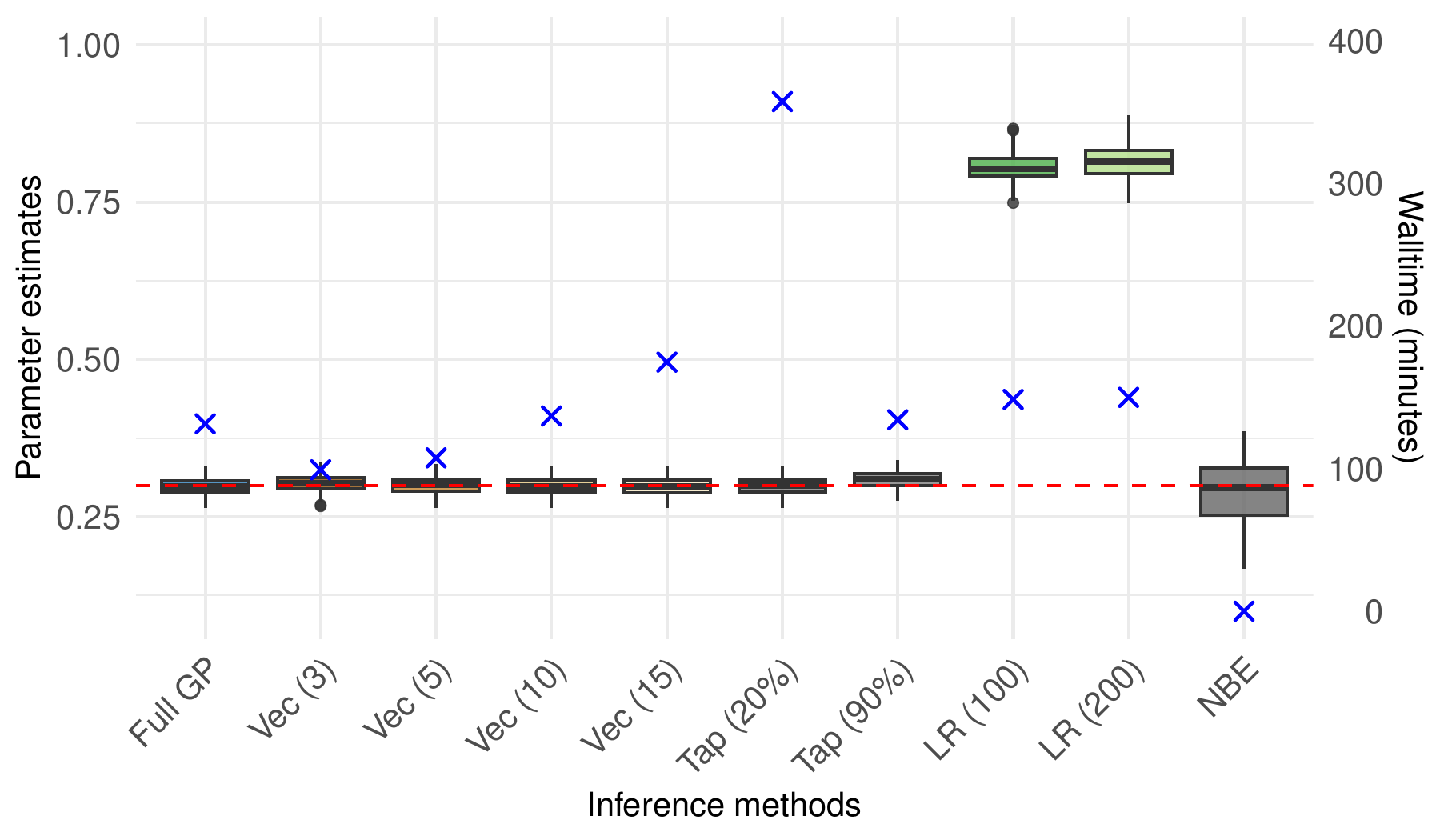}
    \caption{Scenario 3: $\alpha=0.3$, $n=1000$, $T=50$}
\end{subfigure}
\hfill
\begin{subfigure}{0.49\textwidth}
    \centering
    \includegraphics[width=\linewidth]{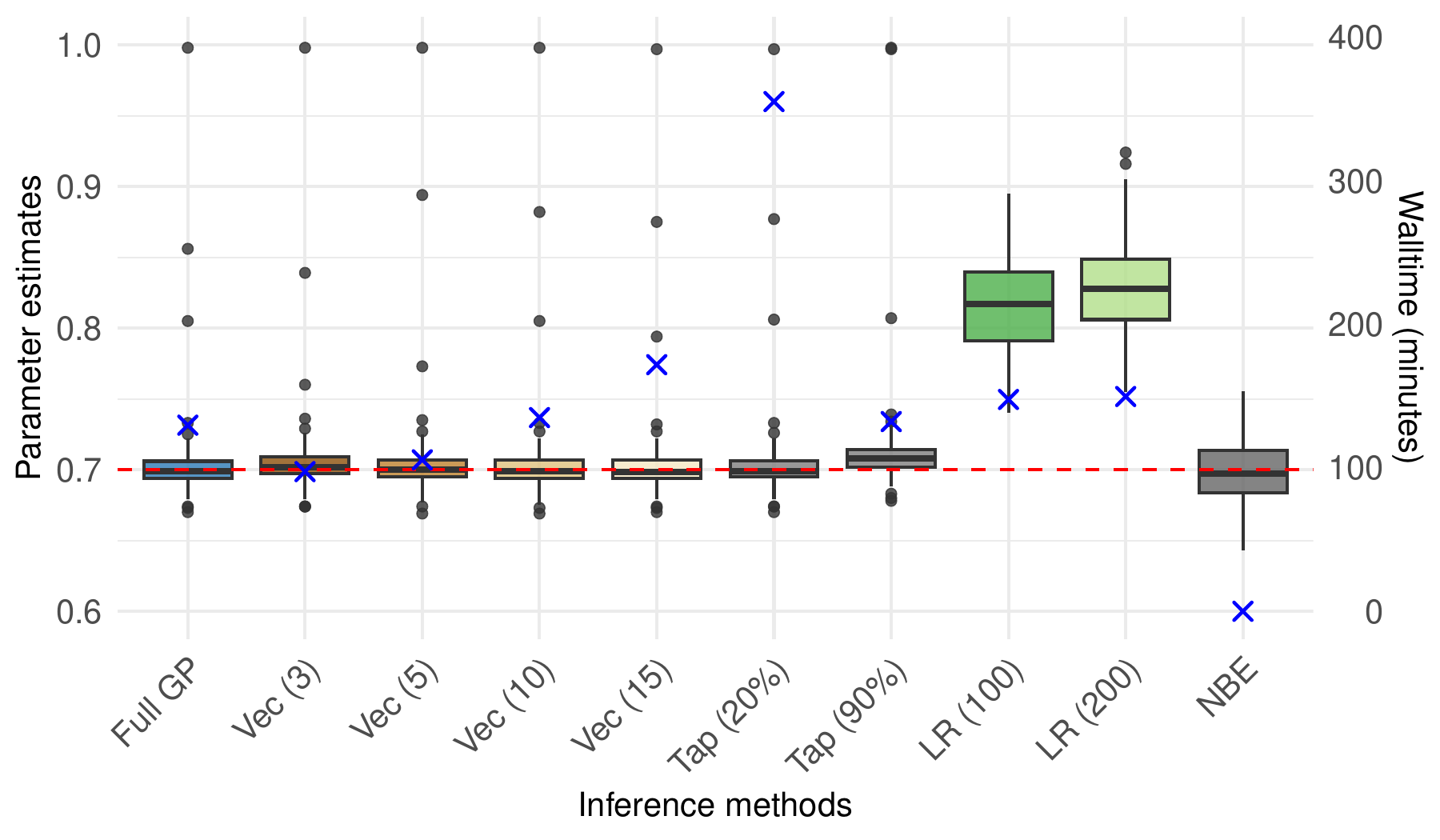}
    \caption{Scenario 4: $\alpha=0.7$, $n=1000$, $T=50$}
\end{subfigure}

\caption{
{\bf High-dimensional study.} Comparison of parameter estimation accuracy and computational cost across inference methods under four representative simulation settings. Boxplots show posterior median estimates of the dependence parameter $\alpha$ obtained from 100 simulated datasets, with the true parameter value (dashed red line). Blue crosses denote the walltime (in minutes) required for each method. Panels represent: (a) moderate dependence, large sample ($\alpha=0.3, n=5000, T=50$); (b) strong dependence, large sample ($\alpha=0.7, n=5000, T=50$); (c) moderate dependence, smaller sample ($\alpha=0.3, n=1000, T=50$); (d) strong dependence, smaller sample ($\alpha=0.7, n=1000, T=50$).
}
\label{fig:highdim-panels}
\end{figure}

Low-rank methods continue to exhibit poor coverage and high interval scores, despite moderate computational speedups. 
Covariance tapering yields competitive predictive performance but unreliable inference, particularly for the range parameter. Increasing sparsity (from $20\%$ to $90\%$) improves computational efficiency but comes at the cost of reduced inferential and predictive performance.

The NBE approach delivered the lowest walltimes by several orders of magnitude, making it the only method capable of near-instantaneous estimation at larger spatial and temporal scales. Coverage and interval scores were adequate, but exhibited more variability in inference than the likelihood-based approaches. Limitations regarding training and prediction mirror those observed in the low-dimensional setting.

Across scalable likelihood-based methods, predictive performance is broadly comparable once local spatial extremal dependence is adequately represented. Increasing temporal replication improves parameter recovery more than increasing spatial resolution. Overall, Vecchia provides the closest approximation to the benchmark, while NBEs offer a computationally efficient alternative for large-scale applications.

\section{Extreme Heat in the Four Corners Region}\label{s:application}
We center our analysis on heat index estimation in the Four Corners region (see Section~\ref{s:motivation}), using this case study to compare LRSM-based inference approaches. 
This analysis examines how modeling spatial extremal dependence influences interpolation of $T_{skin}$, propagation to sub-grid predictions of $T_{2m}$, and downstream heat index estimation.

\subsection{Fitting and interpolation of the surface skin temperature}
We model surface skin temperature at $n=1,209$ grid points over $T=46$ years (1979–2024) across four seasons, holding out 130 grid points for validation. The LRSM model is fitted on each seasonal dataset using smoothness parameters $\nu = 0.5$ and $\nu = 1.5$. Apart from the full Gaussian process likelihood, we consider scalable approximations: low-rank models (ranks 100 and 500), covariance tapering with a spherical taper (20\% and 90\% sparsity), Vecchia approximations with conditioning set sizes 3, 5, 10, and 20, and the NBEs.

Seasonal posterior credible intervals for $\alpha$ in Tables~\ref{tab:data_results_alpha1} ($\nu=0.5$), and~\ref{tab:data_results_alpha2} ($\nu=1.5$) of the supplementary material, lie below or near the critical value of $0.5$ in summer (JJA) and fall (SON), indicating asymptotic independence. In spring (MAM) and winter (DJF), $\alpha$ is substantially larger; however, the max-stability assumption is not supported uniformly over space (see Section~\ref{s:exploratory}). For the winter season, where max-stability tests largely fail to reject, the posterior mean of $\alpha$ exceeds $0.9$ in both Table~\ref{tab:data_results_alpha1} and Table~\ref{tab:data_results_alpha2}. This indicates strong asymptotic dependence and confirms that the model captures this behavior.

\begin{table}[!t]
\small
\centering
\resizebox{\textwidth}{!}{\begin{tabular}{lllll}
  \hline
Method & Spring & Summer & Fall & Winter \\ 
  \hline
Full GP & 0.72 (0.70, 0.74) & 0.25 (0.23, 0.27) & 0.51 (0.49, 0.53) & 0.91 (0.89, 0.93) \\ 
  Low Rank - 100 & 0.79 (0.79, 0.79) & 0.85 (0.85, 0.85) & 0.81 (0.81, 0.81) & 0.76 (0.76, 0.76) \\ 
  Low Rank - 500 & 0.81 (0.81, 0.81) & 0.84 (0.84, 0.84) & 0.81 (0.81, 0.81) & 0.80 (0.80, 0.80) \\ 
  Taper - 90\% & 0.66 (0.64, 0.68) & 0.26 (0.24, 0.28) & 0.46 (0.44, 0.48) & 0.87 (0.85, 0.89) \\ 
  Taper - 20\% & 0.71 (0.69, 0.73) & 0.25 (0.23, 0.27) & 0.50 (0.48, 0.52) & 0.90 (0.88, 0.93) \\ 
  Vecchia - m=3 & 0.68 (0.66, 0.70) & 0.24 (0.22, 0.26) & 0.47 (0.44, 0.49) & 0.87 (0.85, 0.89) \\ 
  Vecchia - m=5 & 0.74 (0.72, 0.76) & 0.25 (0.23, 0.27) & 0.52 (0.50, 0.55) & 0.94 (0.92, 0.97) \\ 
  Vecchia - m=10 & 0.72 (0.70, 0.74) & 0.25 (0.23, 0.27) & 0.51 (0.49, 0.53) & 0.91 (0.89, 0.93) \\ 
  Vecchia - m=20 & 0.72 (0.70, 0.74) & 0.25 (0.23, 0.27) & 0.51 (0.48, 0.53) & 0.91 (0.89, 0.93) \\ 
   Neural Bayes&  0.66 (0.20, 0.74)& 0.33 (0.05, 0.30) & 0.56 (0.11, 0.58) & 0.72 (0.32, 0.82)\\ 
   \hline
\end{tabular}}
\caption{Posterior mean estimates and 95\% credible intervals for the $\alpha$ parameter under the Mat\'ern covariance function with smoothness parameter $\nu = 0.5$, for the NLDAS surface skin temperature application. Results are shown for spring, summer, fall, and winter. Each method is fitted separately to deseasonalized data across 1,209 grid locations over 46 years.}
\label{tab:data_results_alpha1}
\normalsize
\end{table}

\begin{figure}[!t]
\begin{subfigure}{0.5\textwidth}
    \centering
    \includegraphics[width=\linewidth]{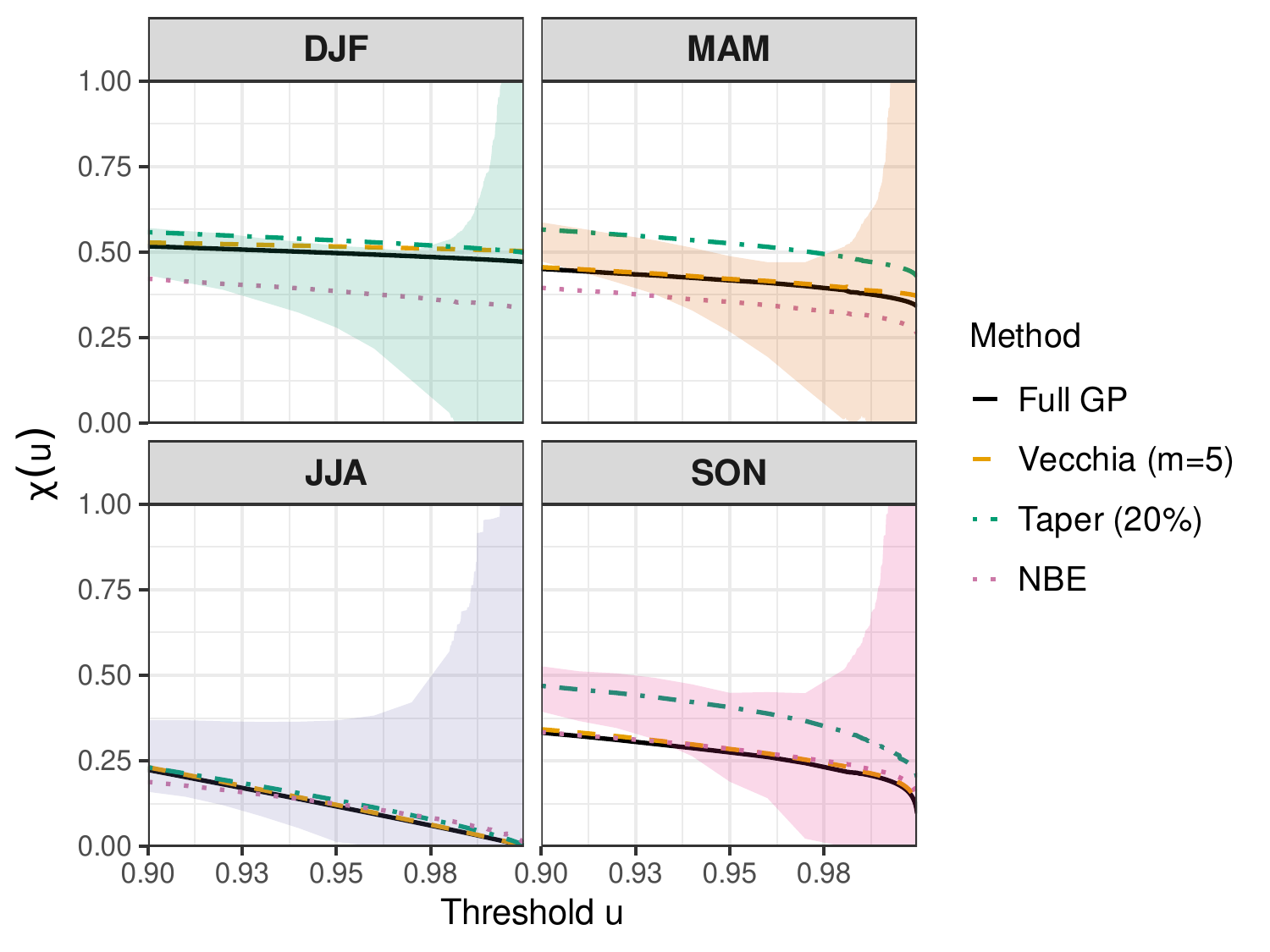}
    \caption{Smoothness $\nu = 0.5$}
\end{subfigure}
\hfill
\begin{subfigure}{0.5\textwidth}
    \centering
    \includegraphics[width=\linewidth]{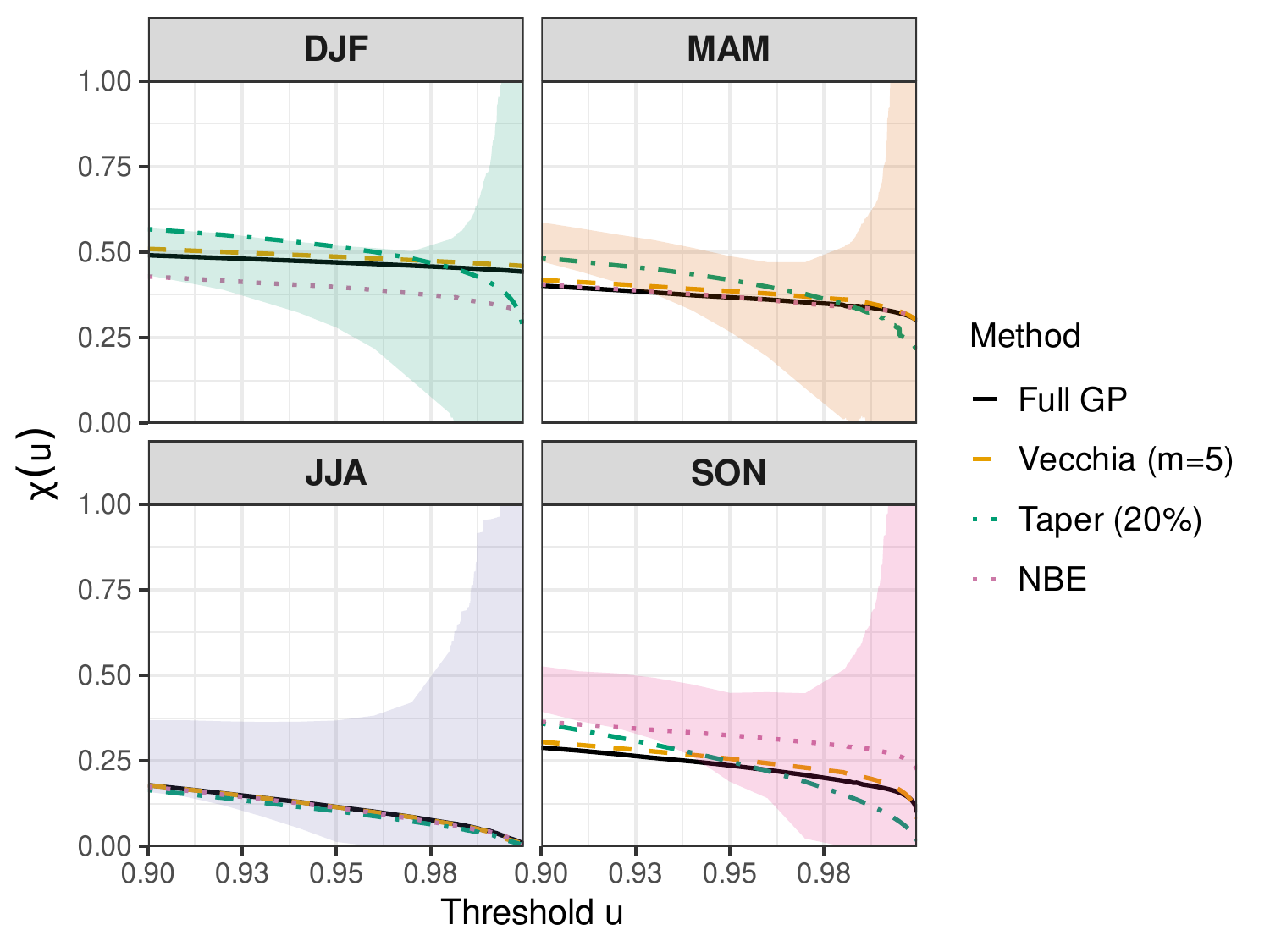}
    \caption{Smoothness $\nu = 1.5$}
\end{subfigure}

\caption{
Estimated $\chi_u(h)$ extremal dependence curves for the NLDAS $T_{skin}$ data, at spatial lag $h=0.177$ ($\sim75$ km). 
Results are shown for the four seasons (DJF, MAM, JJA, SON). The left and right panels correspond to Mat\'ern covariance functions with smoothness parameters $\nu=0.5$ and $\nu=1.5$, respectively. Modeling approaches include the full Gaussian process (Full GP), the Vecchia approximation with conditioning set size $m=5$, covariance tapering with approximately $20\%$ sparsity, and the Neural Bayes Estimator (NBE). Shaded regions represent empirical uncertainty envelopes of $\chi_u(h)$ estimated from the observed data. 
}

\label{fig:chi_seasonal_nldas}
\end{figure}

Together with the aforementioned max-stability diagnostics,  the estimated $\chi$-plots (Figure~\ref{fig:chi_seasonal_nldas}) and posterior summaries of $\alpha$ highlight the need for models that can accommodate \emph{both} asymptotic dependence and asymptotic independence, like the proposed LRSM. Max-stable models fail to capture weakening dependence at finite thresholds~\citep[see also][]{huser2025modeling}, while asymptotically independent models will miss the strong asymptotic dependence observed during winter and in localized regions.  

Across inference methods, the full GP, tapering, and Vecchia approximations yield similar estimates of $\alpha$, with clear seasonal variation. The $\chi_u(h)$ curves in Figure~\ref{fig:chi_seasonal_nldas} support these findings, showing rapid decay in summer and fall and sustained dependence in winter.

\subsection{Validation using sub-grid GHCN stations}
\begin{figure}[!t]
\centering
\includegraphics[width=\textwidth]{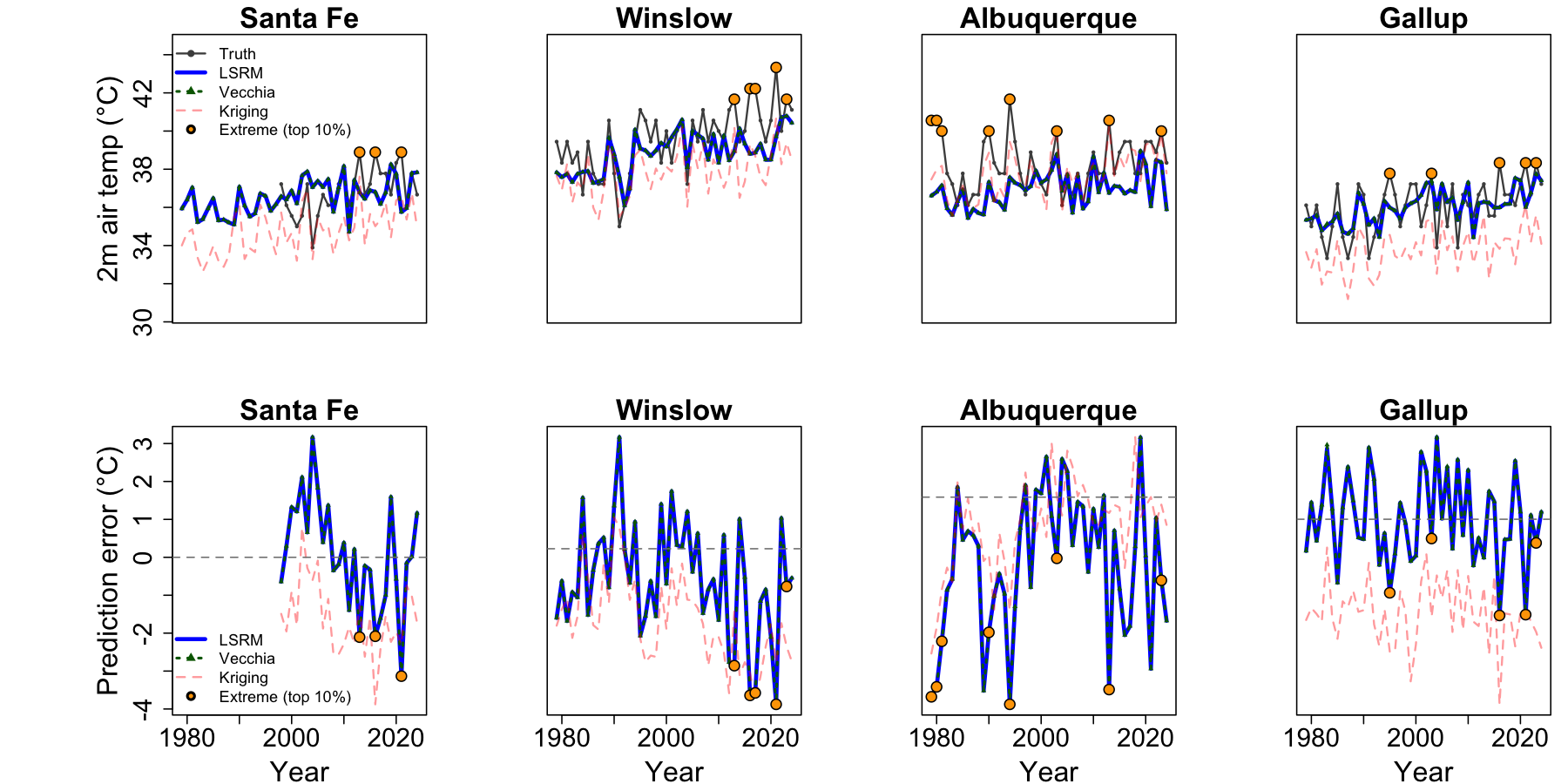}
\caption{
\textbf{Top row:} Observed (from GHCN) and predicted summer maximum 2-m air temperature at four validation sites. The LRSM-based predictions are obtained by first fitting the LRSM to $T_{skin}$, then propagating the resulting posterior predictions of $T_{skin}$ through the trained XGBoost calibration model to obtain $T_{2m}$.
\textbf{Bottom row:} Signed prediction errors (predicted minus observed), where negative values indicate underestimation. Extreme years are highlighted with circled points. Traditional kriging of $T_{2m}$ underestimates high temperatures, while LRSM predictions stay closer to zero and better capture extremes. Vecchia with five neighbors closely matches the LRSM based on the full GP.
}
\label{fig:ghcn_validation}
\end{figure}
An XGBoost model is trained by integrating complementary data sources to map \textit{interpolated} $T_{skin}$ to $T_{2m}$ using physically motivated covariates that modulate the {$T_{skin} - T_{2m}$} gradient. XGBoost is widely used in remote sensing and environmental prediction due to its ability to capture nonlinear and nonadditive relationships \citep{schutz2025evaluating}.  Boosting methods are well suited for tabular data due to their computational efficiency, robustness in moderate sample sizes, and strong predictive performance with minimal tuning \citep{chen2016xgboost}. Boosting has also been adopted for modeling of spatial-extremes \citep{koh2023gradient,velthoen2023gradient} and, in combination with Bayesian spatial models, XGBoost has been shown to reliably forecast extreme environmental events \citep{hu2026xgboost}.

We emphasize that our XGBoost model is not trained as a standalone predictor of $T_{2m}$ from raw inputs alone. Following the literature on remotely sensed estimation of near-surface air temperature, we combine LRSM-based posterior predictions of $T_{skin}$ with vegetation and meteorological variables (e.g., NDVI, wind speed, elevation, etc.) to capture evapotranspiration, aerodynamic resistance, and land--atmosphere coupling \citep{oyler2016remotely,cho2020improvement}. Thus, the quality of the downstream $T_{2m}$ estimates depends directly on the quality of the upstream LRSM fit to the spatial extremal dependence in $T_{skin}$. Advantages of XGBoost in this setting include its built-in handling of missing values, which is important for satellite products with gaps due to cloud cover or sensor limitations, and its ability to quantify variable importance, allowing us to assess which variables are most influential for predicting $T_{2m}$ across seasons and locations. The details of the XGBoost implementation are provided in Section~\ref{sec:XGBoost} of the supplementary material.

Gridded products provide the areal \(T_{2m}\) average on a lattice of grid cells, but many applications require predictions at arbitrary point locations within the domain (i.e., a sub-grid). To validate the modeling framework, we evaluate the trained XGBoost model at four sub-grid independent GHCN stations (Santa Fe, NM; Winslow, AZ; Albuquerque, NM; and Gallup, NM), using posterior predictions (from the fitted LRSM) of summer maximum $T_{skin}$ as inputs to estimate point-level $T_{2m}$. The relevant topographic and meteorological control is \emph{station-specific}, so we use the elevation, NDVI, and other covariates at the station coordinates as the predictor for sub-grid $T_{2m}$. As a competing approach, we consider standard GP kriging of the NLDAS $T_{2m}$ surface, which is a GP model fitted directly to $T_{2m}$ without a random scaling component and then kriged to the four validation locations.

Across the validation sites (Figure~\ref{fig:ghcn_validation}), the LRSM--XGBoost pipeline based on the full GP and its Vecchia approximation with conditioning set size $m=5$ outperform standard GP kriging of $T_{2m}$. The full GP and Vecchia approximations show strong agreement, producing near-identical results across all locations. In contrast, traditional GP kriging tends to underestimate summer maximum $T_{2m}$, particularly for the most extreme temperatures. Although NLDAS $T_{2m}$ is highly reliable, it remains a 1/8-degree grid box average and therefore does not necessarily interpolate well to a specific sub-grid point location (or a finer-resolution/downscaled grid box). This difference is most pronounced in Santa Fe, Winslow, and Gallup, where kriging substantially underestimates high temperatures. These results highlight that accurate downstream prediction of $T_{2m}$ depends on first capturing the spatial extremal dependence structure in $T_{skin}$ through the LRSM.

Here, direct kriging of \(T_{2m}\) serves as a pragmatic baseline accessible to practitioners. A more principled statistical formulation is a change-of-support prediction model \citep[e.g.,][]{wikle2005combining}, but existing methods are primarily designed under simpler GP settings. Extending this to spatial extremes is technically challenging and beyond the scope of this study, as it requires embedding change-of-support modeling within a spatial extremes framework. In contrast, our LRSM--XGBoost pipeline provides a practical workaround by learning the \(T_{\mathrm{skin}}\)–\(T_{2m}\) relationship and generates point-level predictions by evaluating the learned mapping at point-level covariates (e.g., elevation, NDVI) rather than areal averages.

\subsection{Downstream application: Heat index}
After validating the LRSM--XGBoost pipeline, we use it to propagate uncertainty from the fitted LRSM into downstream heat-risk summaries. Specifically, posterior predictions of $T_{skin}$ at holdout locations from the fitted LRSM are passed through the trained XGBoost calibration model to obtain posterior predictions of $T_{2m}$, which are then combined with relative humidity (i.e., 2-m dew point temperature) to compute the ``feels-like'' NWS heat index using the \texttt{weathermetrics} \texttt{R} package. In arid regions, the evaporative cooling of low humidity can suppress heat index values relative to $T_{2m}$ alone; nevertheless, high-risk thresholds (e.g., 103$^\circ$F $\approx$ 39.4$^\circ$C) are sensitive to errors in the upper-tail of $T_{2m}$. Here, we directly use the NLDAS 2-m dew point temperatures for the holdout locations; if 2-m dew point temperatures were truly missing at an unobserved location, the same $T_{skin}$--$T_{2m}$ modeling workflow could be applied to impute it. 

We focus on 2000, the year with the hottest summer in the Four Corners region in the NLDAS $T_{skin}$ record.
\begin{figure}[!t]
\centering
\includegraphics[width=0.9\textwidth]{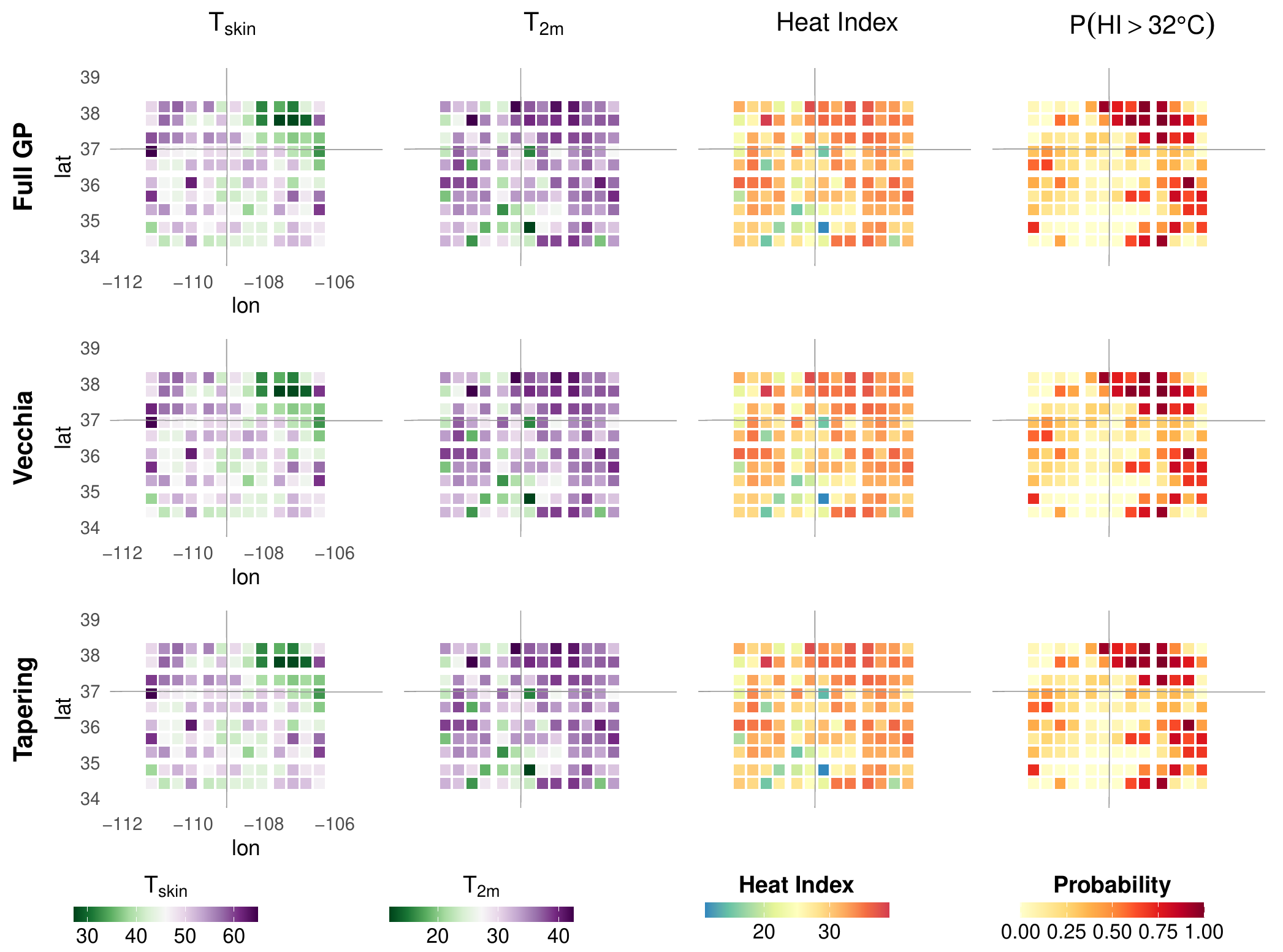}
\vskip -0.1cm
\caption{Posterior median fields of $T_{skin}$, $T_{2m}$, and heat index (HI), and exceedance probabilities $\Pr(\mathrm{HI} > 32^\circ\mathrm{C})$ for 2000. Rows show inference methods: Full GP {\bf (top)}, Vecchia ($m=5$; {\bf middle}), and covariance tapering (20\% sparsity; {\bf bottom}). The threshold $\Pr(\mathrm{HI} > 32^\circ\mathrm{C})$ corresponds to the National Weather Service ``extreme caution'' category. Supplementary Figure~\ref{fig:sd_posterior_summary_2000} displays pointwise 95\% posterior credible intervals for the $T_{2m}$ and heat index predictions and for the exceedance probabilities shown here.}
\label{fig:posterior_summary_2000}
\end{figure}
Our goal is to compare statistical models for surface temperature extremes by propagating their outputs through XGBoost and into downstream heat index. The full GP, Vecchia approximation with conditioning set size $m=5$, and tapering with $20\%$ sparsity perform similarly in posterior inference of $\alpha$ (Table~\ref{tab:data_results_alpha1}); therefore, we focus on these three methods in the downstream analysis. Figure~\ref{fig:posterior_summary_2000} summarizes the posterior median fields obtained by this workflow, showing how $T_{skin}$ propagates through $T_{2m}$ to the heat index. We also report the posterior exceedance probability $\Pr(\mathrm{HI}>32^\circ\mathrm{C})$, a useful metric linked to heat-related health risk. Here, $32$--$41^\circ\mathrm{C}$ ($90$--$105^\circ\mathrm{F}$) corresponds to the NWS ``Extreme Caution'' category \citep{noaa_heat_safety}.

The values of $\Pr(\mathrm{HI}>32^\circ\mathrm{C})$ are nearly indistinguishable across all methods. Similar agreement is observed for $T_{skin}$, $T_{2m}$, and $\mathrm{HI}$ (Figure~\ref{fig:posterior_summary_2000}). Given this, the Vecchia approximation with $m=5$ is a strong scalable alternative to the full GP LRSM, achieving near-indistinguishable results with substantially lower computational cost.
\begin{figure}[!t]
\centering
\includegraphics[width=0.56\linewidth]{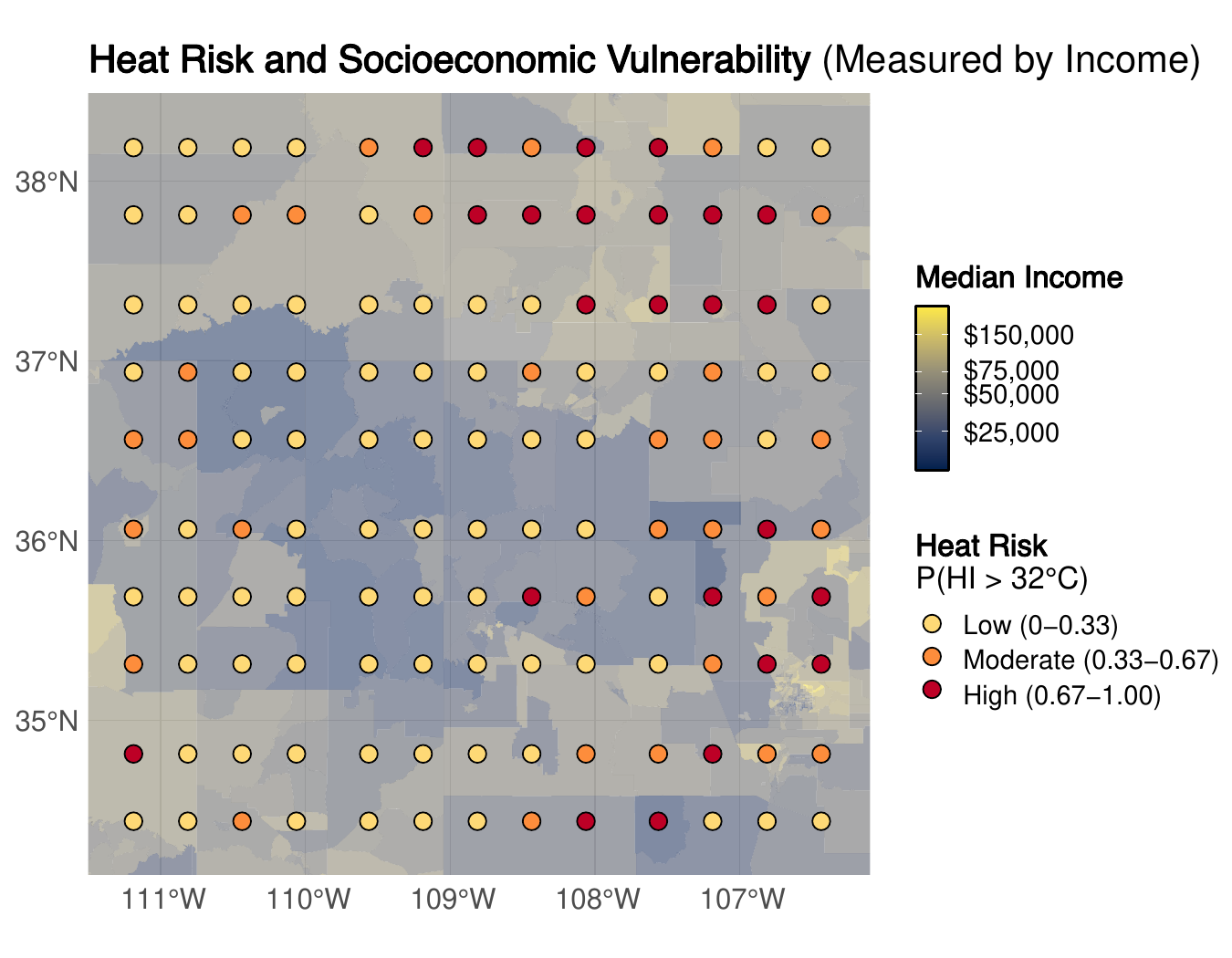}
\vskip -0.4cm
\caption{
Heat risk and socioeconomic vulnerability in the Four Corners region. Median household income is shown by census tract, with posterior median exceedance probabilities $\Pr(\mathrm{HI}>32^\circ\mathrm{C})$ from the LRSM--XGBoost workflow plotted at holdout locations. High-risk, low-income locations indicate greater vulnerability.}
\label{fig:vulnerability_map}
\end{figure}
\vspace*{-0.3cm}\paragraph{Practical implications}
This framework enables the use of high-resolution satellite $T_{skin}$ products to generate spatially detailed $T_{2m}$ and heat index estimates beyond what coarse reanalyses (e.g., ERA5 at $\sim$31 km) or NLDAS-like products can represent. By combining socioeconomic indicators with \emph{posterior exceedance probabilities} derived from the LRSM--XGBoost workflow, the proposed LRSM framework can be used to identify vulnerable communities where elevated heat risk and low median household income co-occur. Figure~\ref{fig:vulnerability_map} links the posterior heat-index exceedance surface $\Pr(\mathrm{HI}>32^\circ\mathrm{C})$ to the median household income at the census-tract level. Several locations exhibit both elevated heat risk, as indicated by $\Pr(\mathrm{HI}>32^\circ\mathrm{C})>0.67$, and low median household income ($<\$50{,}000$), suggesting that extreme heat risk is spatially heterogeneous and can disproportionately affect communities with fewer resources to mitigate exposure. High-risk locations are not spread uniformly across the region, but are clustered in parts of the central-to-eastern and northern regions. High-risk points appear in both lower- and higher-income tracts, indicating that heat risk and income do not align uniformly. This spatial heterogeneity underscores the need to jointly analyze exceedance surfaces and tract-level socioeconomic data, as true vulnerability reflects both heat risk and the local capacity to respond. More broadly, these results demonstrate that spatial extremes models can inform downstream societal impacts and targeted mitigation strategies.

\section{Discussion}\label{s:discussion}
This case study analyzes a 46-year record of extreme surface skin temperatures over the Four Corners region and evaluates their downstream impacts for climate risk management. We present an end-to-end analysis from the perspective of a statistical risk consultant, encompassing data acquisition, exploratory analysis, model fitting and validation, and downstream impact assessment. To address the computational challenges of spatial extremes, we develop scalable inference methods based on the L\'evy Random Scale Mixture (LRSM), including a likelihood-free, AI-based approach. We benchmark these methods on the surface skin temperature data, comparing inferential accuracy, predictive performance, and computational cost. To quantify downstream risk, we translate surface skin temperature into 2-m air temperature and combine it with humidity to estimate heat indices, demonstrating how modeling choices for extremes propagate into heat stress assessments.

Beyond the case study, simulation studies evaluate the proposed scalable approaches across varying spatial and extremal dependence settings. Methods that best preserve local spatial structure, especially full Gaussian processes and higher-order Vecchia approximations, consistently perform best. NBEs provide dramatic computational speedups through amortized learning, delivering near-instantaneous inference after training. However, the NBE approach still requires up-front costs for model training in the amortization step, and the current formulation does not offer spatial predictions. Vecchia approximation models with smaller conditioning sets offer a favorable balance between accuracy and runtimes. These findings illustrate the trade-offs among modern computational tools for spatial extremes with block maxima and inform their use in large-scale environmental applications.

The assumption of stationary extremal dependence is limiting, as localized climate regimes may induce varying spatial dependence features. Adapting the proposed framework to accommodate nonstationarity would improve model flexibility. Extending the framework to encompass fully unified Bayesian hierarchical models that jointly capture marginal and extremal dependence \citep[e.g.,][]{zhang2024leveraging, shi2024spatial} remains an important direction for future research, particularly in developing computationally scalable approaches for large spatial datasets. The costs associated with the marginal distributions are substantial since neither $F_X^{-1}(\cdot; \alpha)$ nor $f_X(\cdot; \alpha)$ has a closed form, requiring numerical approximation and evaluation at all $n \times T$ location–time pairs. Developing tractable approximations or surrogate formulations would substantially improve computational efficiency. 

Beyond the current LRSM setting under a block-maxima framework, the scalable inference strategies considered here are not inherently limited to this model class. Methods such as Vecchia approximations, covariance tapering, and amortized likelihood-free inference may also prove useful in peaks-over-threshold formulations and in other models for spatial extremal dependence, including Bayesian conditioning-based approaches \citep{Simpson2023}. Developing and assessing such extensions is a promising next step toward a broader computational toolkit for spatial extremes.
\vspace{0.2in}

\noindent \textbf{Data Availability Statement} The surface skin temperature data used in this study are derived from the North American Land Data Assimilation System (NLDAS) Phase 2. The processed data are available at \url{https://doi.org/10.5281/zenodo.19632717}.

\clearpage

{
\setstretch{0.75}

\setlength{\bibsep}{1pt}
\addcontentsline{toc}{section}{References}
\bibliography{main}
}
\newpage
\appendix
\section*{Supplementary Material for``Spatial Extremes at Scale: A Case Study of Surface Skin Temperature and Heat Risk in the United States''}

\vspace{1em}

\setcounter{page}{1}
\renewcommand{\thepage}{S.\arabic{page}}
\setcounter{section}{0}
\renewcommand{\thesection}{S.\arabic{section}}
\setcounter{equation}{0}
\renewcommand{\theequation}{S.\arabic{equation}}
\setcounter{table}{0}
\renewcommand{\thetable}{S.\arabic{table}}
\setcounter{figure}{0}
\renewcommand{\thefigure}{S.\arabic{figure}}
\setcounter{algorithm}{0}
\renewcommand{\thealgorithm}{S.\arabic{algorithm}}

\begin{abstract}
This supplement provides methodological, theoretical, and computational details supporting the main manuscript. We describe NLDAS surface skin temperature data preprocessing, seasonal GEV modeling with Gaussian process priors and uniform-scale transformations. Exploratory analyses and diagnostics highlight that max-stable models may not be appropriate for modeling seasonal extremes. We develop theoretical properties of the L\'evy random scale mixture model (LRSM), including extremal dependence and marginal behavior, and provide derivations for the LRSM likelihood functions. Details of our scalable inference strategies, including Vecchia approximations, covariance tapering, and low-rank models, are provided alongside discussion of the implementation procedures. Additional simulation results under model misspecification are provided, as well as supplemental figures and tables from the primary simulation study. The supplement is organized into sections on data processing, diagnostics, theory, computation, and simulations. \end{abstract}

\setcounter{tocdepth}{2}   
\tableofcontents

\newpage

\def\spacingset#1{\renewcommand{\baselinestretch}%
{#1}\small\normalsize} \spacingset{1}

\doublespacing
\numberwithin{equation}{section}
\numberwithin{figure}{section}

\section{NLDAS Data Processing}\label{sm:NLDAS_process}
In this section, we describe how we process surface skin temperature data from the NLDAS (see Sections~\ref{s:motivation:data} and \ref{s:exploratory} in the main text). In particular, we detail the construction of seasonal block maxima and their transformation to the uniform scale.

The Four Corners study region spans longitudes $[-111.2,-106.4]$ and latitudes $[34.44,38.19]$, corresponding to ranges of $4.8^\circ$ and $3.75^\circ$, respectively. Using a mean latitude of $36.31^\circ$, one degree of latitude is approximately $111.3$ km, while one degree of longitude is approximately $111.3\cos(36.31^\circ)\approx 89.7$ km. This yields an approximate domain size of $430.6$ km in the longitudinal direction and $417.5$ km in the latitudinal direction yielding a maximum (diagonal) distance of $600$ km. In this study, the longitude and latitude coordinates were scaled to the unit square.

For data processing, we first assign each month to a climatological season:
DJF$=\{12,1,2\}$, MAM$=\{3,4,5\}$, JJA$=\{6,7,8\}$, and SON$=\{9,10,11\}$. To ensure that DJF spans a single season year, December is advanced to the subsequent calendar year. Second, for each of the 1,209 grid locations, we extract seasonal block maxima from the daily series for each season $d\in\{\text{DJF, MAM, JJA, SON}\}$ and season year $t$, denoted by
\begin{equation*}
\{Y^d_{it} \equiv Y_t^d(\bs_i): i=1,\ldots,1209, \; t=1,\ldots,46\}.
\end{equation*}
Seasonality is handled by fitting separate GEV models for each season, effectively removing the intra-annual cycle.  
Within each season $d$ and location $\bs_i$, we explicitly account for measurement errors and relate the observed seasonal maxima $Y^d_{it}$ to a latent GEV field through a Gaussian measurement model with unknown error variance $\sigma_\varepsilon^2$:  
\begin{equation}\label{eqn:data_model}
    \begin{split}
        Y^d_{it}\mid (X^d_{it}, \sigma_\varepsilon^2 )&= X^d_{it} + \varepsilon_{it}, 
        \qquad \varepsilon_{it}\stackrel{\text{iid}}{\sim}\mathcal{N}(0,\sigma_\varepsilon^2),\\
        X^d_{it}\mid \mu^d(\bs_i),\,\sigma^d(\bs_i),\,\xi^d(\bs_i) &\stackrel{\text{iid}}{\sim} 
        \mathrm{GEV}\!\left(\mu^d(\bs_i),\,\sigma^d(\bs_i),\,\xi^d(\bs_i)\right),
        \quad t=1,\ldots,46,
    \end{split}
\end{equation}
where the conditional independence structure of the latent GEV field $\{X^d_{it}\}$ follows the standard spatial GEV model 
\citep[see, e.g.,][]{cooley2007bayesian, risser2019probabilistic}, augmented with a Gaussian nugget term to mitigate small-sample noise. 

To propagate spatial information, we impose a Gaussian process (GP) prior on each parameter field
$\kappa^d(\bs_i) \in \{\mu^d(\bs_i),\sigma^d(\bs_i),\xi^d(\bs_i)\}$,
using elevation as a fixed effect and a Matérn covariance:
\begin{equation}\label{eqn:para_model}
    \begin{split}
        \kappa^d(\bs_i) \;=\; \beta^\theta_0 &+ \beta^\theta_1\,\mathrm{elev}(\bs_i) + \eta(\bs_i),\qquad
        \eta(\bs) \sim \mathcal{GP}(0,C_{\boldsymbol{\theta}^d}(\bs,\bs')),
    \end{split}
\end{equation}
where $\mathrm{elev}(\bs_i)$ denotes elevation (km), and $C_{\boldsymbol{\theta}^d}(\cdot,\cdot)$ is the Mat\'ern correlation function with range $\rho^d$, smoothness $\nu^d$, and marginal variance $\sigma^{2,d}$; see Equation~\eqref{eq:MATERN} for its expression.

Applying a full Bayesian MCMC fit for the hierarchical model defined by~\eqref{eqn:data_model} and \eqref{eqn:para_model} can be computationally intensive for $n=1209$ locations.  
Instead, we obtain the joint posterior mode, or maximum \textit{a posteriori} (MAP) estimate, which can be very efficiently obtained using gradient-based optimization with automatic differentiation in \texttt{PyTorch}. Specifically, we define
\begin{equation*}
\boldsymbol{\phi}_\kappa \;=\; \big(\beta^\kappa_0,\;\beta^\kappa_1,\;\sigma_{\kappa}^2,\;\rho_\kappa,\;\nu_\kappa\big),
\qquad \kappa\in\{\mu,\sigma,\xi\},
\end{equation*}
and collect all hyperparameters as $\boldsymbol{\Phi} \equiv \{\boldsymbol{\phi}_\mu,\boldsymbol{\phi}_\sigma,\boldsymbol{\phi}_\xi,\sigma_\varepsilon^2\}$.
By Bayes’ theorem, the joint posterior is
\begin{equation*}
\Pr(\mathbf{X}^d, \boldsymbol{\Theta}^d,\boldsymbol{\Phi} \mid \mathbf{Y}^d)
~\propto~ \Pr(\mathbf{Y}^d\mid \mathbf{X}^d,\sigma_\varepsilon^2)\,
\Pr(\mathbf{X}^d\mid \boldsymbol{\Theta}^d)\,
\Pr(\boldsymbol{\Theta}^d\mid \boldsymbol{\phi}_\mu,\boldsymbol{\phi}_\sigma,\boldsymbol{\phi}_\xi)\,
\Pr(\boldsymbol{\Phi}),
\end{equation*}
where $\boldsymbol{\Theta}^d=\{\mu^d(\bs_i),\sigma^d(\bs_i),\xi^d(\bs_i)\}_{i=1}^n$ denotes the three parameter fields at all sites.  
The joint MAP estimate is then
\begin{equation*}
(\hat{\mathbf{X}}^d, \hat{\boldsymbol{\Theta}}^d, \hat{\boldsymbol{\Phi}})
= 
\arg\max_{\mathbf{X}^d,\,\boldsymbol{\Theta}^d,\,\boldsymbol{\Phi}}
\, \Pr(\mathbf{X}^d, \boldsymbol{\Theta}^d, \boldsymbol{\Phi}
\mid 
\mathbf{Y}^d).
\end{equation*}
Figures~\ref{fig:gev_params} (in the main text) and~\ref{fig:gev_params2} show the spatial maps for the MAP estimates of the three GEV parameters. The bottom panels of Figure~\ref{fig:gev_params2} also show the Anderson–Darling goodness-of-fit $p$-values for the GEV model fitted to the seasonal maxima over our study region. The seasonal GEV fit is broadly satisfactory across most of the domain, with a few localized areas where fit may be weaker and the $p$-value is less than 0.05 (shown in gray).

Using the MAP fields, we construct a nugget-free latent copula via the probability integral transform
\begin{equation*}
U^d_{it}
\;=\;
G\left(\hat{X}^d_{it}\,;\,\hat\mu(\bs_i),\hat\sigma(\bs_i),\hat\xi(\bs_i)\right),
\end{equation*}
where $G(\cdot;\cdot)$ is the GEV distribution function.  
The transformed variables $\{U^d_{it}\}$ define a spatially regularized, physically coherent surface of extremes, free from marginal (i.e., location-wise) variations and small-sample estimation noise in the parameters.
\begin{figure}[!t]
\raggedright
    \includegraphics[width=\linewidth]{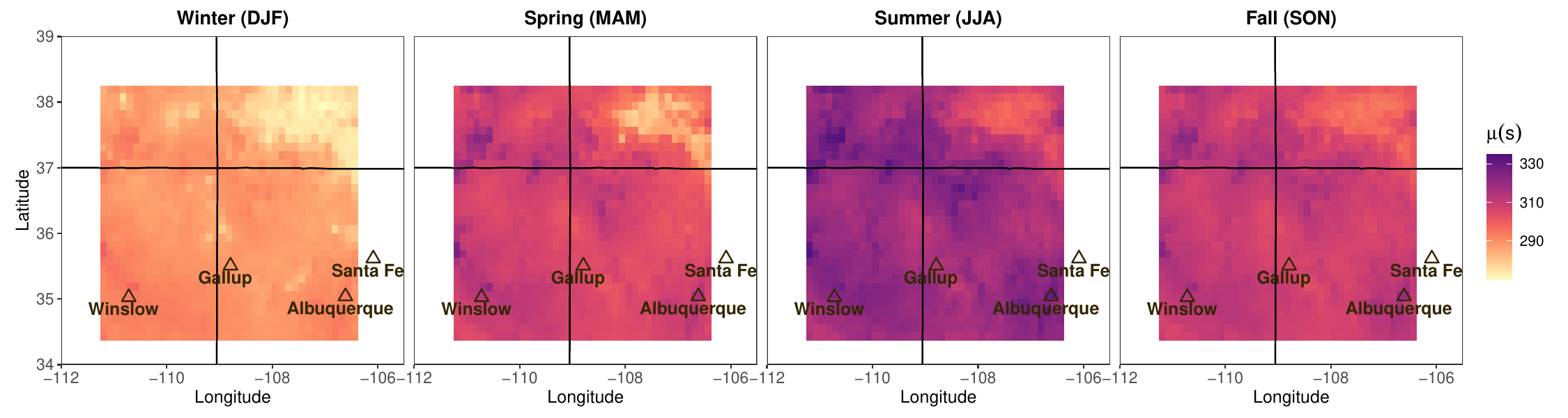}
    \includegraphics[width=\linewidth]{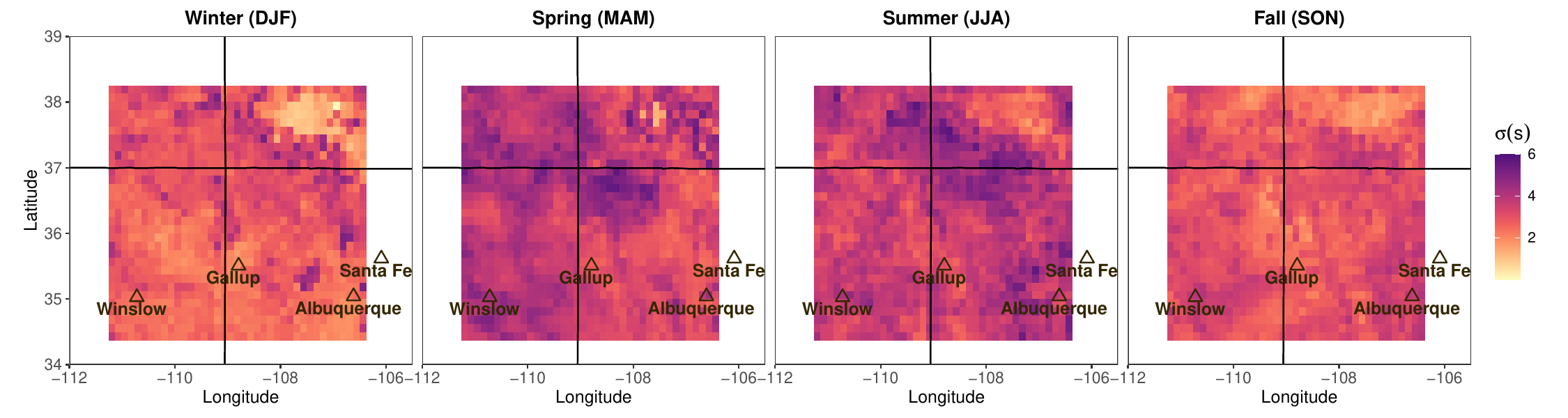}
    \hspace*{0.07cm}\includegraphics[width=1.03\linewidth]{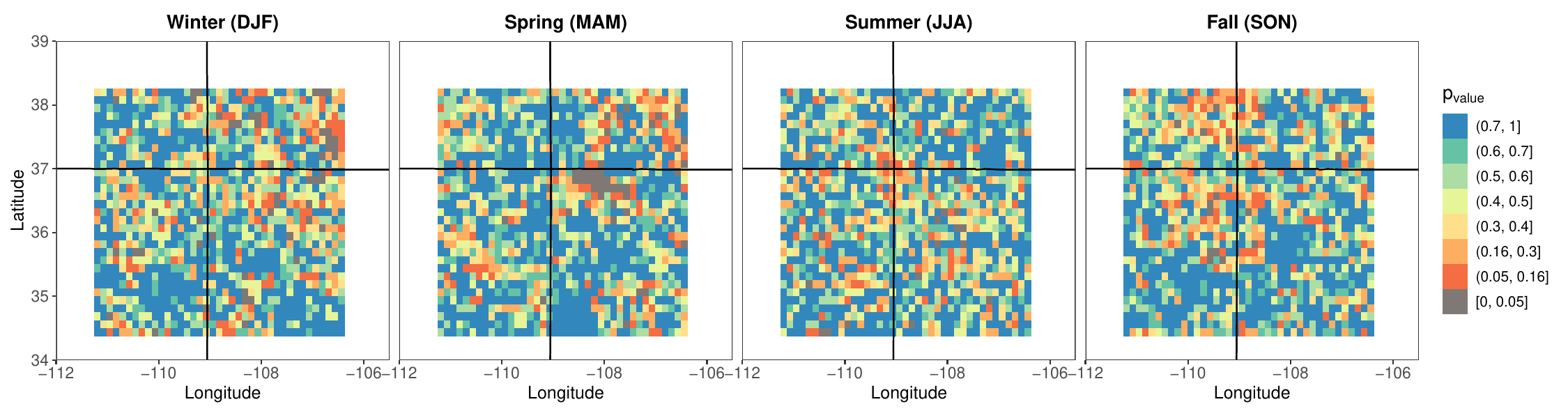}
    \caption{Seasonal maps of the MAP estimates of  $\{\mu(\bs)\}$ (top panels) and  $\{\sigma(\bs)\}$ (middle panels) over the study region for each season. The estimated fields of $\{\xi(\bs)\}$ are provided in Figure~\ref{fig:gev_params} of the main text; see its caption for further details. The bottom panels show Anderson--Darling goodness-of-fit $p$-values for the fitted GEV distribution at each grid cell. Smaller $p$-values, especially those below 0.05 (gray), provide stronger evidence against the adequacy of the local GEV fit.}
    \label{fig:gev_params2}
\end{figure}

The MAP-based approach is computationally efficient and provides reliable point estimates for the inputs $\{U^d_{it}\}$ used in our subsequent comparative study. However, it ignores posterior uncertainty in the marginal fits of our data, and this is not propagated through to the estimation of the spatial extremal dependence.  
Accurate uncertainty quantification, credible intervals, or posterior predictive inference would require full Bayesian sampling (e.g., MCMC). Recent studies \citep[e.g.,][]{zhang2022accounting, risser2025data, shi2024spatial} emphasize that modeling the marginal distributions and the copula dependence jointly within a unified Bayesian hierarchical framework can yield more accurate inferences than two-step estimation approaches where the marginal parameters and dependence parameters are estimated separately. \citet{shi2024spatial}, for instance, demonstrate through extensive simulation that unified models consistently outperform two-step estimation approaches in both parameter recovery and predictive performance (see their Figure~8). Moreover, \cite{kakampakou2024spatial} highlight that two-step estimation can sometimes lead to spurious conclusions about the nature of extremal dependence if the marginal step is not handled with appropriate care. 

Nevertheless, we adopt the two-step estimation framework here, as incorporating the marginal transformation layer into the hierarchical model introduces challenges for the covariance estimation methods compared in this study. Since the two-step approach remains standard practice in spatial extremes, we focus on this formulation and leave the unified model for future work.

\section{Background on Spatial (Extremes) Modeling}\label{s:preliminaries}
This section provides background on commonly-used spatial process models and their limitations for modeling spatial extremes. The discussion focuses on Gaussian and max-stable processes and their contrasting extremal dependence properties, supporting Sections \ref{s:exploratory} and \ref{s:rsm} of the main text, respectively.

\paragraph{Gaussian processes (GPs):} A Gaussian process is a collection of random variables $\{ Z(\bs) : \bs \in \mathcal{S} \}$, such that for any finite set of locations $\bs_1, \ldots, \bs_n \in \mathcal{S}$, the random vector 
$\bZ=\bigl(Z(\bs_1), \ldots, Z(\bs_n)\bigr)^{\top}$ follows a multivariate normal distribution. We denote a GP as $Z(\cdot) \sim \mathcal{GP}\bigl(\mu(\cdot), C_{\boldsymbol{\theta}}(\cdot, \cdot)\bigr)$, where the process is defined by its mean function $\mu(\bs) = \mathbb{E}\{Z(\bs)\}$ and positive-definite covariance function
$C_{\boldsymbol{\theta}}(\bs, \bs') = \mathrm{Cov}\{Z(\bs), Z(\bs')\}$ which is parameterized by $\boldsymbol{\theta}$. For simplicity, we hereafter assume that all spatial processes are standardized to have unit marginal variance at each site $\bs \in \mathcal{S}$, and so $C_{\boldsymbol{\theta}}(\bs, \bs')$ is the correlation function.

In spatial applications, it is common to assume that the correlation function $C_{\boldsymbol{\theta}}(\bs,\bs')$ is second-order stationary and isotropic, so that it depends only on the Euclidean distance between pairs of spatial locations, i.e., $h=\|\bs-\bs'\|$. A widely used choice  of correlation function is the Mat\'ern class, given by
\begin{equation}
\label{eq:MATERN}
C_{\boldsymbol{\theta}}(\bs,\bs')
=
\frac{1}{\Gamma(\nu)2^{\nu-1}}
\left(\sqrt{2\nu}\frac{h}{\rho}\right)^{\nu}
K_{\nu}\!\left(\sqrt{2\nu}\frac{h}{\rho}\right),
\end{equation}
which is parameterized by $\boldsymbol{\theta}=(\rho,\nu)$, and where $K_{\nu}(\cdot)$ denotes the modified Bessel function of the second kind. The parameter  $\rho>0$ is the spatial range parameter controlling the rate of correlation decay, and $\nu>0$ governs the smoothness of the process. Alternative parametric correlation models, including the spherical family, the generalized Cauchy class \citep{gneiting2004stochastic}, and the Dagum family \citep{berg2008dagum}, and compactly supported constructions such as the Askey \citep{askey1969integral} and Wendland  \citep{wendland1995piecewise} functions, offer flexibility in terms of the types of dependence that a GP can capture.

Despite their sub-asymptotic flexibility, GPs exhibit asymptotic independence between the process observed at any pair of distinct locations $\bs$ and $\bs'$ \citep{Sibuya1960,davison2013geostatistics}. Specifically, a GP has $\eta(\bs,\bs')=(1+C_{\boldsymbol{\theta}}(\bs, \bs'))/2$, implying that Gaussian processes cannot exhibit asymptotic dependence regardless of the covariance function, except in the case of \emph{perfect} dependence (with $C_{\boldsymbol{\theta}}(\bs, \bs')=1$). A class of stochastic process models that can capture asymptotic dependence is the max-stable process, and is widely applied in the spatial extremes literature \citep[see, e.g.,][]{blanchet2011spatial,reich2012hierarchical,oesting2017statistical,shao2024flexible}.

\paragraph{Max-stable processes (MSPs):} 

Suppose that $\{T_i(\bs)\}_{i=1}^n,$ are independent copies of a stochastic process $\{T(\bs):\bs\in\mathcal{S}\}$ with continuous sample paths on $\mathcal{S}\subset \mathbb{R}^2$, and that there exist normalizing functions $a_n(\bs)>0$ and $b_n(\bs)\in\mathbb{R}$ such that 
\[\frac{ \max_{1 \le i \le n}T_i(\bs)-b_n(\bs)}{a_n(\bs)} \]
converges weakly, as $n\rightarrow \infty$, to a stochastic process $\{X(\bs): \bs \in \mathcal{S}\subset \mathbb{R}^2\}$ with non-degenerate margins. Then $\{X(\cdot)\}$ is a max-stable process with GEV margins (see Section~\ref{s:exploratory} of the main text). Under mild regularity conditions and with an appropriate marginal standardization (without loss of generality), MSPs can be constructed using the random scale spectral representation of \citet{de1984spectral}, by setting
\begin{equation}\label{eq_MSP}
X(\bs){=} \sup_{i \geq 1} R_i W_i(\bs),\end{equation}
where $\{R_i\}_{i \in \mathbb{N}}$ are points of a Poisson process on $(0,\infty)$ with intensity $r^{-2}\,\mathrm{d}r$, and $\{W_i(\bs)\}_{i\in\mathbb{N}}$ are independent replicates of a non-negative stochastic process $W(\bs)$ with $\mathbb{E}\{W(\bs)\}=1$ for all $\bs\in\mathcal{S}$. 

Max-stable processes are the only possible limit for rescaled pointwise maxima of random fields with non-degenerate margins, making them an appealing model for use in the modeling of spatial extremes. However, max-stable processes suffer from two limitations. First, by their construction, max-stable processes exhibit asymptotic dependence at all spatial lags. That is, for any distinct pair of locations $\bs$ and $\bs'$, we have either 
$\chi(\bs,\bs')>0$ or perfect independence between $X(\bs)$ and $X(\bs')$. This is a result of the heavy-tailed behavior of the common Poisson point $R_i$, which influences multiple spatial locations through the shared spectral process $W_i(\cdot)$ \citep{huser2022}. While max-stable processes provide a natural theoretical framework for spatial block maxima, their inherent asymptotic dependence means they are restrictive for applications where extremal dependence weakens at high levels, (i.e., under asymptotic independence).

Secondly, as the likelihood function is generally intractable for max-stable processes observed at a moderate number of spatial locations \citep[$n\sim 10$;][]{padoan2010likelihood,castruccio2016high}, likelihood-based Bayesian inference for these models is only tractable for special cases or using pseudo likelihoods \citep[see, e.g.,][]{ribatet2012bayesian, reich2012hierarchical}. Due in part to these limitations, \citet{huser2025modeling} have advocated against the use of max-stable processes for modeling of spatial extremes. In the main text, we thus propose to use a more flexible random scale construction, which can capture both asymptotic dependence and asymptotic independence, and enjoys tractability in high dimensions.

\section{Max-Stability Tests}\label{sec:maxstab-tests}
This section provides additional details on the max-stability tests employed in Section~\ref{s:exploratory} of the main text. Max-stable processes impose a very rigid form of extremal dependence; see Section~\ref{s:preliminaries} for a brief overview. In particular, for any max-stable process with unit Fr\'echet margins, the tail dependence coefficient admits the first-order expansion
\begin{equation}\label{eqn:max-stab}
  \chi_u(h) - \chi(h) = O(1-u), \qquad u \nearrow 1,
\end{equation}
where $\chi(h) = \lim_{u\nearrow 1} \chi_u(h)>0$  always corresponds to asymptotic dependence  at any lag $h> 0$ \citep[see, e.g.,][]{de1984spectral,ribatet2013spatial,huser2019modeling}. Thus, for a max-stable process, the function $\chi_u(h)$ in \eqref{eqn:chi_definition} is forced to approach its limit \emph{linearly} in $(1-u)$, and this rate is entirely determined by the exponent measure. By contrast, for the LRSM model in~\eqref{eqn:lrsm}, a subasymptotic expansion (Theorem~\ref{thm:chi_eta}) shows that the approach of $\chi_u(h)$ to its limit can be substantially slower or faster than linear, depending on model parameters. Within our specification, the parameter $\alpha$ directly governs this rate of decay: for $\alpha<1/2$ the process is asymptotically independent, and $\chi_u(h)$ decays to zero more rapidly (yielding non-max-stability), whereas for $\alpha>1/2$ the process is asymptotically dependent but may still display weakened dependence at realistic, non-asymptotic levels. This additional flexibility is crucial for reproducing the empirical subasymptotic behavior observed in practice.

For the NLDAS surface temperature data over the Four Corners region, the spatial domain extends $600$ km (the diagonal distance from $(-111.2,34.44)$ to $(-106.4,38.19)$ is approximately $600$ km). Over such a large, topographically heterogeneous region with distinct climatological regimes, it is physically implausible that extremes share a single max-stable exponent measure. Our diagnostics confirm that max-stability is indeed too restrictive for these data. First, the posterior credible intervals for $\alpha$ in Tables~\ref{tab:data_results_alpha1} and~\ref{tab:data_results_alpha2} are well below, or close to, the critical value of $0.5$ in summer (JJA) and fall (SON), indicating asymptotic independence and, therefore, non-max-stability in those seasons. Even in spring (MAM) and winter (DJF), where the posterior for $\alpha$ suggests asymptotic dependence on average, the max-stability assumption is not supported uniformly over space.

\begin{figure}
    \centering
    \includegraphics[height=0.28\linewidth]{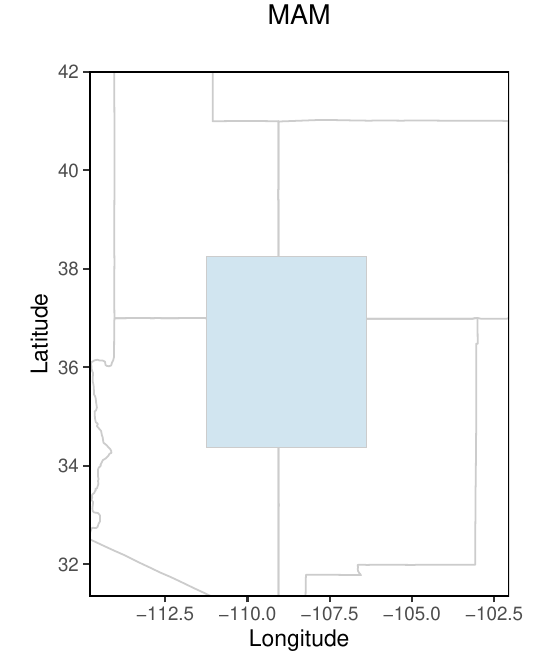}
    \hspace*{-0.2cm}\includegraphics[height=0.28\linewidth]{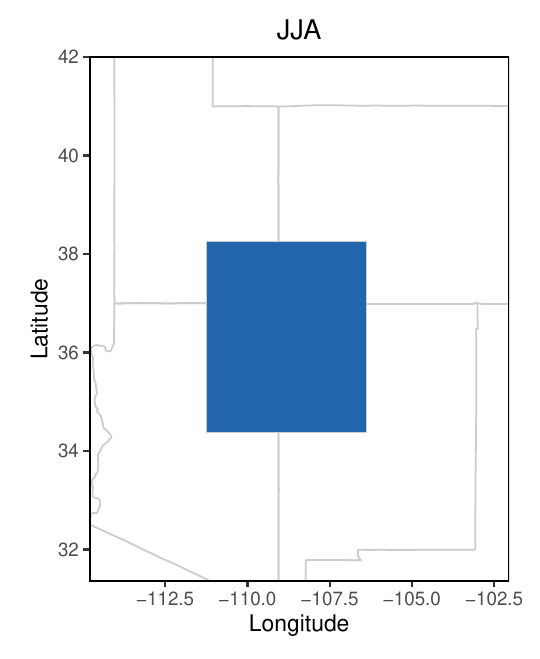}
    \hspace*{-0.15cm}\includegraphics[height=0.28\linewidth]{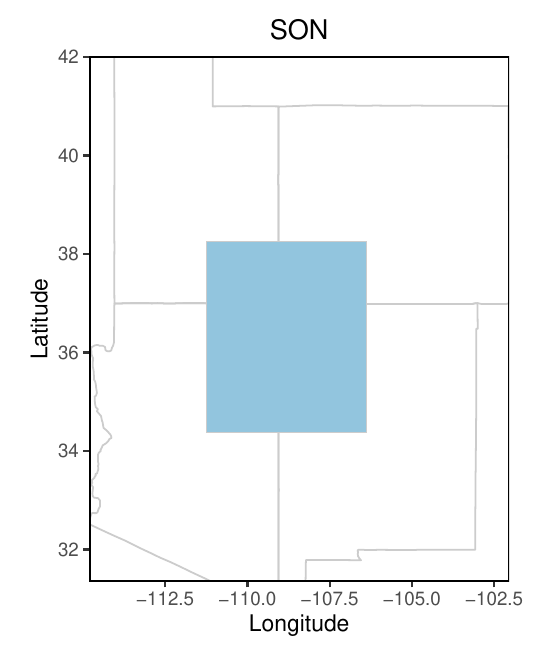}
    \includegraphics[height=0.28\linewidth]{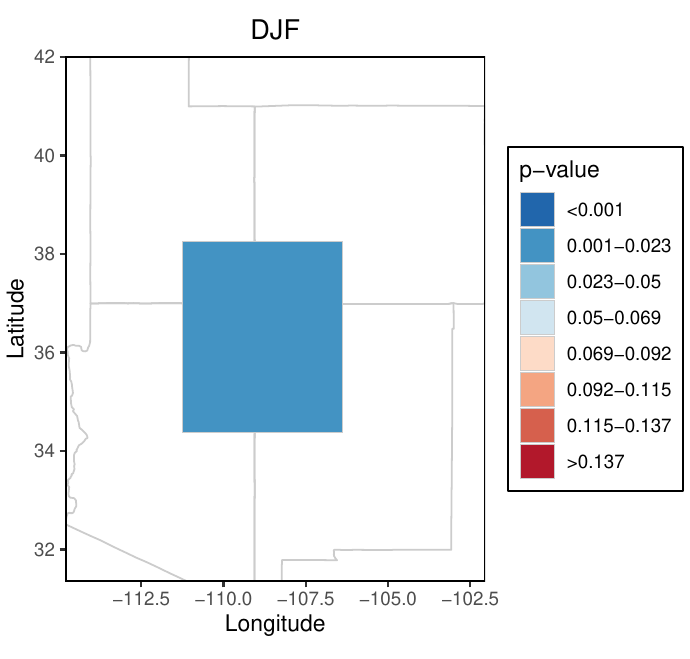}

    \includegraphics[height=0.28\linewidth]{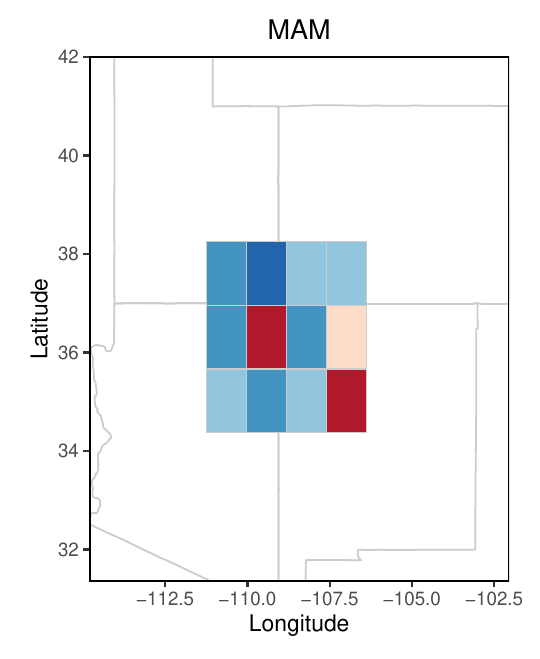}
    \hspace*{-0.2cm}\includegraphics[height=0.28\linewidth]{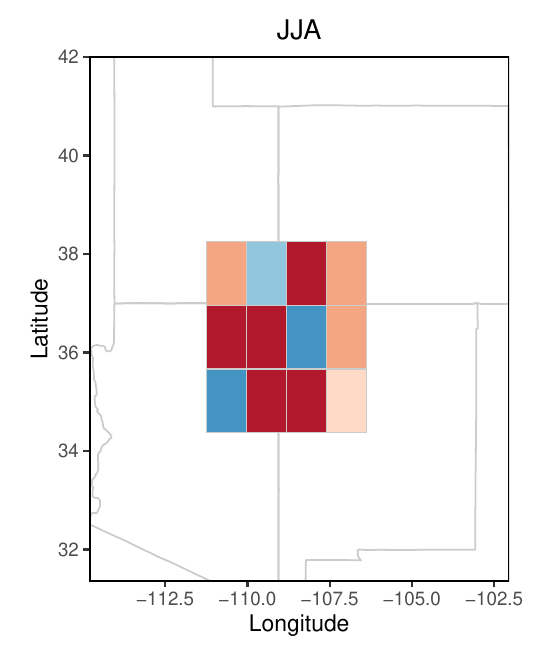}
    \hspace*{-0.15cm}\includegraphics[height=0.28\linewidth]{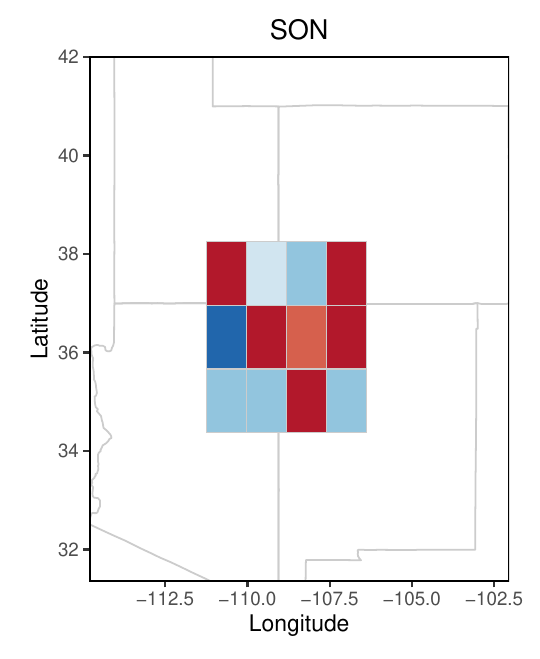}
    \includegraphics[height=0.28\linewidth]{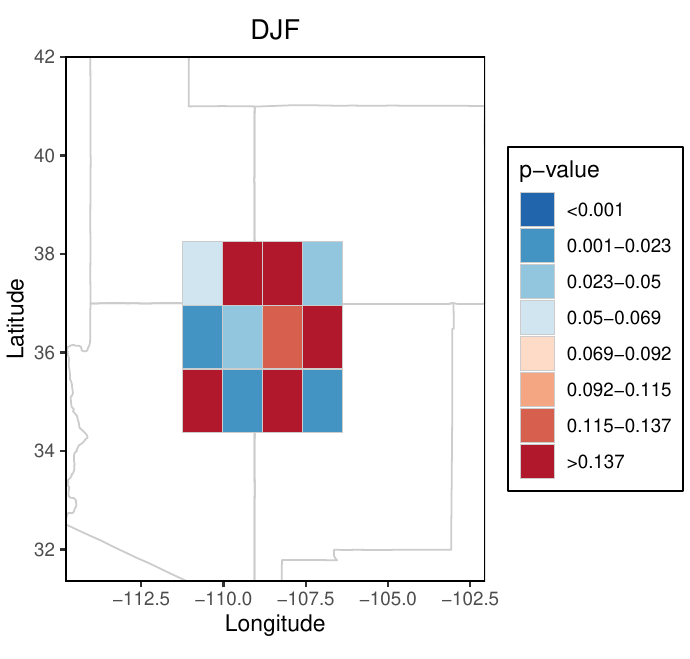}

    \includegraphics[height=0.28\linewidth]{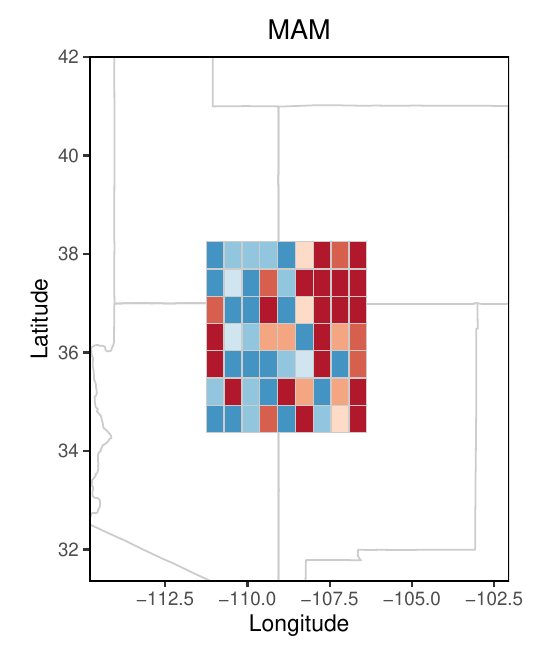}
    \hspace*{-0.2cm}\includegraphics[height=0.28\linewidth]{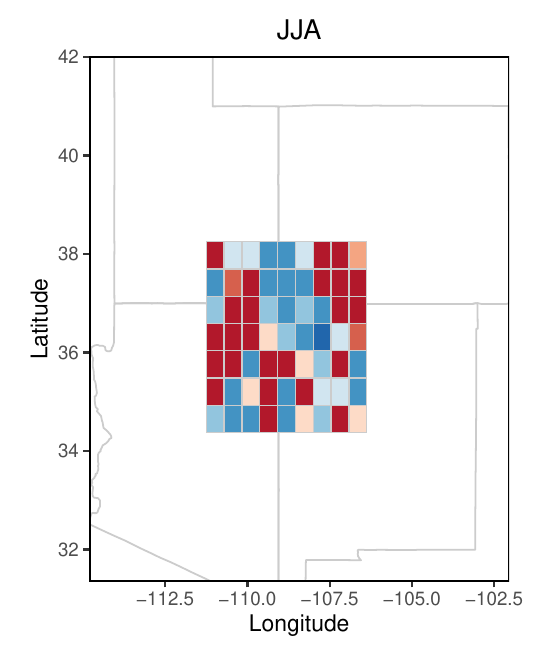}
    \hspace*{-0.15cm}\includegraphics[height=0.28\linewidth]{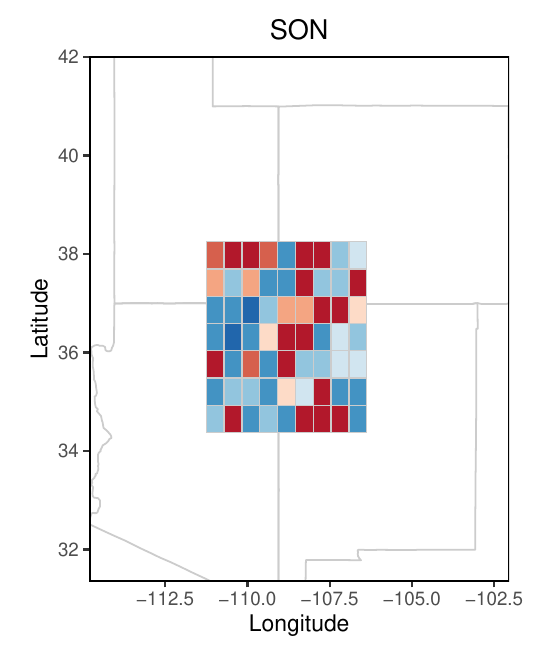}
    \includegraphics[height=0.28\linewidth]{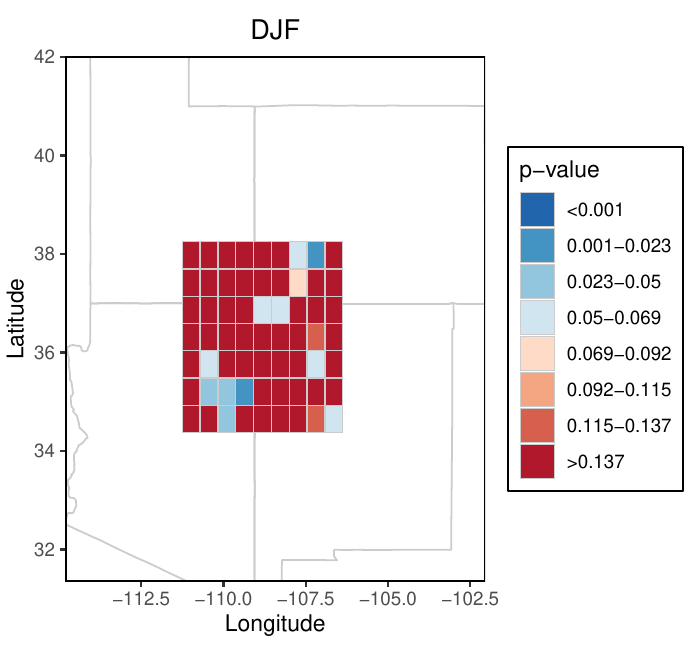}
    \caption{Max-stability diagnostics for NLDAS surface temperature seasonal maxima using the test of~\citet{koh2024space}. \emph{Top row:} $p$-values when the test is applied to the entire Four Corners domain. \emph{Middle and bottom rows:} Finer diagnostics based on 12 and 63 rectangular subregions, respectively. Darker blue shades correspond to smaller $p$-values, i.e., stronger evidence \emph{against} max-stability. The plots reveal widespread violations of max-stability for all seasons except winter (DJF), where the test more often fails to reject in small subregions, likely reflecting sparser station coverage.}
    \label{fig:max-stab}
\end{figure}

Figure~\ref{fig:max-stab} provides a more targeted assessment using the bootstrap max-stability test of~\citet{koh2024space}. The test exploits the homogeneity of the exponent measure: under max-stability, the statistic
\begin{equation*}
X_{\mathcal{S}} = \max_{d\in\mathcal{S}} \log X(s_d)
\end{equation*}
should follow a Gumbel distribution with location parameter $\mu_{\mathcal{S}} = \log V_{s_d,d\in\mathcal{S}}(1,\dots,1)$, where $V_{s_d,d\in\mathcal{S}}$ is the exponent measure over the testing set $\mathcal{S}$. Let $x_{d,m}^\star$ be the $m$th observed seasonal maximum at site $d$, and define $X_{\mathcal{S},m}^\star = \max_{d} \log x_{d,m}^\star$, $m=1,\dots,M$. The test statistic is then the Anderson--Darling distance between the empirical distribution of $\{X_{\mathcal{S},m}^\star\}_{m=1}^M$ and a Gumbel distribution with maximum-likelihood estimate $\hat{\mu}_{\mathcal{S}}$.

Applied to the entire Four Corners domain (top row of Figure~\ref{fig:max-stab}), the test rejects max-stability in all seasons at the coarsest spatial scale, showing strong evidence against the null hypothesis. When we refine the domain into 12 and then 63 rectangular subregions (middle and bottom rows), the picture becomes even more revealing: in JJA and SON, a large majority of subregions exhibit very small $p$-values, pointing to pervasive violations of max-stability. For MAM and especially DJF, the test more frequently fails to reject in smaller subregions, suggesting that a max-stable characterization may be tenable locally in winter but not across the full domain. Overall, these diagnostics highlight how quickly the max-stability assumption breaks down once spatial heterogeneity and subasymptotic behavior are taken seriously.

Figure~\ref{fig:max-stab}, together with the $\chi$-plots and posterior summaries of $\alpha$, underscores the need for models that can accommodate \emph{both} asymptotic dependence and asymptotic independence, and that allow the strength of extremal dependence to vary across space and seasons. Purely max-stable models are unable to capture the observed weakening of dependence at high but finite thresholds~\citep[see also][]{huser2025modeling}, while models that are everywhere asymptotically independent (e.g., Gaussian-based extremes) cannot reproduce the pockets of strong dependence that appear in winter and in some smaller regions. The LRSM construction in~\eqref{eqn:lrsm}, with the parameter $\alpha$ smoothly traversing the two regimes, provides the necessary flexibility, letting the data inform which extremal dependence class is appropriate in different parts of the domain. It is particularly reassuring that, for the winter season---where the max-stability tests largely fail to reject---the posterior mean of $\alpha$ exceeds $0.9$ in both Table~\ref{tab:data_results_alpha1} and Table~\ref{tab:data_results_alpha2}, confirming that our model successfully recovers the relevant strong asymptotic dependence.

\section{Theoretical Properties of the L\'evy Random Scale-Mixture Process}\label{sm:chi_eta}

\paragraph{Extremal dependence properties}
Here we derive the values of $\chi(\bs,\bs')$ and $\eta(\bs,\bs')$ for model~\eqref{eqn:lrsm}, as presented in Section~\ref{s:modeling_framework} of the main text.
\begin{theorem}[Extremal dependence of the L\'evy Random Scale Mixture]
\label{thm:chi_eta}
Let
\[
\{X(\bs):\bs \in\mathcal{S}\}=R^{\alpha} \{g(Z(\bs)):\bs \in \mathcal{S}\},
\]
where $\alpha \geq 0$, and where $R\sim \text{L\'{e}vy}(m,c),m\in \mbR, c>0$, is independent of the Gaussian process $Z(\cdot) \sim \mathcal{GP}\bigl(0, C_{\boldsymbol{\theta}}(\cdot, \cdot)\bigr)$ which has unit variance and correlation function $C_{\boldsymbol{\theta}}(\cdot, \cdot)$. Then, for any pair of distinct locations $(\bs,\bs')$:

\begin{enumerate}
\item If $\alpha \ge c$, the process $X$ exhibits asymptotic dependence with
\[
\chi(\bs,\bs')
=
=\mathbb{E}\left[\frac{\min\{g^{c/\alpha}(Z(\bs)), g^{c/\alpha}(Z(\bs'))\}}{\mathbb{E}[g^{c/\alpha}(Z(\bs))]}\right],
\qquad
\eta(\bs,\bs')=1.
\]

\item If $\alpha < c$, the process $X$ is asymptotically independent, with
\[
\chi(\bs,\bs')=0,
\]
and
\[
\eta(\bs,\bs')
=
\begin{cases}
\eta_Z(\bs,\bs'), & \text{if } \eta_Z(\bs,\bs')>\alpha/c,\\
\alpha/c, & \text{otherwise},
\end{cases}
\]
where $\eta_Z(\bs,\bs')=(1+C_{\boldsymbol{\theta}}(\bs,\bs'))/2$.
\end{enumerate}
\end{theorem}

\begin{proof}
First, note that a $ \text{L\'{e}vy}(m,c)$ random variable is also Stable($1/2,1,c,m$), and so $R\in \mathrm{RV}_{-c}$ where $X \in \mathrm{RV}_{\lambda}$ is a random variable that is regularly-varying with index $\lambda\in\mathbb{R}$ such that its survival function, $\bar{F}_X$, satisfies $\bar{F}_X(tx)/\bar{F}_X(t) \rightarrow x^\lambda$ as $t\rightarrow \infty$ for any $x >0$ \citep[Chapter 1.5][]{nolan2020univariate}. Noting that $R\in \mbox{RV}_{-c}$ implies that $R^\alpha\in \mbox{RV}_{-c/\alpha}$ and that $g\{Z(\bs)\}\in \mbox{RV}_{-1}$ for all $\bs$, the stated values of $\chi(\bs, \bs')$ and $\eta(\bs,\bs')$ follow from arguments provided by \citet{engelke2019extremal}. 
\end{proof}

\paragraph{Marginal properties of $X$:}
Recall that
\[
X = R^{\alpha}\left[\frac{1}{1-\Phi(Z)} - 1\right],
\]
where $Z \sim \mathcal{N}(0,1)$, $R$ is independent of $Z$, and $R$ follows a L\'evy$(m,c)$ distribution with density $f_R$.

\subsection*{Marginal distribution function}
\noindent The distribution function of $X$ has the form
\begin{align}\label{EQ:CDF_X}
\begin{split}
        F_X(x; \alpha)=\Pr(X<x)&= \Pr\Big(R^{\alpha}\Big[\frac{1}{1-\Phi(Z)}-1\Big]<x\Big)\\
        & = \Pr\left(\Big[\frac{1}{1-\Phi(Z)}-1\Big]<xR^{-\alpha}\right)\\
        & = \int_{0}^{\infty}\Pr\Big(\Big[\frac{1}{1-\Phi(Z)}-1\Big]<xr^{-\alpha}\Big)f_R(r)\mathrm{d}r\\
                & = \int_{0}^{\infty}\Pr\Big(\Big[\frac{1}{1-\Phi(Z)}\Big]<1+xr^{-\alpha}\Big)f_R(r)\mathrm{d}r\\
       &  = \int_{0}^{\infty}\Big[1-\frac{1}{1+xr^{-\alpha}}\Big]f_R(r)\mathrm{d}r\\
        & = 1-\int_{0}^{\infty}\frac{f_R(r)}{1+xr^{-\alpha}}\mathrm{d}r,
        \end{split}
\end{align}
where $f_R$ denotes the L\'evy($m$, $c$) density function. When $c = 1/2$ and $m = 0$, the L\'evy density reduces to
\[
f_R(r)
=
(4\pi r^3)^{-1/2}
\exp\!\left(-\frac{1}{4r}\right),
\qquad r>0.
\]
\subsection*{Marginal density}
Differentiating \eqref{EQ:CDF_X} with respect to $X$ and applying the Leibniz integral rule gives the marginal density
\begin{align}\label{EQ:MarginalDensityX}
\begin{split}
        f_{X}(x; \alpha) & = \frac{\partial }{\partial x}\Bigg[1-\int_{0}^{\infty}\frac{f_R(r)}{1+xr^{-\alpha}}\mathrm{d}r\Bigg] \\
        & = \int_{0}^{\infty} \frac{\partial }{\partial x} \frac{f_R(r)}{1+xr^{-\alpha}}\mathrm{d}r  \\
        & = \int_{0}^{\infty} \frac{f_R(r)}{r^{\alpha}(1+xr^{-\alpha})^2}\mathrm{d}r\\
        & = \int_{0}^{\infty} \frac{r^{\alpha}f_R(r)}{(x+r^{\alpha})^2}\mathrm{d}r.
        \end{split}
\end{align}

\section{Likelihood Function of the LRSM}
\label{sm:likelihood}

This section details the likelihood that is used for Bayesian inference of the LRSM model, and supports Section \ref{sec:bayes_inference} of the main text.

For each replicate $t$, the transformation from the random variable  $U_t(\bs_i)$ to $Z_t(\bs_i)$ is applied componentwise according to
\begin{equation}\label{e:z_transformation}
Z_t(\bs_i)
=
h^{-1}\!\left(U_t(\bs_i); R_t,\alpha \right),
\qquad i=1,\ldots,n, 
\end{equation}
\noindent where 
\begin{align}\label{EQ:Inv_H_function}
    h^{-1}(\cdot;R_t,\alpha)&=\Phi^{-1}\Bigg(1-\frac{R_t^{\alpha}}{F^{-1}_{X}(\cdot; \alpha)+R_t^{\alpha}}\Bigg).
\end{align}
Similarly, the transformation from the random variable  $Z_t(\bs_i)$ to $U_t(\bs_i)$ is 
\begin{equation}\label{e:u_transformation}
U_t(\bs_i)
=
h\!\left(Z_t(\bs_i); R_t,\alpha \right),
\qquad i=1,\ldots,n , 
\end{equation}
where 
\begin{align}\label{EQ:H_function}
h(\cdot;R_t,\alpha)&=F_X\left\{R_t^\alpha \left(\frac{1}{1-\Phi(\cdot)}-1\right); \alpha\right\}.
\end{align}
The corresponding transformations for the observations are $z_t(\bs_i)=h^{-1}\!\left(u_t(\bs_i);R_t,\alpha\right)$ and $u_t(\bs_i)=h\!\left(z_t(\bs_i);R_t,\alpha\right)$. We define the vector of marginal quantile 
transformations as
\[
F_X^{-1}(\bu_t; \alpha)
:=
\bigl(
F_X^{-1}(u_{1,t}; \alpha),
\ldots,
F_X^{-1}(u_{n,t}; \alpha)
\bigr)^\top.
\]

As a result of~\eqref{e:z_transformation}, the conditional density of \(\bU\) given \(\alpha\), \(\bR\), and \(\btheta\) follows from the multivariate change-of-variables theorem, with
\begin{align}\label{Eq:pdf_U}
f_{\bU}(\bu \mid \bR, \alpha, \boldsymbol{\theta})
&=
\prod_{t=1}^{T}
f_{\bZ}\!\bigl(h^{-1}\!(F_{X}^{-1}(\bu_t;\alpha); R_t,\alpha \bigr);\btheta\bigr)
\left|
\frac{\partial \bz_t}{\partial \bu_t}
\right| \nonumber \\
&=
(2\pi)^{-nT/2}
|\bC_{\boldsymbol{\theta}}|^{-T/2}
\exp\!\left\{
-\frac{1}{2}
\sum_{t=1}^{T}
h^{-1}\!\bigl(F^{-1}_{X}(\bu_t;\alpha); R_t, \alpha\bigr)^\top
\bC_{\boldsymbol{\theta}}^{-1}
h^{-1}\!\bigl(F^{-1}_{X}(\bu_t;\alpha);R_t,\alpha\bigr)
\right\}
\nonumber \\
&\quad \times
\prod_{t=1}^{T}
\prod_{i=1}^{n}
\left(
\frac{1}{
\phi\!\left(
\Phi^{-1}\!\left(
1 - \frac{R_t^{\alpha}}{
F^{-1}_{X}(u_{i,t};\alpha) + R_t^{\alpha}
}
\right)
\right)
}
\frac{R_t^{\alpha}}{
\bigl(F^{-1}_{X}(u_{i,t};\alpha) + R_t^{\alpha}\bigr)^2
}
\frac{1}{
f_{X}\!\bigl(F^{-1}_{X}(u_{i,t};\alpha);\alpha\bigr)
}
\right),
\end{align}
\noindent where $f_{\bZ}(\cdot;\btheta)$ denotes the multivariate Gaussian density with covariance matrix $\bC_{\boldsymbol{\theta}}$ derived from the Mat\'ern covariance function in~\eqref{eq:MATERN} with parameter vector $\boldsymbol{\theta}$, and $\phi$ and $\Phi$ represent the standard normal density and distribution functions, respectively.

\paragraph{Deriving the Jacobian $\left|
\frac{\partial \bz_t}{\partial \bu_t}
\right| $:} 

Recall that $Z_t(\bs) = h^{-1}\!\left(U_t(\bs);R_t,\alpha\right)$ in \eqref{EQ:Inv_H_function} and, by construction,
\[
Z_t(\bs) \sim \mathcal{N}\!\left(\bzero, \bC_{\boldsymbol{\theta}}\right).
\]
The derivative of the inverse standard normal CDF is
\[
\frac{\partial}{\partial x}\Phi^{-1}(x)
=
\frac{1}{\phi\!\left(\Phi^{-1}(x)\right)},
\]
where $\phi(\cdot)$ denotes the standard normal density. Next, note that for fixed $R_t$,
\[
\frac{\partial}{\partial y}
\left(
1 - \frac{R_t^{\alpha}}{y + R_t^{\alpha}}
\right)
=
\frac{R_t^{\alpha}}{(y + R_t^{\alpha})^{2}},
\]
and that the derivative of the inverse marginal distribution function satisfies
\[
\frac{\partial}{\partial u} F_X^{-1}(u;\alpha)
=
\frac{1}{f_X\!\left(F_X^{-1}(u;\alpha);\alpha\right)},
\]
where $f_X(\cdot;\alpha)$ is the marginal density defined in
(\ref{EQ:MarginalDensityX}).

Applying the chain rule, the Jacobian of the transformation from
$U_t(\bs)$ to $Z_t(\bs)$ is
\begin{equation}
\label{eq:jacobian}
\frac{\partial Z_t(\bs)}{\partial U_t(\bs)}
=
\frac{1}{
\phi\!\left(
\Phi^{-1}\!\left(
1 - \frac{R_t^{\alpha}}
{F_X^{-1}(U_t(\bs);\alpha) + R_t^{\alpha}}
\right)
\right)
}
\,
\frac{R_t^{\alpha}}
{\left(F_X^{-1}(U_t(\bs);\alpha) + R_t^{\alpha}\right)^{2}}
\,
\frac{1}
{f_X\!\left(F_X^{-1}(U_t(\bs);\alpha);\alpha\right)}.
\end{equation}

\section{Scalable Approaches for Approximating Gaussian Processes}\label{sec:scalable}
To address the computational burden of large spatial covariance matrices, we provide additional details on three widely used scalable strategies: Vecchia approximations, covariance tapering, and low-rank basis representations. This supports Section \ref{sec:bayes_inference} of the main text.

\paragraph{Vecchia approximation}
The Vecchia approximation \citep{vecchia1988estimation,stein2004approximating} is widely used in spatial statistics as a scalable method for estimation and prediction over large spatial fields. Consider the transformed spatial process 
$\{Z_t(\bs): \bs \in \mathcal{S}\}$ observed at locations $(\bs_1,\ldots,\bs_n)$, as in Equation \eqref{e:z_transformation}. 
The vector $\bz_t$ is modeled as a Mat\'ern Gaussian process with parameter vector 
$\boldsymbol{\theta}$. The joint density of $\bz_t$  can be factorized as
\begin{align}
f_{\bZ}(\bz_t; \boldsymbol{\theta})
&=
\prod_{i=1}^{n}
f\!\left(z_{i,t} \mid z_{<i,t}; \boldsymbol{\theta}\right),
\end{align}
where $z_{i,t} := Z_t(\bs_i)$ and
$z_{<i,t} := \bigl(z_{1,t}, \ldots, z_{i-1,t}\bigr)$
denotes the observations at the first $i-1$ locations under a fixed ordering. The first term in the product reduces to the marginal density $f(z_{1,t}; \boldsymbol{\theta})$. The Vecchia approximation restricts the conditioning set of $z_{i,t}$ to a subset 
$z_{(i,t)} \subset z_{<i,t}$, so that
\begin{align}\label{e:vecchia}
    f_{\bZ}(\bz_t; \boldsymbol{\theta}) &\approx \prod_{i=1}^n f(z_{i,t}|z_{(i,t)};\boldsymbol{\theta}). 
\end{align}

Equation \eqref{e:vecchia} defines a valid likelihood for any choice of ordering and conditioning sets, though approximation quality depends on these choices. We use a max-min ordering of locations which has been shown to perform well across a range of settings \citep{guinness2018_kl}. For each location $\bs_i$, the conditioning set $z_{(i,t)}$ contains the $m$ nearest neighbors among $\{z_{1,t},\ldots,z_{i-1,t}\}$ based on the chosen ordering (e.g., coordinate-based or maxmin), with $m \ll n$. While more complex formulations are possible, this approach has been shown to work well in practice \citep[see e.g.,][]{stein2004approximating,guinness2018_kl}. The Vecchia approximation has also been used for modeling spatial extremes processes which often have intractable likelihoods \citep{MAJUMDER2023100755, Huser02072024,majumder2024modeling}.

\paragraph{Covariance tapering}
Covariance tapering \citep{furrer2006covariance} approximates the dense covariance matrix 
$\bC_{\boldsymbol{\theta}}$ by the sparse matrix $\bC_{\boldsymbol{\theta},\psi}^\ast 
= \bC_{\boldsymbol{\theta}} \circ \bT_\psi$,
where $\circ$ denotes the Hadamard product. 
Let $h_{ij} = \|\bs_i - \bs_j\|$ be the Euclidean distance between $\bs_i$ and $\bs_j$. Given a compactly supported positive definite tapering function 
$t(\cdot,\psi) : [0,\infty) \to \mathbb{R}$ and cutoff distance $\psi$, the taper matrix is defined as
\[
\bT_\psi = \left[ t(h_{ij},\psi) \right]_{i,j=1}^n.
\]
Since $t(\cdot,\psi)$ is positive definite, $\bT_\psi$ is symmetric and positive definite, 
with $(\bT_\psi)_{ii} = 1$ and $(\bT_\psi)_{ij} = 0$ whenever $h_{ij} \ge \psi$ \citep{furrer2006covariance}. 
Thus, $\bT_\psi$ is a known compactly supported positive definite correlation matrix that smoothly decays to zero beyond a cutoff distance. Hence, the distant entries of $\bC_{\boldsymbol{\theta},\psi}^\ast$ are forced to vanish while positive definiteness is preserved. This sparse approximation reduces storage overhead and computational costs of Cholesky factorizations and matrix–vector products, allowing Gaussian process models to scale to larger datasets at the expense of approximating long-range dependence.

A widely used tapering function is the spherical taper. For two locations separated by distance $h$, the spherical taper is defined as
$$
t(h,\psi) = \begin{cases}
\left(1 - \frac{h}{\psi}\right)^2 \left(1 + \frac{h}{2\psi}\right), & 0 \le h < \psi, \\
0, & h \ge \psi,
\end{cases}
$$
where $\psi > 0$ is the taper range parameter. The function $t(h,\psi)$ is positive definite in $\mathbb{R}^d$ for dimensions $d \le 3$, and it smoothly decreases from one at $h = 0$ to zero at the cutoff distance $\psi$. $\bC_{\boldsymbol{\theta},\psi}^\ast$ preserves the local covariance structure while setting long‐range correlations, which are often near-negligible, to zero. Examples of tapering functions include those in \citet{gneiting2002compactly}, \citet{mitra2003polynomial}, and \citet{bolin2016spatially}.

\paragraph{Low-rank basis representation}


Low-rank basis representations \citep{cressie2008frk} provide an effective way to approximate large covariance matrices. 
The key idea is to represent spatial dependence using a reduced set of basis functions, 
thereby avoiding storage and manipulation of the full dense covariance matrix. 

Rather than working with the full $n \times n$ covariance matrix $\bC_{\boldsymbol{\theta}}$, 
low-rank models introduce $k \ll n$ basis functions 
$\{b_1(\bs),\ldots,b_k(\bs)\}$ and represent the process as
\[
Z_t(\bs_i) = \sum_{j=1}^{k} b_j(\bs_i)\,\delta_{j,t}+\epsilon_{j,t},
\]
or, in vector form,
\[
\bZ_t = \mathbf{B}\boldsymbol{\delta}_t +\boldsymbol{\epsilon}_t,
\]
where $\mathbf{B} = [\mathbf{b}_1 \ \cdots \ \mathbf{b}_k]$ is the $n \times k$ 
basis matrix with entries $B_{ij} = b_j(\bs_i)$,  $\boldsymbol{\delta}_t = (\delta_{1,t},\ldots,\delta_{k,t})^\top$ 
are latent coefficients, and $\boldsymbol{\epsilon}_t= (\epsilon_{1,t},\ldots,\epsilon_{n,t})^\top$ are the independent and identically distributed observational errors.

Assuming $\boldsymbol{\delta}_t \sim N(\mathbf{0}, \mathbf{K})$ and $\boldsymbol{\epsilon}_t \sim N(\mathbf{0}, \tau^2\mathbf{I})$, 
this representation induces the covariance structure
\[
\text{Cov}(\bz_t) = \mathbf{BKB}^\top+\tau^2\bI,
\]
where $\mathbf{K}$ is the $k \times k$ covariance matrix of the basis coefficients.

Since $k\ll n$, the resulting covariance has rank at most $k$, reducing the cost of matrix-vector multiplication, likelihood evaluation, and Cholesky factorizations to $\mathcal{O}(nk^2)$. There are many possible choices of basis functions, such as Wendland basis functions \citep{nychka2015multiresolution}, wavelets \citep{royle2005efficient}, bisquare basis functions \citep{sengupta2013hierarchical}, splines \citep{wahba1990spline}, or finite elements \citep{lindgren2011spde}. With a suitable basis, low-rank methods efficiently capture large-scale spatial dependence while substantially reducing computational cost.

\section{Neural Bayes estimation of the LRSM}
\label{sec:NBE_discuss}
We now describe likelihood-free Bayesian inference of the LRSM using neural Bayes estimators \citep{sainsbury2022fast}. Define the set of parametric probability distributions $\mathcal{P}$ on sample space $\Omega \subseteq \mathbb{R}^n$; $\mathcal{P}$ is parameterized by a parameter vector $\boldsymbol{\phi} :=(\alpha,\boldsymbol{\theta})$ such that ${\mathcal{P}:= \{F_{\mathbf{U}_t}(\cdot; \boldsymbol{\phi}):\boldsymbol{\phi}\in\Phi\}}$, where $\Phi \subseteq \mathbb{R}^p,p=|\boldsymbol{\phi}|,$ is the parameter space and $F_{\mathbf{U}_t}(\cdot; \boldsymbol{\phi})$ is the distribution function of $\mathbf{U}_t$, marginalized over the random scale $R_t$. Recall that $\mathbf{U}=(\mathbf{U}_1^\top,\dots,\mathbf{U}_T^\top)^\top$, where $\mathbf{U}_1,\dots,\mathbf{U}_T$ are mutually independent realizations from $F_{\mathbf{U}_t}(\cdot; \boldsymbol{\phi})$, and $\mathbf{U}_t=(U_t(\bs_1),\dots,U_t(\bs_n))^\top$, {$t=1,\dots,T,$} for sites $\{\bs_1,\dots,\bs_n \}$. We then define the point estimator $\widehat{\boldsymbol{\phi}}:\Omega^T\mapsto \Phi$. To assess the quality of a constructed estimator $\widehat{\boldsymbol{\phi}}(\mathbf{U})$ for a single parameter vector $\boldsymbol{\phi}$ and dataset $\mathbf{U}$, we require a loss function $L: \Phi \times \Phi \mapsto \mathbb{R}_{\geq 0}$. The Bayes' risk is then the integrated posterior expected loss:
\begin{equation}
\label{EQ:Bayes_risk}
r(\hat{\boldsymbol{\phi}}(\cdot)):=\int_\Phi \int_{\Omega^T} L(\boldsymbol{\phi},\widehat{\boldsymbol{\phi}}(\mathbf{U}))f(\mathbf{U}\mid\boldsymbol{\phi})\mathrm{d}\mathbf{U}\mathrm{d}\Pi(\boldsymbol{\phi}),
\end{equation}
where $\Pi(\cdot)$ is some prior measure on $\boldsymbol{\phi}\in\Phi$. A minimizer of \eqref{EQ:Bayes_risk} is a Bayes estimator with respect to $L(\cdot,\cdot)$ and $\Pi(\cdot)$, and has attractive theoretical properties under mild regularity conditions; for details, see \citet{lehmann2006theory}. The prior measure $\Pi(\cdot)$ is taken to be equivalent to the prior used by all other competing inference methods. Note that the smoothness parameter $\nu$ is fixed during construction and training of the estimator, and separate estimators are built for different choices of $\nu$; thus, $\Phi = \mathbb{R}_{>0}\times \mathbb{R}_{\geq 0}$.

Approximating the mapping $\widehat{\boldsymbol{\phi}}(\cdot)$ using neural networks is referred to as neural Bayes estimation; see, e.g., \citet{sainsbury2022fast}. We implement a neural Bayes estimator using graph neural networks via the \texttt{Julia} \citep{bezanson2017julia} package NeuralEstimators.jl \citep{sainsbury2022fast},
with the smoothness parameter $\nu$ treated as fixed. Thus, we construct and train four estimators for $\boldsymbol{\phi}=(\alpha,\rho)'$: two each for the marginal posterior medians and marginal posterior $95\%$ credible intervals, under two values of $\nu$: $\nu=0.5$ and $\nu=1.5$. Estimators use the architecture as described in \citet{sainsbury2023neural} and are trained for up to 72 hours on a NVIDIA A100 GPU. Training of the estimators is performed by minimizing an empirical estimate of the Bayes Risk \eqref{EQ:Bayes_risk}, with the sample size $T$ and spatial  locations treated as random. We place a discrete $\rm{Uniform}\{5,120\}$ prior on $T$ and, for the spatial locations $\{\bs_1,\dots,\bs_n\}$, we use a Mat\'ern cluster process prior with the expected number of locations $n$ taken to be 500 \citep{baddeley2016spatial}. 
Once trained, the estimator can be queried for negligible computational cost and is amortized with respect to the number of samples, $T$, and the number $n$ and configuration of the sample locations. That is, we train the estimator only once and then reuse it throughout the application and simulation study.

\section{Implementation Details for Simulation Studies}\label{sec:implementation}
This section provides computational details for the simulation studies presented in Section~\ref{s:simulation_study} of the main text. It both lists the hardware used and details the settings for each study.

Computation was distributed across four high-performance computing (HPC) systems hosted at different institutions. To mitigate systematic differences in hardware and environments, we ensured that each HPC was used for every scenario in the low- and high-dimensional cases.  In the low-dimensional setting (20 scenarios), each HPC performed 25 replicates per scenario, yielding 500 model fits per HPC. In the high-dimensional setting (8 scenarios), we employed three HPCs and tasks were similarly distributed across each scenario. Reported walltimes are averaged across HPC systems to account for system-specific differences. The following HPC resources were used in this study:
\begin{enumerate}
    \item \href{https://orc.gmu.edu}{HOPPER} (GMU): single 2.4 GHz Intel Xeon Gold 6240R processors using 30 cores on each node.
    \item \href{https://doit.missouri.edu/services/software-directory/hellbender/}{Hellbender} (Missouri): dual 2.0 GHz AMD EPYC 7713 processors with 512 GB RAM, 64 cores per socket and 128 cores per node.
    \item \href{https://hpc.uark.edu/}{AHPCC} (UArk): 2.1 GHz Intel Xeon Gold 6130 processors, using 30/32 cores on each node.
    \item \href{https://digitalresearchservices.ed.ac.uk/resources/eddie}{Eddie} (Edinburgh): 20 cores, on a randomly selected node comprising Intel Xeon Gold (2--2.8 GHz) processors - for full details, see \href{https://www.wiki.ed.ac.uk/spaces/ResearchServices/pages/296793802/Memory+Specification}{https://tinyurl.com/2667mj58}.
\end{enumerate}
Training and estimation for the neural Bayes estimator was also conducted on Eddie, but with a single NVIDIA A100 GPU.

For the low-rank approach, we used $100$ and $500$ eigenbasis functions obtained from the leading eigenvectors of a pre-specified Mat\'ern covariance matrix with parameters $\nu=0.5$ and $\rho=0.5$, selected based on inspection of the data. The configuration with $100$ basis functions represents a strongly dimension-reduced model emphasizing smooth, global spatial structure, whereas the $500$-basis configuration allows for finer spatial resolution and more localized variation. The basis functions were row-normalized to have unit $\ell_2$ norm, ensuring that the outer-product covariance matrix has unit diagonal entries. Consequently, the Gaussian process ${Z(\bs)}$ has unit marginal variance.

In the Vecchia approximations, we considered conditioning set of sizes $m \in \{3,5,10,15\}$ to assess the impact of set size (i.e., number of neighbors) on model performance. The maximum value $m=15$ was imposed due to computational constraints.

For covariance tapering, we used the spherical taper function with tapering thresholds $\psi \in \{0.2, 0.75\}$, corresponding to approximately $90\%$ and $20\%$ sparsity in the tapered covariance matrix, respectively. These settings allow us to evaluate how differing levels of tapered sparsity affects model performance and computational costs.

As the L\'evy Random Scale Mixture process does not admit a closed-form conditional distribution, and the Neural Bayes estimator provides point estimates of posterior medians and credible intervals, the latter does not support out-of-sample prediction. Thus, for the NBE we report only inference metrics (coverage and interval scores) and walltimes. 

For the Bayesian hierarchical models, parameter inference and predictions proceeded via posterior sampling via the Metropolis-Hastings algorithm. The sampling algorithm was run for 50{,}000 iterations. Convergence was assessed by monitoring the batch-means standard errors \citep{jones2006fixed} over the course of the algorithm. To improve mixing, parameters were updated using a one-at-a-time Metropolis scheme, which resulted in faster mixing than block updates. Proposal distributions were tuned using the Log-Adaptive Proposal (LAP) strategy of \citet{shaby2010exploring}, with adaptation performed every 200 iterations. All implementations and variants of the LSRM (i.e., full Gaussian process, low-rank approximation, Vecchia approximation, and covariance tapering) were fitted using the same posterior sampling scheme and MCMC tuning protocol.

\section{Validation Metrics for Simulation Studies}\label{sec:validation_metrics}
This section details the validation metrics that were used to compare results across different methods and different simulation settings in Section \ref{s:simulation_study} of the main text.

For inference in our simulation studies, we report the empirical coverage and interval scores \citep{gneiting2007strictly} for $\alpha$ and $\rho$. We consider two tail-weighted CRPS (twCRPS) measures \citep{gneiting2014probabilistic} and a smooth variant that emphasize different regions of the predictive distribution. We evaluate (i) a lower-tail twCRPS using the empirical median of the training data as a cutoff; (ii) a weighted CRPS with a smooth Gaussian weight function centered at the training-data median; and (iii) an extreme-tail twCRPS using the empirical 80th percentile of the posterior predictive distribution as the cutoff. 

\paragraph{Empirical coverage}
Empirical coverage reports the frequency of when the Bayesian credible intervals contain the true parameter value $\theta_0$ across replicate simulations. Let $[l_j, u_j]$ denote the $(1-\alpha^*)$ credible interval obtained from the $j$th experiment, for $j=1,\dots,B$. The empirical coverage probability is defined as
\begin{align}
\widehat{\mathrm{Cov}}
=
\frac{1}{B}
\sum_{j=1}^{B}
\mathbf{1}\{ \theta_0 \in [l_j, u_j] \},
\end{align}
where $\mathbf{1}\{\cdot\}$ is the indicator function. Values of $\widehat{\mathrm{Cov}}$ close to the nominal level $(1-\alpha^*)$ suggest that posterior uncertainty is well calibrated, but should be examined alongside predictive accuracy (e.g., proper scoring rules) and interval precision (e.g., interval scores).
\paragraph{Interval score}
The interval score is a proper scoring rule \citep{gneiting2007strictly} that rewards narrow intervals while penalizing the loss of coverage. Let $[l, u]$ denote a $(1-\alpha^*)$ credible interval for a scalar parameter with true value $\theta_0$. The interval score is defined as
\begin{align}
\mbox{IS}_{\alpha^*}(l,u;\theta_0)
=
(u-l)
+ \frac{2}{\alpha^*}(l-\theta_0)\mathbf{1}\{\theta_0 < l\}
+ \frac{2}{\alpha^*}(\theta_0-u)\mathbf{1}\{\theta_0 > u\},
\end{align}
where $\mathbf{1}\{\cdot\}$ is the indicator function. 

For each simulation scenario, we compute the interval score across the $B=100$ repeated experiments and report the weighted average
\begin{align}
\overline{\mbox{IS}}_{\alpha^*}
=
\frac{\alpha^*}{2}
\cdot
\frac{1}{B}
\sum_{j=1}^{B}
\mbox{IS}_{\alpha^*}(l_j,u_j;\theta_0).
\end{align} Lower interval scores are preferred, with near-zero values indicating narrow intervals that cover the true parameter.

\paragraph{Tail-weighted continuous rank probability scores}
Predictive performance was evaluated using three variants of the continuous ranked probability score (CRPS) \citep{matheson1976scoring,gneiting2007strictly}. Let $X$ denote the observed value and let $F$ denote the posterior predictive distribution obtained from the posterior samples $\{X^{(m)}\}_{m=1}^{M}$. The CRPS is defined as
\[
\mathrm{CRPS}(F,X)
=
\int_{-\infty}^{\infty}
\big(F(z)-\mathbf{1}\{z \ge X\}\big)^2 dz.
\]

To emphasize extremal regions of the predictive distribution, we consider weighted extensions of the CRPS. The tail-weighted CRPS (twCRPS) \citep{gneiting2014probabilistic} is 

\[
\mathrm{twCRPS}(F,X; w)
=
\int_{-\infty}^{\infty}
w(z)\,\big(F(z)-\mathbf{1}\{z \ge X\}\big)^2 \, dz,
\]
with weights $w(z)$.
\medskip
\noindent
\textbf{(i) Lower-tail twCRPS.}  
We compute a threshold-weighted CRPS using the empirical median of the training data, denoted $a$. The score places greater emphasis on predictive accuracy below the cutoff via the indicator weight
\[
w_1(z) = \mathbf{1}\{z \le a\}.
\]

\medskip
\noindent
\textbf{(ii) Smooth weighted CRPS.}  
To avoid the discontinuity induced by indicator weights, we also consider a smoothly varying Gaussian weight centered at the training-data median. Let $\mu$ and $\sigma$ denote the empirical mean and standard deviation of the training data, respectively. The weight function is
\[
w_2(z) = \Phi\!\left(\frac{z-\mu}{\sigma}\right),
\]
where $\Phi(\cdot)$ is the standard normal cumulative distribution function. This formulation progressively increases the contribution of larger outcomes while retaining sensitivity across the distribution.

\medskip
\noindent
\textbf{(iii) Extreme-tail twCRPS.}  
Finally, we evaluate predictive performance in the upper tail using a threshold-weighted CRPS with cutoff $a_{0.8}$, defined as the empirical 80th percentile of the posterior predictive samples. The corresponding weight function is
\[
w_3(z) = \mathbf{1}\{z \ge a_{0.8}\},
\]
thereby prioritizing accuracy for extreme outcomes.

\medskip
For each method, scores are averaged across test observations, with lower values corresponding to better predictive performance.

\section{Simulation Study under Model Misspecification}\label{sec:ms_sim_study}
We conduct an additional simulation study to assess the performance of the proposed LRSM approaches on data generated from models outside the LRSM framework. Specifically, we fit the model to data generated from max-stable (MSP) (Section~\ref{s:preliminaries}) and inverted max-stable processes \citep[IMSP,][]{wadsworth2012dependence}, representing asymptotic dependence and independence, respectively.

\subsection{Simulation design}
Similar to the low-dimensional simulation study, we randomly generate $n=500$ irregularly spaced locations, $\bs_1,\dots,\bs_n,$ within the unit square $\mathcal{S}= [0,1]^2$, with $T=100$ independent replicates. We allocate $n_{\text{test}}=125$ additional locations (with the same $T$) for out-of-sample evaluation. To assess robustness of the LRSM to data from different models, we generate asymptotically dependent  data from a MSP and asymptotically independent data from an IMSP using the \texttt{SpatialExtremes} function \texttt{rmaxstab()}. For both settings, we simulate max-stable processes using the Brown-Resnick model \citep{kabluchko2009stationary} with a power variogram of the form $\gamma(h) = (h/\rho)^{\nu}$, where the range parameter is set to $\rho=0.25$ and the smoothness parameter to $\nu=0.5$. For each process and configuration, we generate 100 replicate datasets. The inference procedures and validation metrics follow those used in the other simulation studies. Computation for the likelihood-based approaches was performed on GMU's HOPPER cluster, while the NBE computations were carried out separately.

\subsection{Results}
Most methods perform well under model misspecification. The full Gaussian process, Vecchia approximations, and tapering methods recover the correct extremal dependence regime across both the MSP and IMSP, as illustrated in Figures~\ref{fig:ms_inv_summary}. In contrast, the low-rank approaches misrepresent the dependence structure for the asymptotically independent case and exhibit noticeably poorer predictive performance.
\begin{figure}[ht]
\centering

\begin{subfigure}[t]{0.48\textwidth}
\centering
\includegraphics[width=\textwidth]{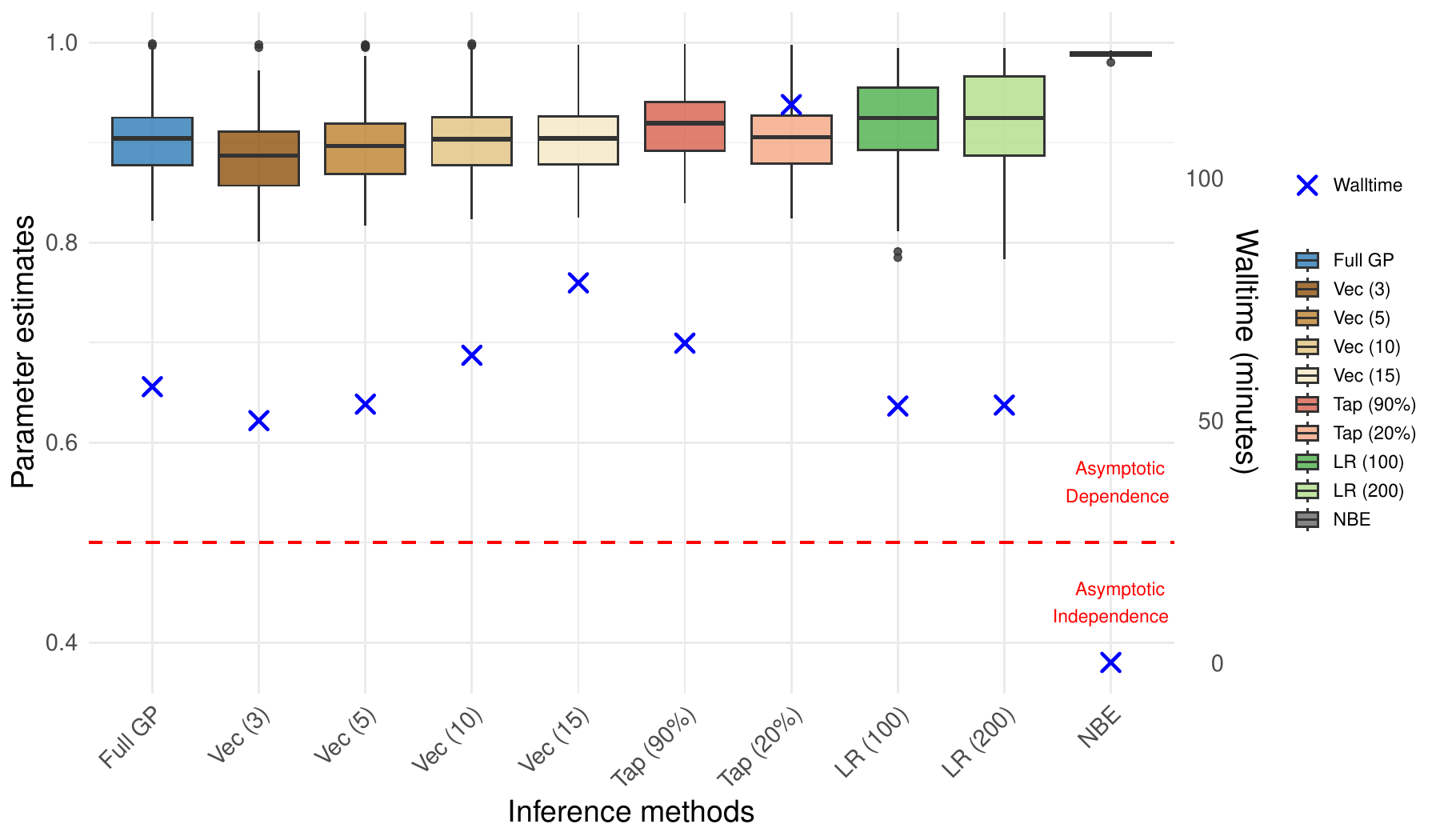}
\caption{Max-stable (AD)}
\end{subfigure}
\hfill
\begin{subfigure}[t]{0.48\textwidth}
\centering
\includegraphics[width=\textwidth]{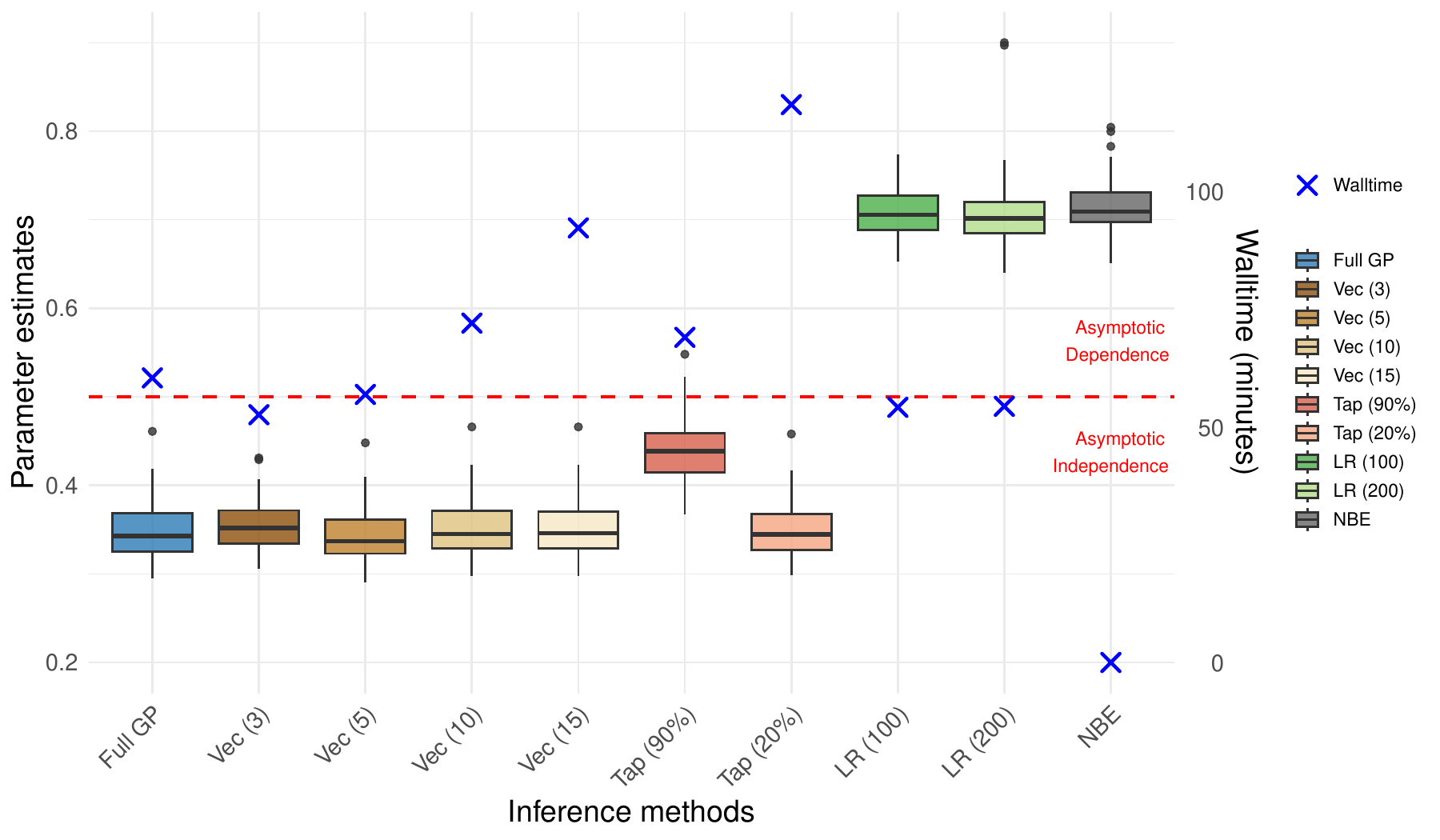}
\caption{Inverted max-stable (AI)}
\end{subfigure}

\caption{
Boxplots of posterior median estimates of $\alpha$ and walltime across inference methods for the max-stable (asymptotic dependence) and inverted max-stable (asymptotic independence) models. The dashed red line at 0.5 denotes the boundary between asymptotic dependence and asymptotic independence. Walltime (in minutes) is overlaid as blue X's using the secondary axis.
}
\label{fig:ms_inv_summary}
\end{figure}

Tables~\ref{tab:sim-ms} and \ref{tab:sim-inv} show that, aside from the low-rank methods, all approaches achieve similar predictive performance on the held-out validation set, with negligible differences in twCRPS. Among these, Vecchia approximations offer the best trade-off between accuracy and computational efficiency, with moderate conditioning set sizes ($m$) achieving performance comparable to the full Gaussian process though at lower cost.

As expected, the NBE encounters some difficulty under model misspecification. NBE preserves the relative ordering of $\alpha$, with MSP estimates exceeding those for the IMSP. However, these estimates do not clearly distinguish between asymptotic dependence and independence, as seen in Figures~\ref{fig:ms_inv_summary}. In particular, the IMSP estimates exceed the dependence threshold, while the MSP estimates lie near the boundary of the parameter space.

\begin{table}[!htbp]
\centering
\caption{Tail-weighted CRPS (twCRPS) and walltime for the max-stable process. The max-stable model corresponds to asymptotic dependence, while the inverted max-stable model captures asymptotic independence.} 
\label{tab:sim-ms}
\begin{tabular}{lrrrr}
   & \multicolumn{3}{c}{twCRPS} & Walltime \\ 
\cline{2-4}
 & 1 & 2 & 3 & (min) \\ \hline
 Full GP & 0.033 & 0.083 & 0.006 & 57.04 \\ 
  Vecchia (M = 3) & 0.034 & 0.087 & 0.007 & 50.04 \\ 
  Vecchia (M = 5) & 0.033 & 0.085 & 0.007 & 53.44 \\ 
  Vecchia (M = 10) & 0.033 & 0.083 & 0.006 & 63.53 \\ 
  Vecchia (M = 15) & 0.033 & 0.083 & 0.006 & 78.47 \\ 
  Taper (90\% sparse) & 0.033 & 0.085 & 0.006 & 66.03 \\ 
  Taper (20\% sparse) & 0.033 & 0.084 & 0.006 & 115.25 \\ 
  Low Rank (p = 100) & 0.050 & 0.132 & 0.012 & 53.05 \\ 
  Low Rank (p = 200) & 0.054 & 0.139 & 0.011 & 53.24 \\ 
  NBE &  &  &  & 0.14 \\ 
   \hline
\end{tabular}

\vspace{1cm}
\caption{Tail-weighted CRPS (twCRPS) and walltime for the inverted max-stable process. The max-stable model corresponds to asymptotic dependence, while the inverted max-stable model captures asymptotic independence.} 
\label{tab:sim-inv}
\begin{tabular}{lrrrr}
   & \multicolumn{3}{c}{twCRPS} & Walltime \\ 
\cline{2-4}
 & 1 & 2 & 3 & (min) \\ \hline
 Full GP & 0.042 & 0.082 & 0.010 & 60.58 \\ 
  Vecchia (M = 3) & 0.044 & 0.086 & 0.011 & 52.75 \\ 
  Vecchia (M = 5) & 0.043 & 0.083 & 0.010 & 57.06 \\ 
  Vecchia (M = 10) & 0.042 & 0.082 & 0.010 & 72.22 \\ 
  Vecchia (M = 15) & 0.042 & 0.082 & 0.010 & 92.44 \\ 
  Taper (90\% sparse) & 0.043 & 0.085 & 0.010 & 69.20 \\ 
  Taper (20\% sparse) & 0.042 & 0.082 & 0.010 & 118.63 \\ 
  Low Rank (p = 100) & 0.063 & 0.122 & 0.017 & 54.33 \\ 
  Low Rank (p = 200) & 0.068 & 0.130 & 0.017 & 54.54 \\ 
  NBE &  &  &  & 0.14 \\ 
   \hline
\end{tabular}
\end{table}

\section{Simulation Design for Low- and High-Dimensional Settings}\label{sec:simulation_design_supplement}
Below, we provide additional details on the simulation designs for both low- and high-dimensional settings of the Simulation studies presented in Section \ref{s:simulation_study} of the main text. The full design, including parameter settings and dataset configurations, is summarized in Table~\ref{tab:design_summary}.

\paragraph{Low-dimensional setting}
In the low-dimensional setting, we used $n=500$ locations with $T=100$ temporal replicates for training and reserved $n_{\text{test}}=125$ locations (with the same $T$) for testing. We consider two Mat\'ern smoothness values, $\nu \in \{0.5, 1.5\},$ and two true range parameter values. The range values are chosen so that the corresponding effective range is $0.15$ and $0.75$, where the effective range is the distance at which the correlation decays to $0.05$. This specification yields range parameters of $\rho\in\{0.05,0.25\}$ for $\nu=0.5$, and $\rho\in\{0.0316,0.1581\}$ for $\nu=1.5$. The choice $\nu=0.5$ reflects the best-fitting smoothness estimate from the NLDAS application (see Section~\ref{s:exploratory}), while $\nu=1.5$ provides a smoother counterpart for comparison. The strength and class of extremal dependence is varied across five levels, $\alpha\in \{0.05, 0.3, 0.45, 0.55, 0.7\}$, spanning asymptotic independence through to asymptotic dependence. In total, this provides $20$ combinations of $\nu$, $\rho$, and $\alpha$. 

\paragraph{High-dimensional setting}
The large-scale simulation study increases the numbers of spatial locations and temporal replicates to assess how the proposed methods scale computationally. Consistent with the exploratory analysis (Section~\ref{s:exploratory}), we fix $\nu = 0.5$ and $\rho = 0.05$ to reflect the dependence structure of our NLDAS dataset. We further consider two values of: $\alpha \in \{0.3, 0.7\}$, corresponding to asymptotic independence and dependence, respectively. Then, we vary the training dataset dimensions: $n \in \{1000, 5000\}$ locations and $T \in \{10, 50\}$ time replicates.  The corresponding test sets use $n_{\text{test}} \in \{150, 750\}$ locations with the same $T$. The eight combinations of spatial locations, temporal replicates, and $\alpha$ values. Each dataset was analyzed using the same competing methods as in the low-dimensional study, except when $n=5000$, where the full Gaussian process and covariance tapering were omitted due to computational cost.

\begin{table}[!htbp]
\small
\centering
\caption{Simulation design summary. 
The low-dimensional study varies smoothness $(\nu)$, range $(\rho)$, and extremal dependence $(\alpha)$. The high-dimensional study varies the number of spatial locations $(n)$, temporal replicates $(T)$, and $\alpha$.}
\small
\begin{tabular}{lcccc}
\toprule
\multicolumn{5}{c}{\textbf{Baseline/Low-Dimensional Settings}} \\[2pt]
\midrule
$n$ & $T$ & $\nu$ & Range $(\rho)$ & $\alpha$ \\
\midrule
500 & 100 & 0.5 & (0.05, 0.25) & $\{0.05,0.30,0.45,0.55,0.70\}$ \\
500 & 100 & 1.5 & (0.0316, 0.1581) & $\{0.05,0.30,0.45,0.55,0.70\}$ \\[2pt]
\addlinespace
\multicolumn{5}{l}{\textit{Total:} 20 unique combinations and 2,000 simulated datasets.} \\[6pt]
\midrule
\multicolumn{5}{c}{\textbf{Scalabliity/High-Dimensional Settings}} \\[2pt]
\midrule
$n$ & $T$ & $\nu$ & Range $(\rho)$ & $\alpha$ \\
\midrule
$\{1000, 5000\}$ & $\{10, 50\}$ & 0.5 & 0.05 & $\{0.3, 0.7\}$ \\[2pt]
\addlinespace
\multicolumn{5}{l}{\textit{Total:} 8 unique combinations and 800 simulated datasets.} \\
\bottomrule
\end{tabular}
\normalsize
\label{tab:design_summary}
\end{table}

\section{Low-dimensional Case: Simulation Study Results}
This section summarizes the results of the low-dimensional simulation study across multiple of dependence, smoothness, and range parameter settings. Tables and figures are provided in their corresponding subsections. It supports Section \ref{s:sim_results} in the main text.

\subsection{Tables} \label{tab:lowdim_results}
Tables~\ref{tab:sim-results-1}--\ref{tab:sim-results-20} report performance comparisons among the Full GP, Vecchia approximations, covariance tapering, low-rank methods, and the NBE approach. For each configuration, we present empirical coverage and average interval scores for the dependence parameter $\alpha$ and range parameter $\rho$, along with tail-weighted continuous ranked probability scores (twCRPS) and computational walltime. 

\begin{table}[htbp]
\small
\centering
\caption{Simulation study results for the low-dimensional case with dependence $\alpha=0.05$, smoothness $\nu=0.5$, and range $\rho=0.05$. We report the coverage and average interval scores for $\alpha$ and $\rho$ as well as the tail-weighted CRPS (twCRPS).} 
\label{tab:sim-results-1}
\begin{tabular}{lrrrrrrrr}
   & \multicolumn{2}{c}{Coverage} & \multicolumn{2}{c}{Interval Score} & \multicolumn{3}{c}{twCRPS} & Walltime \\ 
\cline{2-3}\cline{4-5}\cline{6-8}
 & $\alpha$ & Range & $\alpha$ & Range & 1 & 2 & 3 & (min) \\ \hline
 Full GP & 0.94 & 1.00 & 0.0021 & 0.0001 & 0.0539 & 0.1265 & 0.0120 & 52.9 \\ 
  Vecchia (M=3) & 0.78 & 1.00 & 0.0036 & 0.0001 & 0.0552 & 0.1302 & 0.0123 & 46.8 \\ 
  Vecchia (M=5) & 0.91 & 1.00 & 0.0024 & 0.0001 & 0.0543 & 0.1278 & 0.0121 & 48.5 \\ 
  Vecchia (M=10) & 0.95 & 1.00 & 0.0021 & 0.0001 & 0.0539 & 0.1266 & 0.0120 & 57.1 \\ 
  Vecchia (M=15) & 0.95 & 1.00 & 0.0021 & 0.0001 & 0.0539 & 0.1265 & 0.0120 & 70.2 \\ 
  Taper (90\% sparse) & 0.81 & 0.00 & 0.0035 & 0.0315 & 0.0539 & 0.1262 & 0.0119 & 64.9 \\ 
  Taper (20\% sparse) & 0.95 & 0.00 & 0.0021 & 0.0043 & 0.0539 & 0.1264 & 0.0120 & 110.6 \\ 
  Low Rank (100) & 0.00 &  & 0.8266 &  & 0.0932 & 0.1870 & 0.0257 & 50.4 \\ 
  Low Rank (200) & 0.00 &  & 0.8334 &  & 0.0915 & 0.1838 & 0.0205 & 50.7 \\ 
  NBE & 1.00 & 1.00&0.0046 &0.0005 & - & - & - & 0.14\\
   \hline
\end{tabular}
\normalsize
\end{table}

\begin{table}[htbp]
\small
\centering
\caption{Simulation study results for the low-dimensional case with dependence $\alpha=0.3$, smoothness $\nu=0.5$, and range $\rho=0.05$. We report the coverage and average interval scores for $\alpha$ and $\rho$ as well as the tail-weighted CRPS (twCRPS).} 
\label{tab:sim-results-2}
\begin{tabular}{lrrrrrrrr}
   & \multicolumn{2}{c}{Coverage} & \multicolumn{2}{c}{Interval Score} & \multicolumn{3}{c}{twCRPS} & Walltime \\ 
\cline{2-3}\cline{4-5}\cline{6-8}
 & $\alpha$ & Range & $\alpha$ & Range & 1 & 2 & 3 & (min) \\ \hline
 Full GP & 0.95 & 1.00 & 0.0017 & 0.0001 & 0.0498 & 0.1176 & 0.0111 & 59.6 \\ 
  Vecchia (M=3) & 0.92 & 0.99 & 0.0020 & 0.0001 & 0.0510 & 0.1211 & 0.0114 & 50.6 \\ 
  Vecchia (M=5) & 0.95 & 1.00 & 0.0019 & 0.0001 & 0.0502 & 0.1188 & 0.0112 & 53.8 \\ 
  Vecchia (M=10) & 0.95 & 1.00 & 0.0017 & 0.0001 & 0.0499 & 0.1178 & 0.0111 & 63.9 \\ 
  Vecchia (M=15) & 0.97 & 1.00 & 0.0017 & 0.0001 & 0.0498 & 0.1177 & 0.0111 & 77.2 \\ 
Taper (90\% sparse) & 0.87 & 0.00 & 0.0021 & 0.0300 & 0.0498 & 0.1175 & 0.0111 & 69.0 \\ 
Taper (20\% sparse) & 0.96 & 0.00 & 0.0017 & 0.0038 & 0.0498 & 0.1176 & 0.0111 & 117.8 \\ 
  Low Rank (100) & 0.00 &  & 0.5355 &  & 0.0859 & 0.1751 & 0.0232 & 50.9 \\ 
  Low Rank (200) & 0.00 &  & 0.5419 &  & 0.0851 & 0.1737 & 0.0191 & 50.9 \\ 
   NBE & 0.92 & 1.00 &0.0088 & 0.0005 & - & - & - & 0.14\\
   \hline
\end{tabular}
\normalsize
\end{table}

\begin{table}[htbp]
\small
\centering
\caption{Simulation study results for the low-dimensional case with dependence $\alpha=0.45$, smoothness $\nu=0.5$, and range $\rho=0.05$. We report the coverage and average interval scores for $\alpha$ and $\rho$ as well as the tail-weighted CRPS (twCRPS).} 
\label{tab:sim-results-3}
\begin{tabular}{lrrrrrrrr}
   & \multicolumn{2}{c}{Coverage} & \multicolumn{2}{c}{Interval Score} & \multicolumn{3}{c}{twCRPS} & Walltime \\ 
\cline{2-3}\cline{4-5}\cline{6-8}
 & $\alpha$ & Range & $\alpha$ & Range & 1 & 2 & 3 & (min) \\ \hline
 Full GP & 0.96 & 1.00 & 0.0017 & 0.0001 & 0.0457 & 0.1089 & 0.0101 & 58.5 \\ 
  Vecchia (M=3) & 0.94 & 0.98 & 0.0017 & 0.0001 & 0.0468 & 0.1120 & 0.0104 & 49.7 \\ 
  Vecchia (M=5) & 0.96 & 0.98 & 0.0017 & 0.0001 & 0.0461 & 0.1100 & 0.0102 & 52.9 \\ 
  Vecchia (M=10) & 0.95 & 1.00 & 0.0017 & 0.0001 & 0.0457 & 0.1090 & 0.0101 & 63.4 \\ 
  Vecchia (M=15) & 0.96 & 1.00 & 0.0017 & 0.0001 & 0.0457 & 0.1089 & 0.0101 & 76.1 \\ 
  Taper (90\% sparse) & 0.93 & 0.00 & 0.0016 & 0.0305 & 0.0457 & 0.1087 & 0.0101 & 68.2 \\ 
  Taper (20\% sparse) & 0.95 & 0.00 & 0.0017 & 0.0040 & 0.0457 & 0.1088 & 0.0101 & 117.2 \\ 
  Low Rank (100) & 0.00 &  & 0.3807 &  & 0.0782 & 0.1628 & 0.0208 & 50.7 \\ 
  Low Rank (200) & 0.00 &  & 0.3877 &  & 0.0779 & 0.1618 & 0.0176 & 50.7 \\ 
     NBE & 0.99 & 0.96 &0.0109 &0.0005 & - & - & - & 0.14\\
   \hline
\end{tabular}
\normalsize
\end{table}

\begin{table}[htbp]
\small
\centering
\caption{Simulation study results for the low-dimensional case with dependence $\alpha=0.55$, smoothness $\nu=0.5$, and range $\rho=0.05$. We report the coverage and average interval scores for $\alpha$ and $\rho$ as well as the tail-weighted CRPS (twCRPS).} 
\label{tab:sim-results-4}
\begin{tabular}{lrrrrrrrr}
   & \multicolumn{2}{c}{Coverage} & \multicolumn{2}{c}{Interval Score} & \multicolumn{3}{c}{twCRPS} & Walltime \\ 
\cline{2-3}\cline{4-5}\cline{6-8}
 & $\alpha$ & Range & $\alpha$ & Range & 1 & 2 & 3 & (min) \\ \hline
 Full GP & 0.95 & 1.00 & 0.0016 & 0.0001 & 0.0423 & 0.1029 & 0.0090 & 58.4 \\ 
  Vecchia (M=3) & 0.91 & 1.00 & 0.0016 & 0.0001 & 0.0434 & 0.1060 & 0.0092 & 49.0 \\ 
  Vecchia (M=5) & 0.95 & 1.00 & 0.0016 & 0.0001 & 0.0427 & 0.1040 & 0.0091 & 52.2 \\ 
  Vecchia (M=10) & 0.95 & 1.00 & 0.0016 & 0.0001 & 0.0424 & 0.1031 & 0.0090 & 62.3 \\ 
  Vecchia (M=15) & 0.95 & 1.00 & 0.0016 & 0.0001 & 0.0423 & 0.1030 & 0.0090 & 75.8 \\ 
  Taper (90\% sparse) & 0.92 & 0.00 & 0.0017 & 0.0302 & 0.0424 & 0.1029 & 0.0090 & 67.9 \\ 
  Taper (20\% sparse) & 0.95 & 0.00 & 0.0016 & 0.0039 & 0.0423 & 0.1029 & 0.0090 & 116.3 \\ 
  Low Rank (100) & 0.00 &  & 0.2752 &  & 0.0723 & 0.1548 & 0.0186 & 50.9 \\ 
  Low Rank (200) & 0.00 &  & 0.2789 &  & 0.0724 & 0.1536 & 0.0158 & 50.9 \\ 
     NBE & 1.00 & 0.95 &0.0124 & 0.0006 & - & - & - & 0.14\\
   \hline
\end{tabular}
\normalsize
\end{table}

\begin{table}[htbp]
\small
\centering
\caption{Simulation study results for the low-dimensional case with dependence $\alpha=0.7$, smoothness $\nu=0.5$, and range $\rho=0.05$. We report the coverage and average interval scores for $\alpha$ and $\rho$ as well as the tail-weighted CRPS (twCRPS).} 
\label{tab:sim-results-5}
\begin{tabular}{lrrrrrrrr}
   & \multicolumn{2}{c}{Coverage} & \multicolumn{2}{c}{Interval Score} & \multicolumn{3}{c}{twCRPS} & Walltime \\ 
\cline{2-3}\cline{4-5}\cline{6-8}
 & $\alpha$ & Range & $\alpha$ & Range & 1 & 2 & 3 & (min) \\ \hline
 Full GP & 0.93 & 0.97 & 0.0104 & 0.0004 & 0.0382 & 0.0946 & 0.0080 & 56.9 \\ 
  Vecchia (M=3) & 0.93 & 0.96 & 0.0106 & 0.0004 & 0.0391 & 0.0973 & 0.0082 & 48.3 \\ 
  Vecchia (M=5) & 0.93 & 0.97 & 0.0098 & 0.0004 & 0.0385 & 0.0955 & 0.0081 & 51.9 \\ 
  Vecchia (M=10) & 0.92 & 0.97 & 0.0105 & 0.0004 & 0.0383 & 0.0947 & 0.0080 & 61.5 \\ 
  Vecchia (M=15) & 0.92 & 0.97 & 0.0105 & 0.0004 & 0.0382 & 0.0946 & 0.0080 & 75.1 \\ 
  Taper (90\% sparse) & 0.89 & 0.00 & 0.0105 & 0.0295 & 0.0382 & 0.0945 & 0.0080 & 67.1 \\ 
  Taper (20\% sparse) & 0.93 & 0.00 & 0.0105 & 0.0040 & 0.0382 & 0.0946 & 0.0080 & 115.7 \\ 
  Low Rank (100) & 0.00 &  & 0.1302 &  & 0.0649 & 0.1432 & 0.0160 & 50.8 \\ 
  Low Rank (200) & 0.00 &  & 0.1361 &  & 0.0655 & 0.1433 & 0.0140 & 51.2 \\ 
     NBE & 0.96 & 0.93 &0.0077 & 0.0005& - & - & - & 0.14\\
   \hline
\end{tabular}
\normalsize
\end{table}

\begin{table}[htbp]
\small
\centering
\caption{Simulation study results for the low-dimensional case with dependence $\alpha=0.05$, smoothness $\nu=0.5$, and range $\rho=0.25$. We report the coverage and average interval scores for $\alpha$ and $\rho$ as well as the tail-weighted CRPS (twCRPS).} 
\label{tab:sim-results-6}
\begin{tabular}{lrrrrrrrr}
   & \multicolumn{2}{c}{Coverage} & \multicolumn{2}{c}{Interval Score} & \multicolumn{3}{c}{twCRPS} & Walltime \\ 
\cline{2-3}\cline{4-5}\cline{6-8}
 & $\alpha$ & Range & $\alpha$ & Range & 1 & 2 & 3 & (min) \\ \hline
 Full GP & 0.98 & 0.96 & 0.0029 & 0.0003 & 0.0252 & 0.0532 & 0.0059 & 52.4 \\ 
  Vecchia (M=3) & 0.66 & 0.78 & 0.0114 & 0.0010 & 0.0267 & 0.0571 & 0.0063 & 47.5 \\ 
  Vecchia (M=5) & 0.92 & 0.92 & 0.0042 & 0.0004 & 0.0258 & 0.0546 & 0.0061 & 48.5 \\ 
  Vecchia (M=10) & 0.97 & 0.96 & 0.0034 & 0.0004 & 0.0253 & 0.0535 & 0.0059 & 64.2 \\ 
  Vecchia (M=15) & 0.97 & 0.96 & 0.0031 & 0.0004 & 0.0252 & 0.0533 & 0.0059 & 88.0 \\ 
  Taper (90\% sparse) & 0.00 & 0.00 & 0.8447 & 0.2471 & 0.0254 & 0.0559 & 0.0062 & 64.4 \\ 
  Taper (20\% sparse) & 0.79 & 0.00 & 0.0061 & 0.2239 & 0.0252 & 0.0532 & 0.0059 & 105.2 \\ 
  Low Rank (100) & 0.00 &  & 0.7833 &  & 0.0430 & 0.0882 & 0.0105 & 50.7 \\ 
  Low Rank (200) & 0.00 &  & 0.7891 &  & 0.0421 & 0.0860 & 0.0098 & 50.7 \\ 
      NBE & 1.00 & 1.00 &0.0065 & 0.0017& - & - & - & 0.14\\
   \hline
\end{tabular}
\normalsize
\end{table}

\begin{table}[htbp]
\small
\centering
\caption{Simulation study results for the low-dimensional case with dependence $\alpha=0.3$, smoothness $\nu=0.5$, and range $\rho=0.25$. We report the coverage and average interval scores for $\alpha$ and $\rho$ as well as the tail-weighted CRPS (twCRPS).} 
\label{tab:sim-results-7}
\begin{tabular}{lrrrrrrrr}
   & \multicolumn{2}{c}{Coverage} & \multicolumn{2}{c}{Interval Score} & \multicolumn{3}{c}{twCRPS} & Walltime \\ 
\cline{2-3}\cline{4-5}\cline{6-8}
 & $\alpha$ & Range & $\alpha$ & Range & 1 & 2 & 3 & (min) \\ \hline
 Full GP & 0.98 & 0.99 & 0.0028 & 0.0006 & 0.0233 & 0.0498 & 0.0055 & 59.2 \\ 
  Vecchia (M=3) & 0.77 & 0.79 & 0.0065 & 0.0015 & 0.0249 & 0.0534 & 0.0058 & 50.3 \\ 
  Vecchia (M=5) & 0.94 & 0.95 & 0.0033 & 0.0006 & 0.0239 & 0.0512 & 0.0056 & 54.3 \\ 
  Vecchia (M=10) & 0.98 & 0.98 & 0.0029 & 0.0006 & 0.0235 & 0.0501 & 0.0055 & 70.6 \\ 
  Vecchia (M=15) & 0.98 & 0.99 & 0.0028 & 0.0006 & 0.0234 & 0.0499 & 0.0055 & 94.3 \\ 
  Taper (90\% sparse) & 0.00 & 0.00 & 0.6903 & 0.2472 & 0.0236 & 0.0526 & 0.0057 & 53.5 \\ 
  Taper (20\% sparse) & 0.62 & 0.00 & 0.0063 & 0.2017 & 0.0234 & 0.0498 & 0.0055 & 112.7 \\ 
  Low Rank (100) & 0.00 &  & 0.5314 &  & 0.0402 & 0.0854 & 0.0097 & 50.9 \\ 
  Low Rank (200) & 0.00 &  & 0.5333 &  & 0.0397 & 0.0842 & 0.0092 & 51.0 \\ 
      NBE & 0.95 & 1.00 &0.0096 & 0.0020& - & - & - & 0.14\\
   \hline
\end{tabular}
\normalsize
\end{table}

\begin{table}[htbp]
\small
\centering
\caption{Simulation study results for the low-dimensional case with dependence $\alpha=0.45$, smoothness $\nu=0.5$, and range $\rho=0.25$. We report the coverage and average interval scores for $\alpha$ and $\rho$ as well as the tail-weighted CRPS (twCRPS).} 
\label{tab:sim-results-8}
\begin{tabular}{lrrrrrrrr}
   & \multicolumn{2}{c}{Coverage} & \multicolumn{2}{c}{Interval Score} & \multicolumn{3}{c}{twCRPS} & Walltime \\ 
\cline{2-3}\cline{4-5}\cline{6-8}
 & $\alpha$ & Range & $\alpha$ & Range & 1 & 2 & 3 & (min) \\ \hline
 Full GP & 0.94 & 0.94 & 0.0086 & 0.0017 & 0.0215 & 0.0463 & 0.0049 & 58.7 \\ 
  Vecchia (M=3) & 0.89 & 0.83 & 0.0101 & 0.0024 & 0.0228 & 0.0496 & 0.0052 & 49.6 \\ 
  Vecchia (M=5) & 0.92 & 0.94 & 0.0089 & 0.0018 & 0.0220 & 0.0475 & 0.0050 & 54.0 \\ 
  Vecchia (M=10) & 0.93 & 0.94 & 0.0087 & 0.0018 & 0.0216 & 0.0466 & 0.0049 & 69.0 \\ 
  Vecchia (M=15) & 0.93 & 0.93 & 0.0086 & 0.0017 & 0.0215 & 0.0464 & 0.0049 & 93.8 \\ 
 Taper (90\% sparse) & 0.00 & 0.00 & 0.5471 & 0.2480 & 0.0218 & 0.0495 & 0.0050 & 52.8 \\ 
  Taper (20\% sparse) & 0.72 & 0.00 & 0.0155 & 0.1974 & 0.0215 & 0.0463 & 0.0049 & 111.2 \\ 
  Low Rank (100) & 0.00 &  & 0.3827 &  & 0.0369 & 0.0810 & 0.0087 & 50.5 \\ 
  Low Rank (200) & 0.00 &  & 0.3841 &  & 0.0365 & 0.0804 & 0.0083 & 50.5 \\ 
      NBE & 0.93 & 0.93 &0.0157 & 0.0030& - & - & - & 0.14\\
   \hline
\end{tabular}
\normalsize
\end{table}

\begin{table}[htbp]
\small
\centering
\caption{Simulation study results for the low-dimensional case with dependence $\alpha=0.55$, smoothness $\nu=0.5$, and range $\rho=0.25$. We report the coverage and average interval scores for $\alpha$ and $\rho$ as well as the tail-weighted CRPS (twCRPS).} 
\label{tab:sim-results-9}
\begin{tabular}{lrrrrrrrr}
   & \multicolumn{2}{c}{Coverage} & \multicolumn{2}{c}{Interval Score} & \multicolumn{3}{c}{twCRPS} & Walltime \\ 
\cline{2-3}\cline{4-5}\cline{6-8}
 & $\alpha$ & Range & $\alpha$ & Range & 1 & 2 & 3 & (min) \\ \hline
 Full GP & 0.94 & 0.97 & 0.0071 & 0.0015 & 0.0201 & 0.0437 & 0.0045 & 58.6 \\ 
  Vecchia (M=3) & 0.77 & 0.77 & 0.0097 & 0.0025 & 0.0213 & 0.0468 & 0.0048 & 49.2 \\ 
  Vecchia (M=5) & 0.93 & 0.94 & 0.0075 & 0.0016 & 0.0205 & 0.0449 & 0.0046 & 53.0 \\ 
  Vecchia (M=10) & 0.94 & 0.97 & 0.0072 & 0.0015 & 0.0202 & 0.0439 & 0.0045 & 69.3 \\ 
  Vecchia (M=15) & 0.94 & 0.97 & 0.0073 & 0.0015 & 0.0201 & 0.0438 & 0.0045 & 93.1 \\ 
  Taper (90\% sparse) & 0.00 & 0.00 & 0.4480 & 0.2481 & 0.0204 & 0.0473 & 0.0045 & 52.8 \\ 
  Taper (20\% sparse) & 0.57 & 0.00 & 0.0101 & 0.1951 & 0.0201 & 0.0437 & 0.0045 & 111.4 \\ 
  Low Rank (100) & 0.00 &  & 0.2976 &  & 0.0347 & 0.0799 & 0.0080 & 50.7 \\ 
  Low Rank (200) & 0.00 &  & 0.2939 &  & 0.0341 & 0.0781 & 0.0076 & 50.7 \\ 
      NBE & 0.97 & 0.96 &0.0155 & 0.0034& - & - & - & 0.14\\
   \hline
\end{tabular}
\normalsize
\end{table}

\begin{table}[htbp]
\small
\centering
\caption{Simulation study results for the low-dimensional case with dependence $\alpha=0.7$, smoothness $\nu=0.5$, and range $\rho=0.25$. We report the coverage and average interval scores for $\alpha$ and $\rho$ as well as the tail-weighted CRPS (twCRPS).} 
\label{tab:sim-results-10}
\begin{tabular}{lrrrrrrrr}
   & \multicolumn{2}{c}{Coverage} & \multicolumn{2}{c}{Interval Score} & \multicolumn{3}{c}{twCRPS} & Walltime \\ 
\cline{2-3}\cline{4-5}\cline{6-8}
 & $\alpha$ & Range & $\alpha$ & Range & 1 & 2 & 3 & (min) \\ \hline
 Full GP & 0.90 & 0.89 & 0.0170 & 0.0036 & 0.0184 & 0.0405 & 0.0041 & 57.5 \\ 
  Vecchia (M=3) & 0.76 & 0.77 & 0.0158 & 0.0038 & 0.0196 & 0.0434 & 0.0043 & 48.6 \\ 
  Vecchia (M=5) & 0.92 & 0.90 & 0.0156 & 0.0034 & 0.0188 & 0.0416 & 0.0042 & 52.5 \\ 
  Vecchia (M=10) & 0.90 & 0.90 & 0.0169 & 0.0036 & 0.0185 & 0.0408 & 0.0041 & 69.1 \\ 
  Vecchia (M=15) & 0.91 & 0.90 & 0.0165 & 0.0036 & 0.0184 & 0.0406 & 0.0041 & 93.7 \\ 
  Taper (90\% sparse) & 0.00 & 0.00 & 0.2980 & 0.2484 & 0.0189 & 0.0448 & 0.0040 & 52.8 \\ 
  Taper (20\% sparse) & 0.55 & 0.00 & 0.0206 & 0.1883 & 0.0184 & 0.0405 & 0.0041 & 110.5 \\ 
  Low Rank (100) & 0.00 &  & 0.1691 &  & 0.0321 & 0.0825 & 0.0073 & 50.7 \\ 
  Low Rank (200) & 0.00 &  & 0.1627 &  & 0.0316 & 0.0793 & 0.0070 & 50.4 \\ 
      NBE & 0.94 & 0.91 &0.0118 & 0.0032& - & - & - & 0.14\\
   \hline
\end{tabular}
\normalsize
\end{table}

\begin{table}[htbp]
\small
\centering
\caption{Simulation study results for the low-dimensional case with dependence $\alpha=0.05$, smoothness $\nu=1.5$, and range $\rho=0.032$. We report the coverage and average interval scores for $\alpha$ and $\rho$ as well as the tail-weighted CRPS (twCRPS).} 
\label{tab:sim-results-11}
\begin{tabular}{lrrrrrrrr}
   & \multicolumn{2}{c}{Coverage} & \multicolumn{2}{c}{Interval Score} & \multicolumn{3}{c}{twCRPS} & Walltime \\ 
\cline{2-3}\cline{4-5}\cline{6-8}
 & $\alpha$ & Range & $\alpha$ & Range & 1 & 2 & 3 & (min) \\ \hline
 Full GP & 0.96 & 0.93 & 0.0021 & 0.0001 & 0.0309 & 0.0680 & 0.0071 & 53.2 \\ 
  Vecchia (M=3) & 0.90 & 0.95 & 0.0030 & 0.0000 & 0.0344 & 0.0766 & 0.0079 & 45.8 \\ 
  Vecchia (M=5) & 0.95 & 0.92 & 0.0023 & 0.0001 & 0.0324 & 0.0717 & 0.0074 & 47.1 \\ 
  Vecchia (M=10) & 0.96 & 0.92 & 0.0021 & 0.0001 & 0.0312 & 0.0687 & 0.0071 & 54.4 \\ 
  Vecchia (M=15) & 0.97 & 0.92 & 0.0021 & 0.0001 & 0.0310 & 0.0682 & 0.0071 & 64.8 \\ 
  Taper (90\% sparse) & 0.00 & 0.00 & 0.3534 & 0.0264 & 0.0324 & 0.0736 & 0.0080 & 68.8 \\ 
  Taper (20\% sparse) & 0.86 & 0.00 & 0.0064 & 0.0059 & 0.0312 & 0.0687 & 0.0072 & 112.4 \\ 
  Low Rank (100) & 0.00 &  & 0.8270 &  & 0.0716 & 0.1443 & 0.0178 & 50.5 \\ 
  Low Rank (200) & 0.00 &  & 0.8332 &  & 0.0590 & 0.1189 & 0.0129 & 50.7 \\ 
        NBE & 1.00 & 1.00 &0.0053 & 0.0003& - & - & - & 0.14\\

   \hline
\end{tabular}
\normalsize
\end{table}

\begin{table}[htbp]
\small
\centering
\caption{Simulation study results for the low-dimensional case with dependence $\alpha=0.3$, smoothness $\nu=1.5$, and range $\rho=0.032$. We report the coverage and average interval scores for $\alpha$ and $\rho$ as well as the tail-weighted CRPS (twCRPS).} 
\label{tab:sim-results-12}
\begin{tabular}{lrrrrrrrr}
   & \multicolumn{2}{c}{Coverage} & \multicolumn{2}{c}{Interval Score} & \multicolumn{3}{c}{twCRPS} & Walltime \\ 
\cline{2-3}\cline{4-5}\cline{6-8}
 & $\alpha$ & Range & $\alpha$ & Range & 1 & 2 & 3 & (min) \\ \hline
 Full GP & 0.94 & 0.90 & 0.0020 & 0.0001 & 0.0283 & 0.0626 & 0.0065 & 59.7 \\ 
  Vecchia (M=3) & 0.94 & 0.90 & 0.0021 & 0.0001 & 0.0315 & 0.0704 & 0.0073 & 50.3 \\ 
  Vecchia (M=5) & 0.94 & 0.91 & 0.0020 & 0.0001 & 0.0297 & 0.0660 & 0.0069 & 53.2 \\ 
  Vecchia (M=10) & 0.95 & 0.89 & 0.0020 & 0.0001 & 0.0286 & 0.0634 & 0.0066 & 60.3 \\ 
  Vecchia (M=15) & 0.93 & 0.91 & 0.0021 & 0.0001 & 0.0284 & 0.0628 & 0.0065 & 71.5 \\ 
  Taper (90\% sparse) & 0.00 & 0.00 & 0.2430 & 0.0261 & 0.0297 & 0.0679 & 0.0072 & 68.0 \\ 
  Taper (20\% sparse) & 0.11 & 0.00 & 0.0290 & 0.0050 & 0.0286 & 0.0634 & 0.0067 & 117.9 \\ 
  Low Rank (100) & 0.00 &  & 0.5401 &  & 0.0659 & 0.1349 & 0.0164 & 50.9 \\ 
  Low Rank (200) & 0.00 &  & 0.5493 &  & 0.0540 & 0.1102 & 0.0120 & 50.8 \\ 
   NBE & 0.96 & 1.00 &0.0086 & 0.0002& - & - & - & 0.14\\
   \hline
\end{tabular}
\normalsize
\end{table}

\begin{table}[htbp]
\small
\centering
\caption{Simulation study results for the low-dimensional case with dependence $\alpha=0.45$, smoothness $\nu=1.5$, and range $\rho=0.032$. We report the coverage and average interval scores for $\alpha$ and $\rho$ as well as the tail-weighted CRPS (twCRPS).} 
\label{tab:sim-results-13}
\begin{tabular}{lrrrrrrrr}
   & \multicolumn{2}{c}{Coverage} & \multicolumn{2}{c}{Interval Score} & \multicolumn{3}{c}{twCRPS} & Walltime \\ 
\cline{2-3}\cline{4-5}\cline{6-8}
 & $\alpha$ & Range & $\alpha$ & Range & 1 & 2 & 3 & (min) \\ \hline
 Full GP & 0.98 & 0.92 & 0.0016 & 0.0001 & 0.0260 & 0.0586 & 0.0058 & 59.3 \\ 
  Vecchia (M=3) & 0.92 & 0.94 & 0.0018 & 0.0001 & 0.0289 & 0.0659 & 0.0065 & 49.4 \\ 
  Vecchia (M=5) & 0.98 & 0.92 & 0.0017 & 0.0001 & 0.0272 & 0.0617 & 0.0061 & 52.2 \\ 
  Vecchia (M=10) & 0.98 & 0.92 & 0.0017 & 0.0001 & 0.0263 & 0.0593 & 0.0059 & 60.2 \\ 
  Vecchia (M=15) & 0.98 & 0.93 & 0.0017 & 0.0001 & 0.0261 & 0.0588 & 0.0059 & 69.9 \\ 
  Taper (90\% sparse) & 0.00 & 0.00 & 0.2288 & 0.0254 & 0.0272 & 0.0633 & 0.0064 & 67.2 \\ 
  Taper (20\% sparse) & 0.04 & 0.00 & 0.0296 & 0.0048 & 0.0262 & 0.0593 & 0.0060 & 116.2 \\ 
  Low Rank (100) & 0.00 &  & 0.3785 &  & 0.0602 & 0.1267 & 0.0147 & 50.8 \\ 
  Low Rank (200) & 0.00 &  & 0.3870 &  & 0.0497 & 0.1038 & 0.0110 & 50.9 \\
     NBE & 1.00 & 1.00 &0.0115 & 0.0002& - & - & - & 0.14\\

   \hline
\end{tabular}
\normalsize
\end{table}

\begin{table}[htbp]
\small
\centering
\caption{Simulation study results for the low-dimensional case with dependence $\alpha=0.55$, smoothness $\nu=1.5$, and range $\rho=0.032$. We report the coverage and average interval scores for $\alpha$ and $\rho$ as well as the tail-weighted CRPS (twCRPS).} 
\label{tab:sim-results-14}
\begin{tabular}{lrrrrrrrr}
   & \multicolumn{2}{c}{Coverage} & \multicolumn{2}{c}{Interval Score} & \multicolumn{3}{c}{twCRPS} & Walltime \\ 
\cline{2-3}\cline{4-5}\cline{6-8}
 & $\alpha$ & Range & $\alpha$ & Range & 1 & 2 & 3 & (min) \\ \hline
 Full GP & 0.97 & 0.95 & 0.0015 & 0.0000 & 0.0246 & 0.0555 & 0.0055 & 58.5 \\ 
  Vecchia (M=3) & 0.97 & 0.94 & 0.0016 & 0.0001 & 0.0273 & 0.0624 & 0.0061 & 48.8 \\ 
  Vecchia (M=5) & 0.96 & 0.94 & 0.0016 & 0.0001 & 0.0258 & 0.0585 & 0.0057 & 51.5 \\ 
  Vecchia (M=10) & 0.98 & 0.94 & 0.0015 & 0.0000 & 0.0248 & 0.0561 & 0.0055 & 59.0 \\ 
  Vecchia (M=15) & 0.98 & 0.96 & 0.0015 & 0.0000 & 0.0247 & 0.0557 & 0.0055 & 69.7 \\ 
  Taper (90\% sparse) & 0.00 & 0.00 & 0.2205 & 0.0254 & 0.0257 & 0.0599 & 0.0060 & 66.9 \\ 
  Taper (20\% sparse) & 0.05 & 0.00 & 0.0257 & 0.0049 & 0.0248 & 0.0561 & 0.0056 & 116.4 \\ 
  Low Rank (100) & 0.00 &  & 0.2817 &  & 0.0568 & 0.1211 & 0.0138 & 50.9 \\ 
  Low Rank (200) & 0.00 &  & 0.2957 &  & 0.0471 & 0.0994 & 0.0104 & 51.1 \\ 
     NBE & 1.00 & 1.00 &0.0124 & 0.0002& - & - & - & 0.14\\

   \hline
\end{tabular}
\normalsize
\end{table}

\begin{table}[htbp]
\small
\centering
\caption{Simulation study results for the low-dimensional case with dependence $\alpha=0.7$, smoothness $\nu=1.5$, and range $\rho=0.032$. We report the coverage and average interval scores for $\alpha$ and $\rho$ as well as the tail-weighted CRPS (twCRPS).} 
\label{tab:sim-results-15}
\begin{tabular}{lrrrrrrrr}
   & \multicolumn{2}{c}{Coverage} & \multicolumn{2}{c}{Interval Score} & \multicolumn{3}{c}{twCRPS} & Walltime \\ 
\cline{2-3}\cline{4-5}\cline{6-8}
 & $\alpha$ & Range & $\alpha$ & Range & 1 & 2 & 3 & (min) \\ \hline
 Full GP & 0.95 & 0.93 & 0.0017 & 0.0001 & 0.0222 & 0.0507 & 0.0048 & 57.9 \\ 
  Vecchia (M=3) & 0.96 & 0.94 & 0.0018 & 0.0000 & 0.0247 & 0.0570 & 0.0054 & 48.3 \\ 
  Vecchia (M=5) & 0.95 & 0.96 & 0.0018 & 0.0000 & 0.0233 & 0.0533 & 0.0051 & 51.0 \\ 
  Vecchia (M=10) & 0.96 & 0.94 & 0.0016 & 0.0000 & 0.0225 & 0.0513 & 0.0049 & 58.9 \\ 
  Vecchia (M=15) & 0.95 & 0.95 & 0.0017 & 0.0000 & 0.0223 & 0.0509 & 0.0049 & 69.1 \\ 
  Taper (90\% sparse) & 0.00 & 0.00 & 0.2321 & 0.0249 & 0.0234 & 0.0547 & 0.0053 & 62.6 \\ 
  Taper (20\% sparse) & 0.10 & 0.00 & 0.0286 & 0.0048 & 0.0225 & 0.0513 & 0.0050 & 116.1 \\ 
  Low Rank (100) & 0.00 &  & 0.1516 &  & 0.0515 & 0.1128 & 0.0122 & 50.6 \\ 
  Low Rank (200) & 0.00 &  & 0.1671 &  & 0.0430 & 0.0933 & 0.0094 & 50.5 \\ 
     NBE & 0.97 & 0.98 &0.0090 & 0.0002& - & - & - & 0.14\\

   \hline
\end{tabular}
\normalsize
\end{table}

\begin{table}[htbp]
\small
\centering
\caption{Simulation study results for the low-dimensional case with dependence $\alpha=0.05$, smoothness $\nu=1.5$, and range $\rho=0.158$. We report the coverage and average interval scores for $\alpha$ and $\rho$ as well as the tail-weighted CRPS (twCRPS).} 
\label{tab:sim-results-16}
\begin{tabular}{lrrrrrrrr}
   & \multicolumn{2}{c}{Coverage} & \multicolumn{2}{c}{Interval Score} & \multicolumn{3}{c}{twCRPS} & Walltime \\ 
\cline{2-3}\cline{4-5}\cline{6-8}
 & $\alpha$ & Range & $\alpha$ & Range & 1 & 2 & 3 & (min) \\ \hline
 Full GP & 0.97 & 0.81 & 0.0205 & 0.0007 & 0.0035 & 0.0070 & 0.0008 & 52.1 \\ 
  Vecchia (M=3) & 0.01 & 0.10 & 0.0975 & 0.0015 & 0.0058 & 0.0117 & 0.0014 & 50.2 \\ 
  Vecchia (M=5) & 0.08 & 0.40 & 0.0535 & 0.0004 & 0.0044 & 0.0089 & 0.0010 & 53.6 \\ 
  Vecchia (M=10) & 0.85 & 0.94 & 0.0062 & 0.0001 & 0.0037 & 0.0075 & 0.0009 & 64.8 \\ 
  Vecchia (M=15) & 0.93 & 0.94 & 0.0042 & 0.0001 & 0.0036 & 0.0072 & 0.0008 & 85.6 \\ 
  Taper (90\% sparse) & 0.00 & 0.00 & 0.9490 & 0.3347 & 0.0098 & 0.0261 & 0.0022 & 52.9 \\ 
  Taper (20\% sparse) & 0.00 & 0.00 & 0.9486 & 0.1603 & 0.0060 & 0.0134 & 0.0014 & 102.6 \\ 
  Low Rank (100) & 0.00 &  & 0.7468 &  & 0.0104 & 0.0271 & 0.0023 & 50.8 \\ 
  Low Rank (200) & 0.00 &  & 0.7397 &  & 0.0074 & 0.0206 & 0.0016 & 50.9 \\ 
     NBE & 0.95 & 0.98 &0.0092 & 0.0010& - & - & - & 0.14\\

   \hline
\end{tabular}
\normalsize
\end{table}

\begin{table}[htbp]
\small
\centering
\caption{Simulation study results for the low-dimensional case with dependence $\alpha=0.3$, smoothness $\nu=1.5$, and range $\rho=0.158$. We report the coverage and average interval scores for $\alpha$ and $\rho$ as well as the tail-weighted CRPS (twCRPS).} 
\label{tab:sim-results-17}
\begin{tabular}{lrrrrrrrr}
   & \multicolumn{2}{c}{Coverage} & \multicolumn{2}{c}{Interval Score} & \multicolumn{3}{c}{twCRPS} & Walltime \\ 
\cline{2-3}\cline{4-5}\cline{6-8}
 & $\alpha$ & Range & $\alpha$ & Range & 1 & 2 & 3 & (min) \\ \hline
 Full GP & 0.94 & 0.94 & 0.0032 & 0.0002 & 0.0032 & 0.0065 & 0.0007 & 59.4 \\ 
  Vecchia (M=3) & 0.41 & 0.39 & 0.0224 & 0.0010 & 0.0053 & 0.0107 & 0.0012 & 50.5 \\ 
  Vecchia (M=5) & 0.54 & 0.56 & 0.0162 & 0.0006 & 0.0040 & 0.0082 & 0.0009 & 54.5 \\ 
  Vecchia (M=10) & 0.82 & 0.89 & 0.0058 & 0.0002 & 0.0034 & 0.0069 & 0.0008 & 69.7 \\ 
  Vecchia (M=15) & 0.85 & 0.93 & 0.0051 & 0.0002 & 0.0033 & 0.0066 & 0.0008 & 92.9 \\ 
  Taper (90\% sparse) & 0.00 & 0.00 & 0.6990 & 0.3346 & 0.0096 & 0.0253 & 0.0021 & 52.8 \\ 
  Taper (20\% sparse) & 0.00 & 0.00 & 0.6990 & 0.1990 & 0.0058 & 0.0129 & 0.0013 & 101.9 \\ 
  Low Rank (100) & 0.00 &  & 0.4888 &  & 0.0095 & 0.0261 & 0.0021 & 50.7 \\ 
  Low Rank (200) & 0.00 &  & 0.4791 &  & 0.0071 & 0.0229 & 0.0014 & 50.8 \\ 
       NBE & 1.00 & 0.99 &0.0113 & 0.0010& - & - & - & 0.14\\
   \hline
\end{tabular}
\normalsize
\end{table}
\begin{table}[htbp]
\small
\centering
\caption{Simulation study results for the low-dimensional case with dependence $\alpha=0.45$, smoothness $\nu=1.5$, and range $\rho=0.158$. We report the coverage and average interval scores for $\alpha$ and $\rho$ as well as the tail-weighted CRPS (twCRPS).} 
\label{tab:sim-results-18}
\begin{tabular}{lrrrrrrrr}
   & \multicolumn{2}{c}{Coverage} & \multicolumn{2}{c}{Interval Score} & \multicolumn{3}{c}{twCRPS} & Walltime \\ 
\cline{2-3}\cline{4-5}\cline{6-8}
 & $\alpha$ & Range & $\alpha$ & Range & 1 & 2 & 3 & (min) \\ \hline
 Full GP & 0.96 & 0.95 & 0.0042 & 0.0002 & 0.0030 & 0.0062 & 0.0007 & 59.3 \\ 
  Vecchia (M=3) & 0.47 & 0.39 & 0.0181 & 0.0009 & 0.0049 & 0.0102 & 0.0011 & 49.4 \\ 
  Vecchia (M=5) & 0.55 & 0.53 & 0.0169 & 0.0007 & 0.0038 & 0.0078 & 0.0009 & 53.3 \\ 
  Vecchia (M=10) & 0.73 & 0.87 & 0.0093 & 0.0003 & 0.0032 & 0.0066 & 0.0007 & 68.5 \\ 
  Vecchia (M=15) & 0.81 & 0.91 & 0.0063 & 0.0002 & 0.0031 & 0.0063 & 0.0007 & 92.0 \\ 
  Taper (90\% sparse) & 0.00 & 0.00 & 0.5490 & 0.3344 & 0.0094 & 0.0251 & 0.0020 & 52.7 \\ 
  Taper (20\% sparse) & 0.00 & 0.00 & 0.5490 & 0.2536 & 0.0056 & 0.0127 & 0.0012 & 97.2 \\ 
  Low Rank (100) & 0.00 &  & 0.3494 &  & 0.0094 & 0.0303 & 0.0020 & 51.0 \\ 
  Low Rank (200) & 0.00 &  & 0.3176 &  & 0.0074 & 0.0282 & 0.0014 & 51.1 \\ 
       NBE & 0.93 & 0.99 &0.0156 & 0.0011& - & - & - & 0.14\\

   \hline
\end{tabular}
\normalsize
\end{table}
\begin{table}[htbp]
\small
\centering
\caption{Simulation study results for the low-dimensional case with dependence $\alpha=0.55$, smoothness $\nu=1.5$, and range $\rho=0.158$. We report the coverage and average interval scores for $\alpha$ and $\rho$ as well as the tail-weighted CRPS (twCRPS).} 
\label{tab:sim-results-19}
\begin{tabular}{lrrrrrrrr}
   & \multicolumn{2}{c}{Coverage} & \multicolumn{2}{c}{Interval Score} & \multicolumn{3}{c}{twCRPS} & Walltime \\ 
\cline{2-3}\cline{4-5}\cline{6-8}
 & $\alpha$ & Range & $\alpha$ & Range & 1 & 2 & 3 & (min) \\ \hline
 Full GP & 0.94 & 0.94 & 0.0215 & 0.0010 & 0.0028 & 0.0059 & 0.0006 & 57.7 \\ 
  Vecchia (M=3) & 0.48 & 0.51 & 0.0263 & 0.0015 & 0.0046 & 0.0097 & 0.0010 & 48.8 \\ 
  Vecchia (M=5) & 0.60 & 0.60 & 0.0301 & 0.0014 & 0.0035 & 0.0074 & 0.0008 & 52.2 \\ 
  Vecchia (M=10) & 0.83 & 0.89 & 0.0160 & 0.0007 & 0.0030 & 0.0063 & 0.0007 & 68.6 \\ 
  Vecchia (M=15) & 0.87 & 0.92 & 0.0193 & 0.0008 & 0.0029 & 0.0060 & 0.0006 & 92.6 \\ 
  Taper (90\% sparse) & 0.00 & 0.00 & 0.4490 & 0.3344 & 0.0093 & 0.0249 & 0.0019 & 52.8 \\ 
  Taper (20\% sparse) & 0.00 & 0.00 & 0.4490 & 0.2942 & 0.0054 & 0.0124 & 0.0012 & 95.0 \\ 
  Low Rank (100) & 0.00 &  & 0.2487 &  & 0.0088 & 0.0296 & 0.0018 & 51.1 \\ 
  Low Rank (200) & 0.00 &  & 0.2146 &  & 0.0075 & 0.0314 & 0.0013 & 51.0 \\
       NBE & 1.00 & 0.90 &0.0170 & 0.0015& - & - & - & 0.14\\

   \hline
\end{tabular}
\normalsize
\end{table}
\begin{table}[htbp]
\small
\centering
\caption{Simulation study results for the low-dimensional case with dependence $\alpha=0.7$, smoothness $\nu=1.5$, and range $\rho=0.158$. We report the coverage and average interval scores for $\alpha$ and $\rho$ as well as the tail-weighted CRPS (twCRPS).} 
\label{tab:sim-results-20}
\begin{tabular}{lrrrrrrrr}
   & \multicolumn{2}{c}{Coverage} & \multicolumn{2}{c}{Interval Score} & \multicolumn{3}{c}{twCRPS} & Walltime \\ 
\cline{2-3}\cline{4-5}\cline{6-8}
 & $\alpha$ & Range & $\alpha$ & Range & 1 & 2 & 3 & (min) \\ \hline
 Full GP & 0.89 & 0.90 & 0.0207 & 0.0010 & 0.0025 & 0.0054 & 0.0006 & 57.2 \\ 
  Vecchia (M=3) & 0.46 & 0.42 & 0.0286 & 0.0016 & 0.0042 & 0.0090 & 0.0009 & 47.9 \\ 
  Vecchia (M=5) & 0.54 & 0.49 & 0.0279 & 0.0013 & 0.0032 & 0.0069 & 0.0007 & 51.8 \\ 
  Vecchia (M=10) & 0.80 & 0.87 & 0.0208 & 0.0009 & 0.0027 & 0.0058 & 0.0006 & 67.3 \\ 
  Vecchia (M=15) & 0.83 & 0.88 & 0.0240 & 0.0011 & 0.0026 & 0.0056 & 0.0006 & 90.8 \\ 
  Taper (90\% sparse) & 0.00 & 0.00 & 0.2990 & 0.3342 & 0.0091 & 0.0245 & 0.0018 & 52.6 \\ 
  Taper (20\% sparse) & 0.00 & 0.00 & 0.2990 & 0.3232 & 0.0052 & 0.0121 & 0.0011 & 94.6 \\ 
  Low Rank (100) & 0.00 &  & 0.1140 &  & 0.0083 & 0.0314 & 0.0016 & 50.6 \\ 
  Low Rank (200) & 0.01 &  & 0.0928 &  & 0.0078 & 0.0412 & 0.0012 & 50.7 \\ 
       NBE & 0.98 & 0.91 &0.0205 & 0.0019& - & - & - & 0.14\\

   \hline
\end{tabular}
\normalsize
\end{table}

\subsection{Figures}
Figures~\ref{fig:lowdim-1}--\ref{fig:lowdim-5} report performance comparisons among the Full GP, Vecchia approximations, covariance tapering, low-rank methods, and the NBE approach. The figures report the posterior medians of $\alpha$ across methods as well as the computational walltimes. 

\begin{figure}[htbp]
\centering

\lowcase{1}{0.05}{0.5}{0.05}\hfill
\lowcase{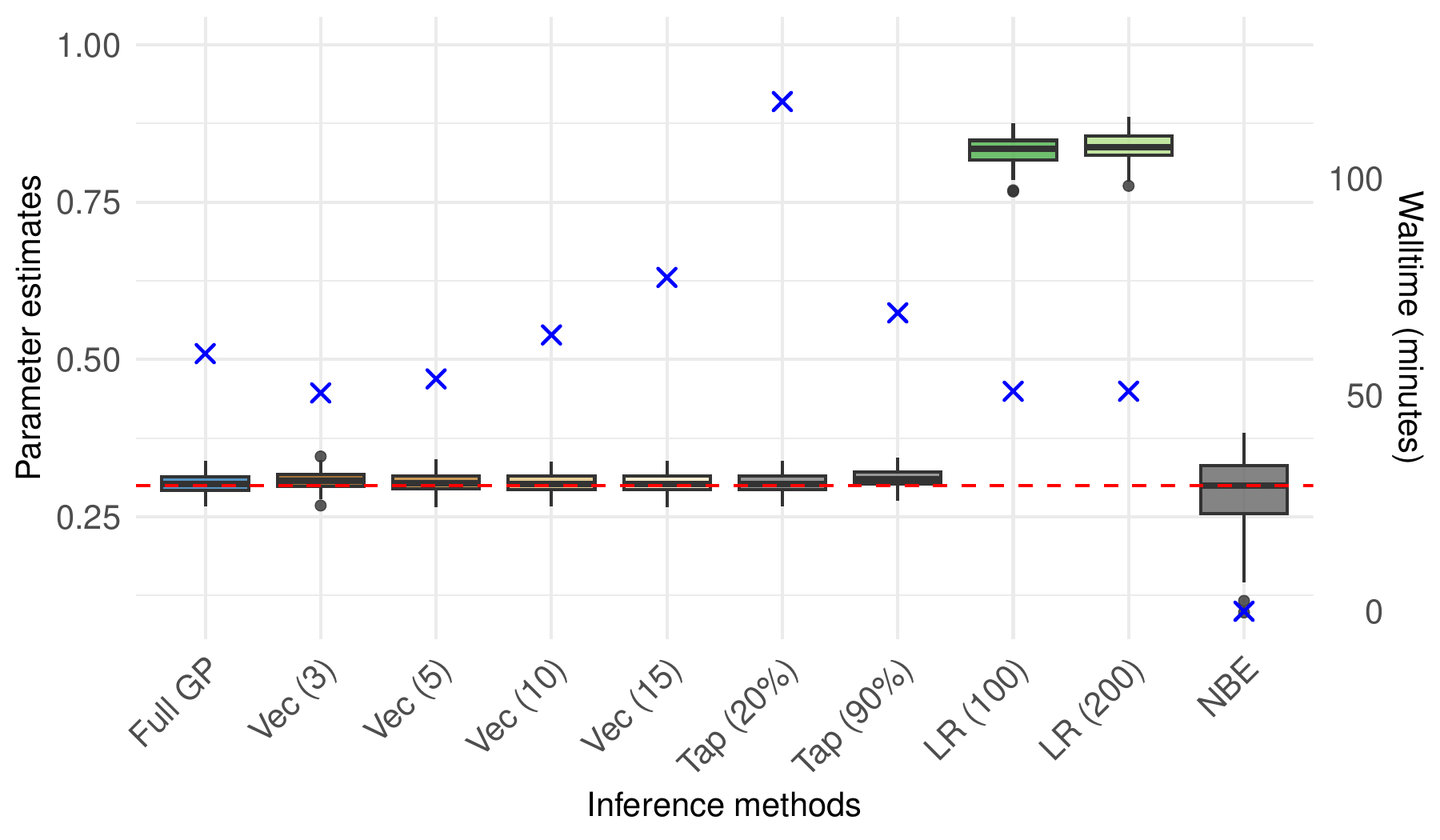}{0.25}{0.5}{0.05}

\vspace{0.35cm}

\lowcase{3}{0.45}{0.5}{0.05}\hfill
\lowcase{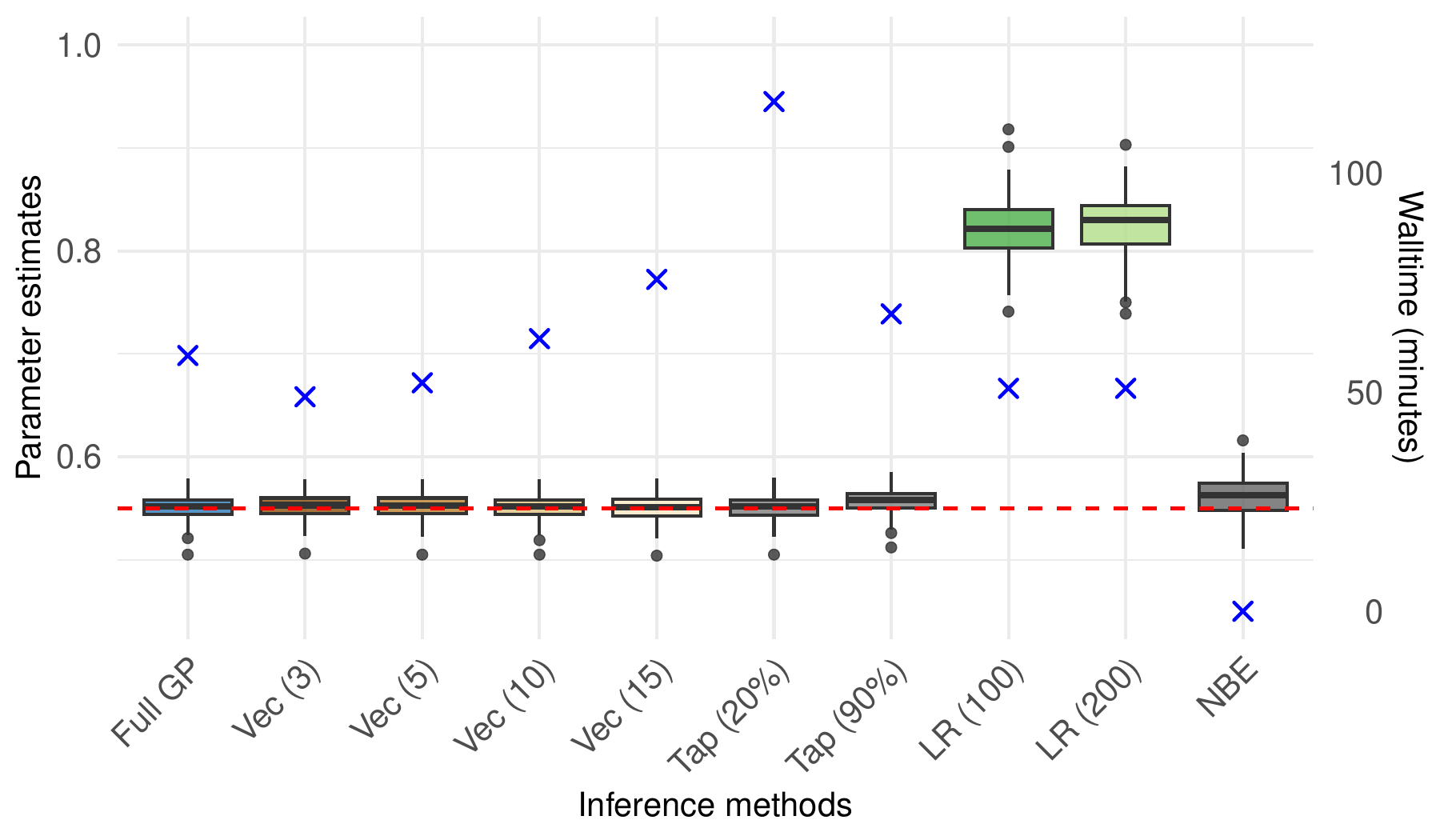}{0.70}{0.5}{0.05}

\caption{Simulation results for Cases 1--4.}
\label{fig:lowdim-1}
\end{figure}

\begin{figure}[htbp]
\centering

\lowcase{5}{0.05}{0.5}{0.10}\hfill
\lowcase{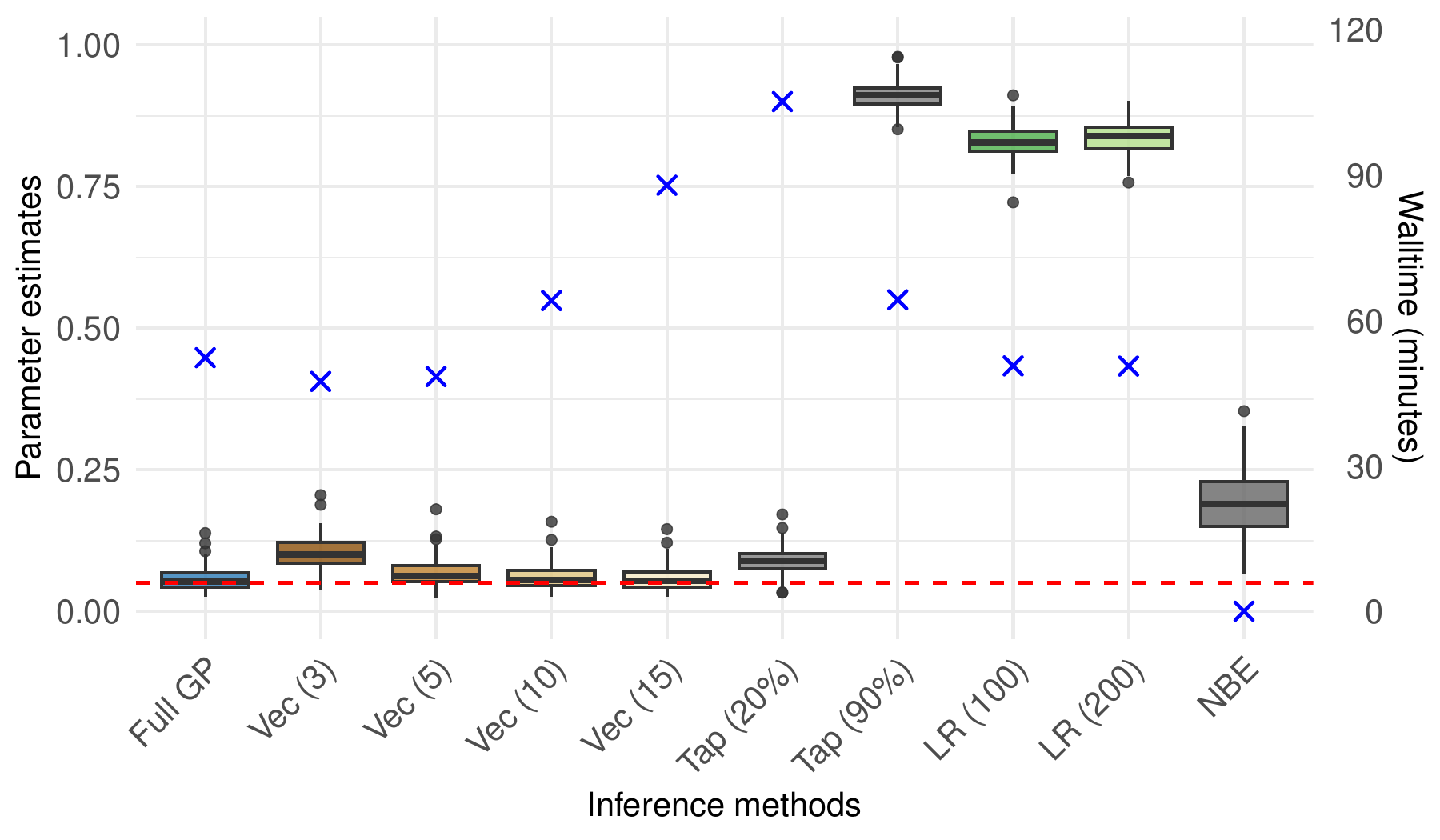}{0.25}{0.5}{0.10}

\vspace{0.35cm}

\lowcase{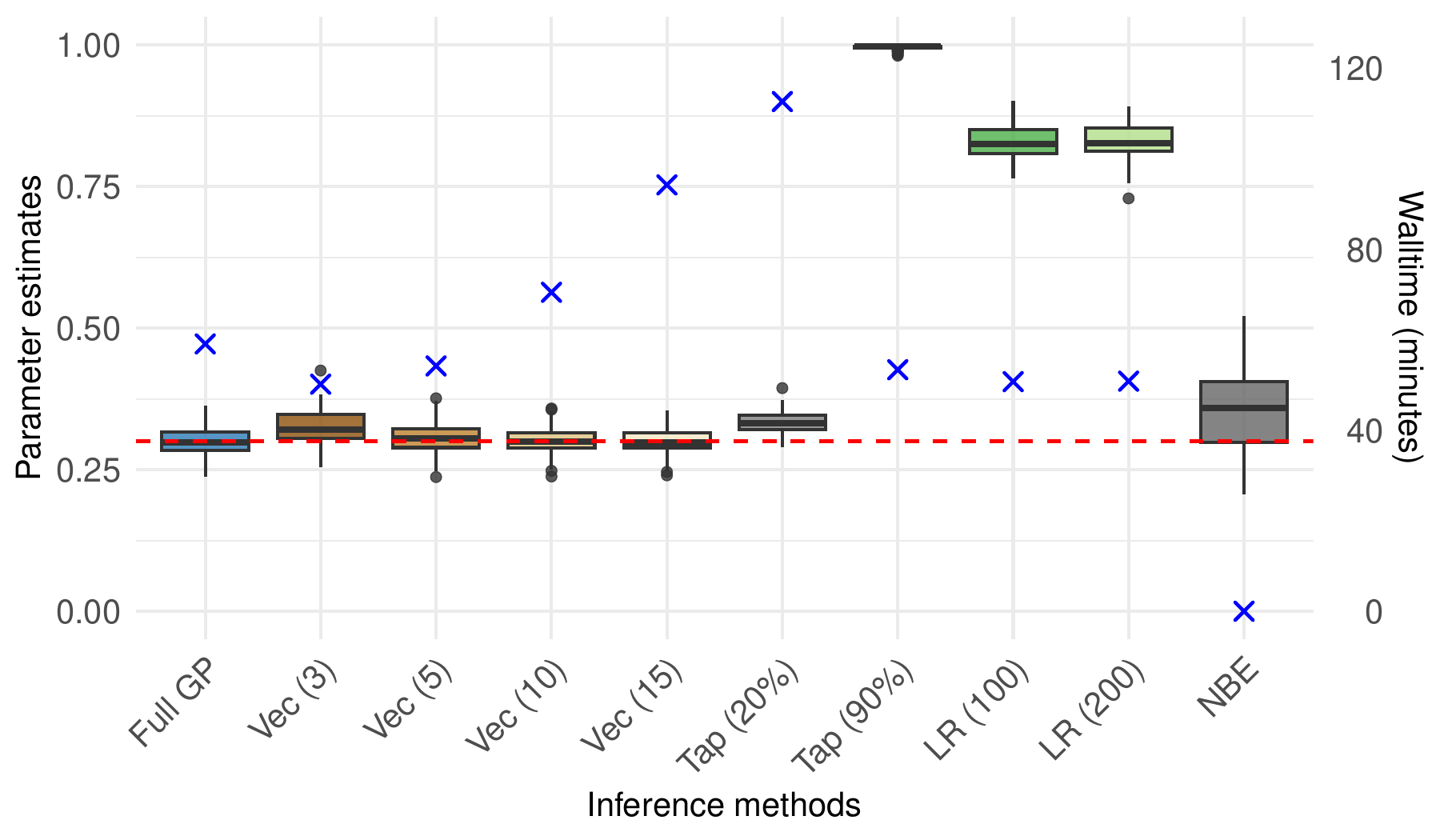}{0.45}{0.5}{0.10}\hfill
\lowcase{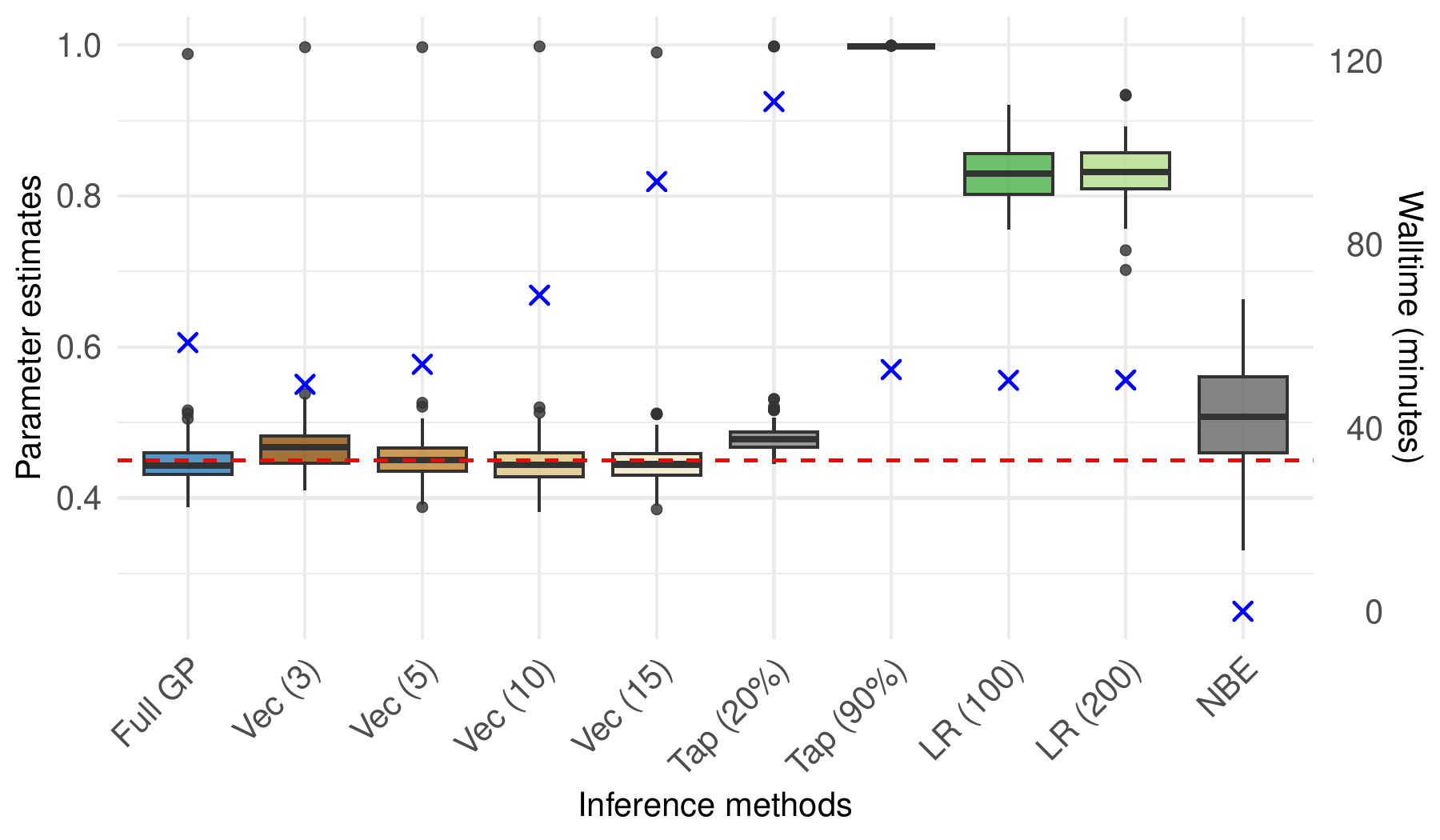}{0.70}{0.5}{0.10}

\caption{Simulation results for Cases 5--8.}
\label{fig:lowdim-2}
\vspace{0.25cm}
\centering

\lowcase{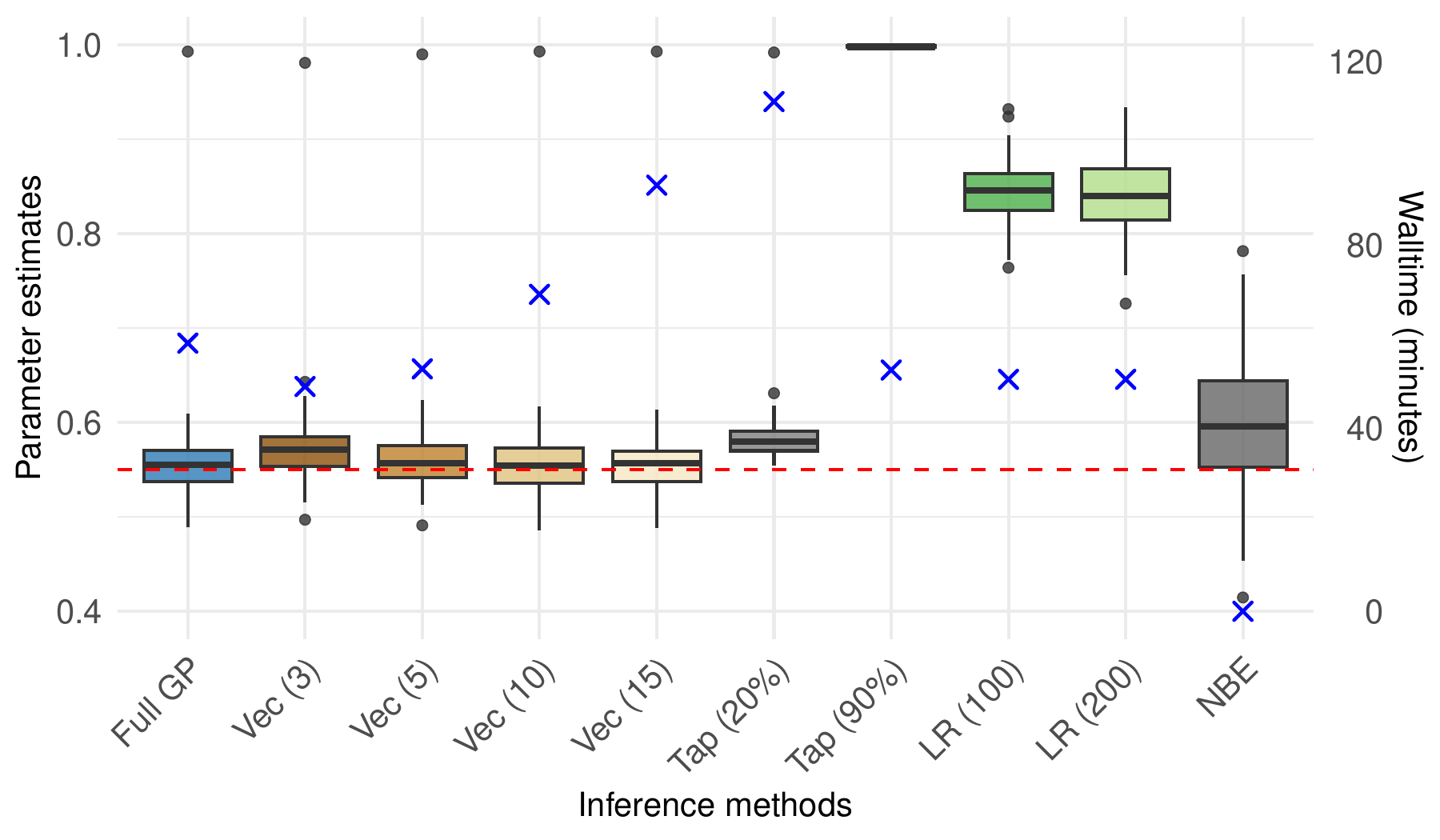}{0.05}{0.5}{0.20}\hfill
\lowcase{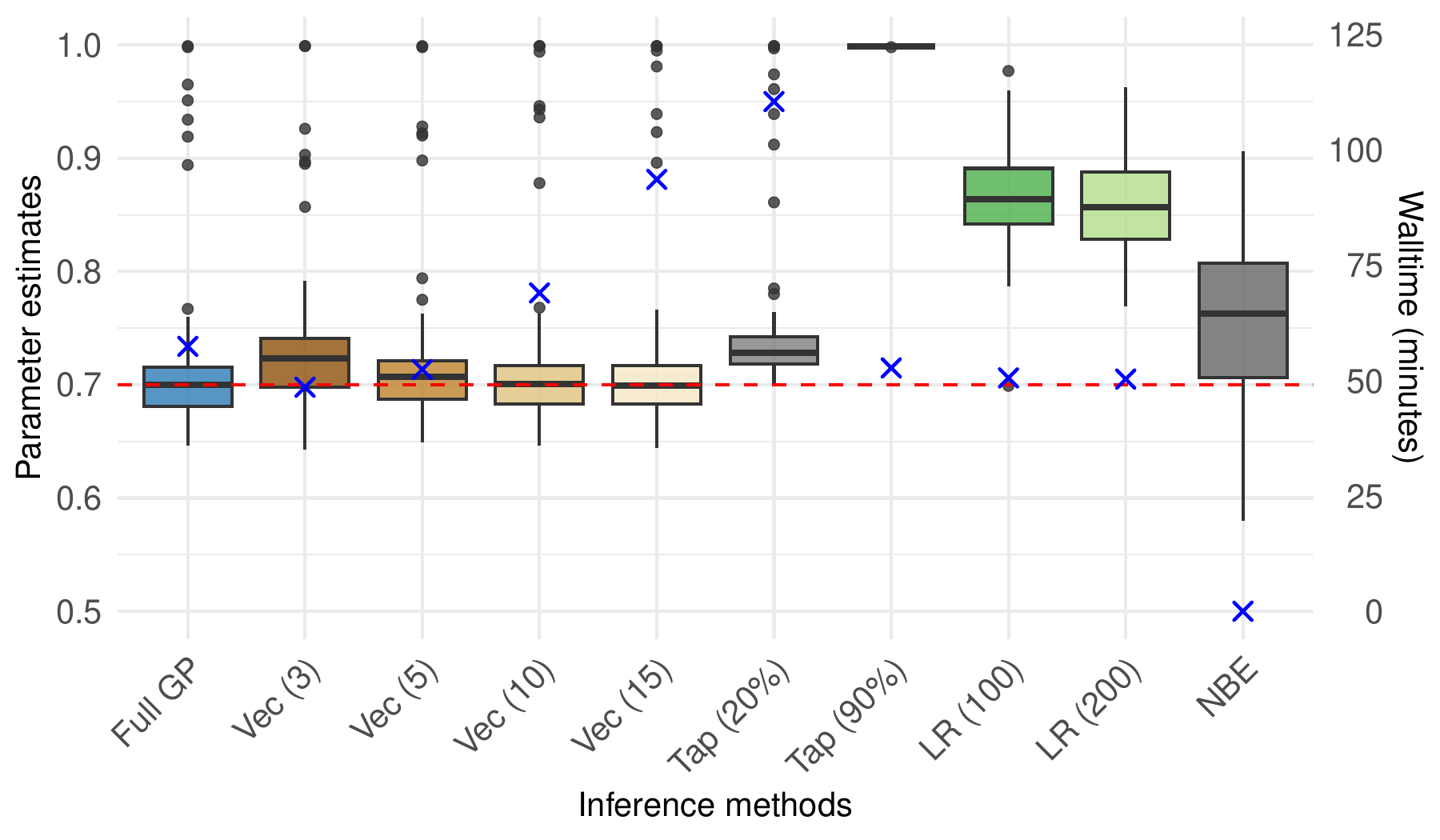}{0.25}{0.5}{0.20}

\vspace{0.35cm}

\lowcase{11}{0.45}{0.5}{0.20}\hfill
\lowcase{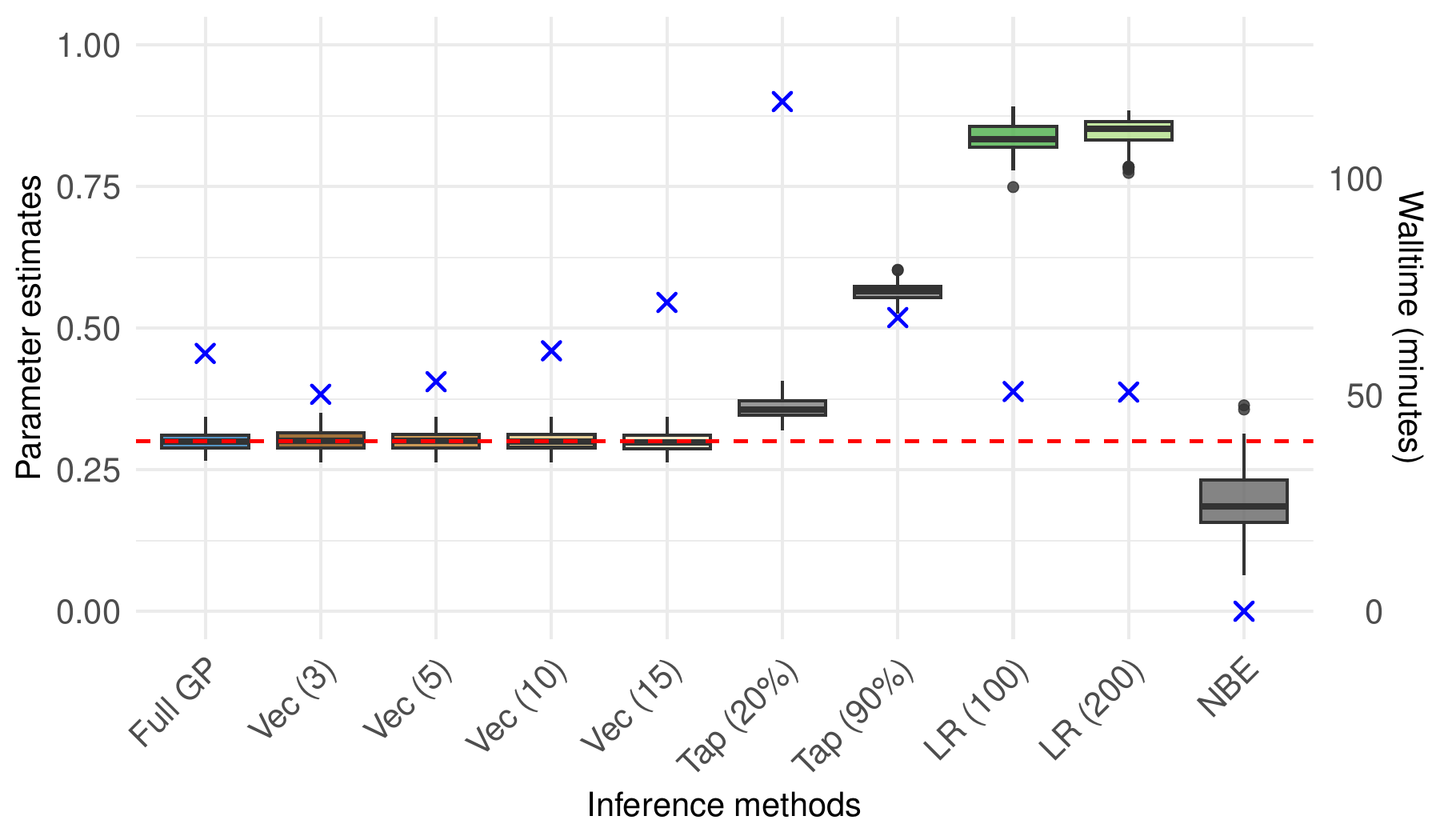}{0.70}{0.5}{0.20}

\caption{Simulation results for Cases 9--12.}
\label{fig:lowdim-3}
\end{figure}

\begin{figure}[htbp]
\centering

\lowcase{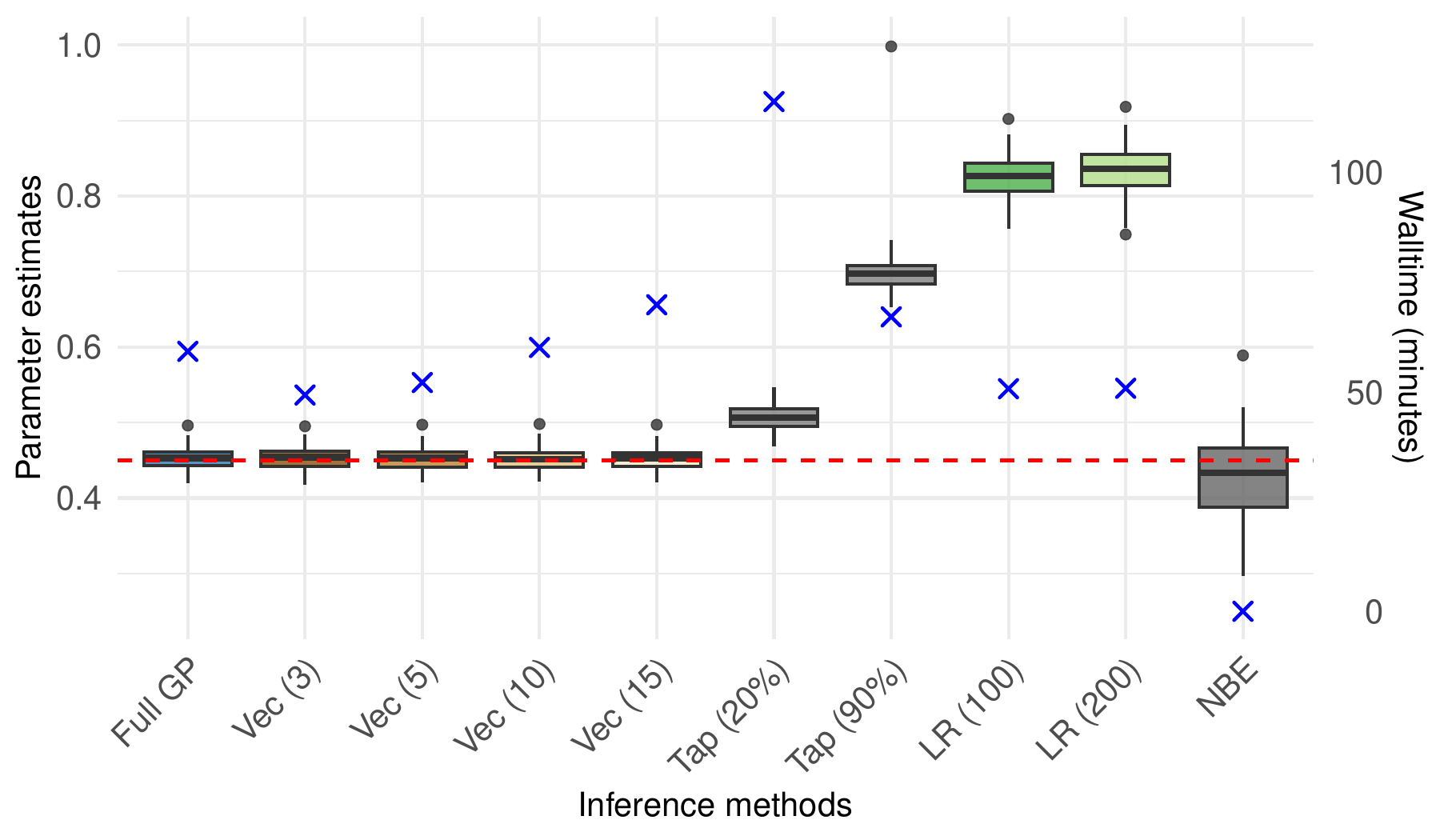}{0.05}{1.5}{0.032}\hfill
\lowcase{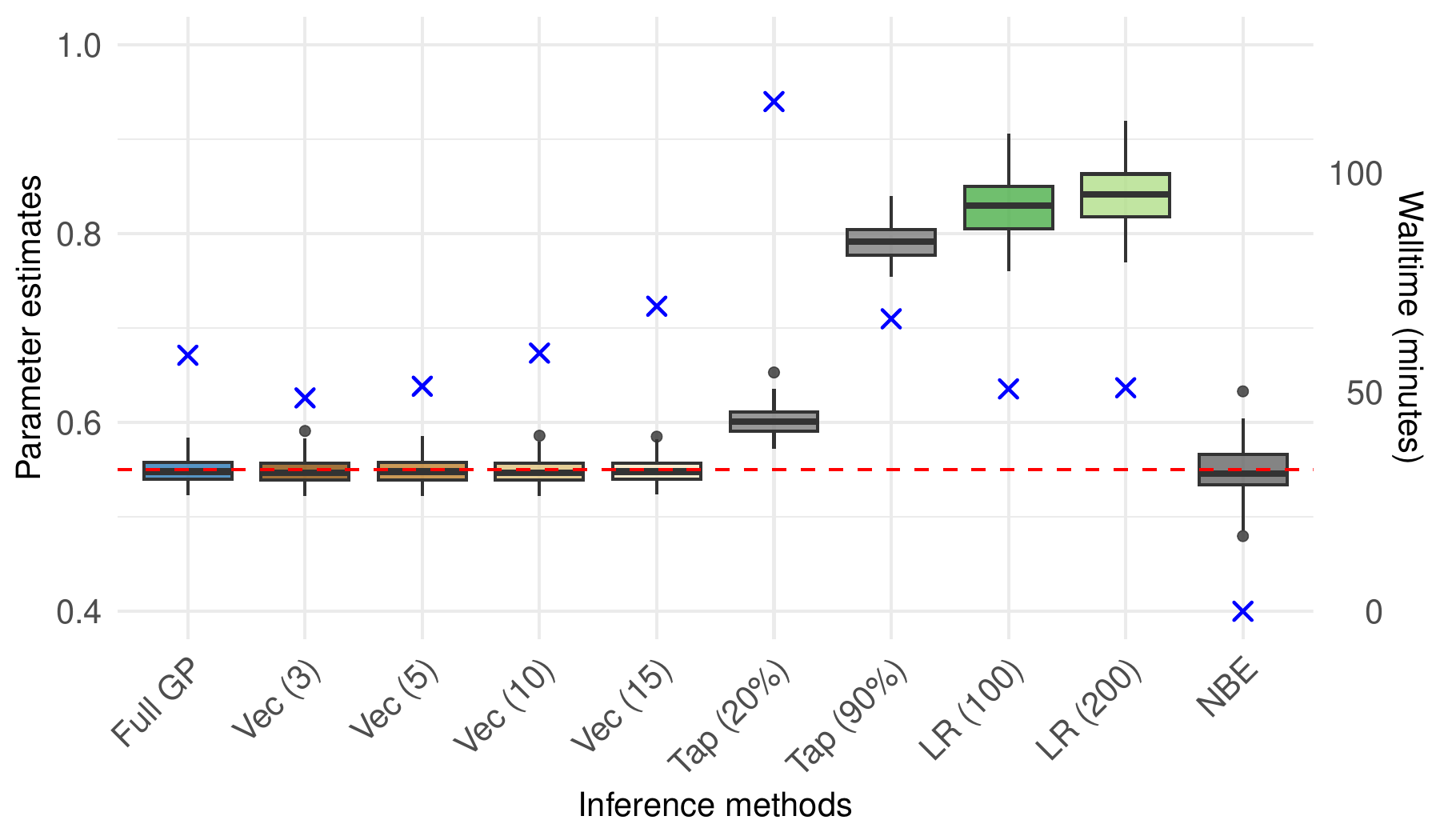}{0.25}{1.5}{0.032}

\vspace{0.35cm}

\lowcase{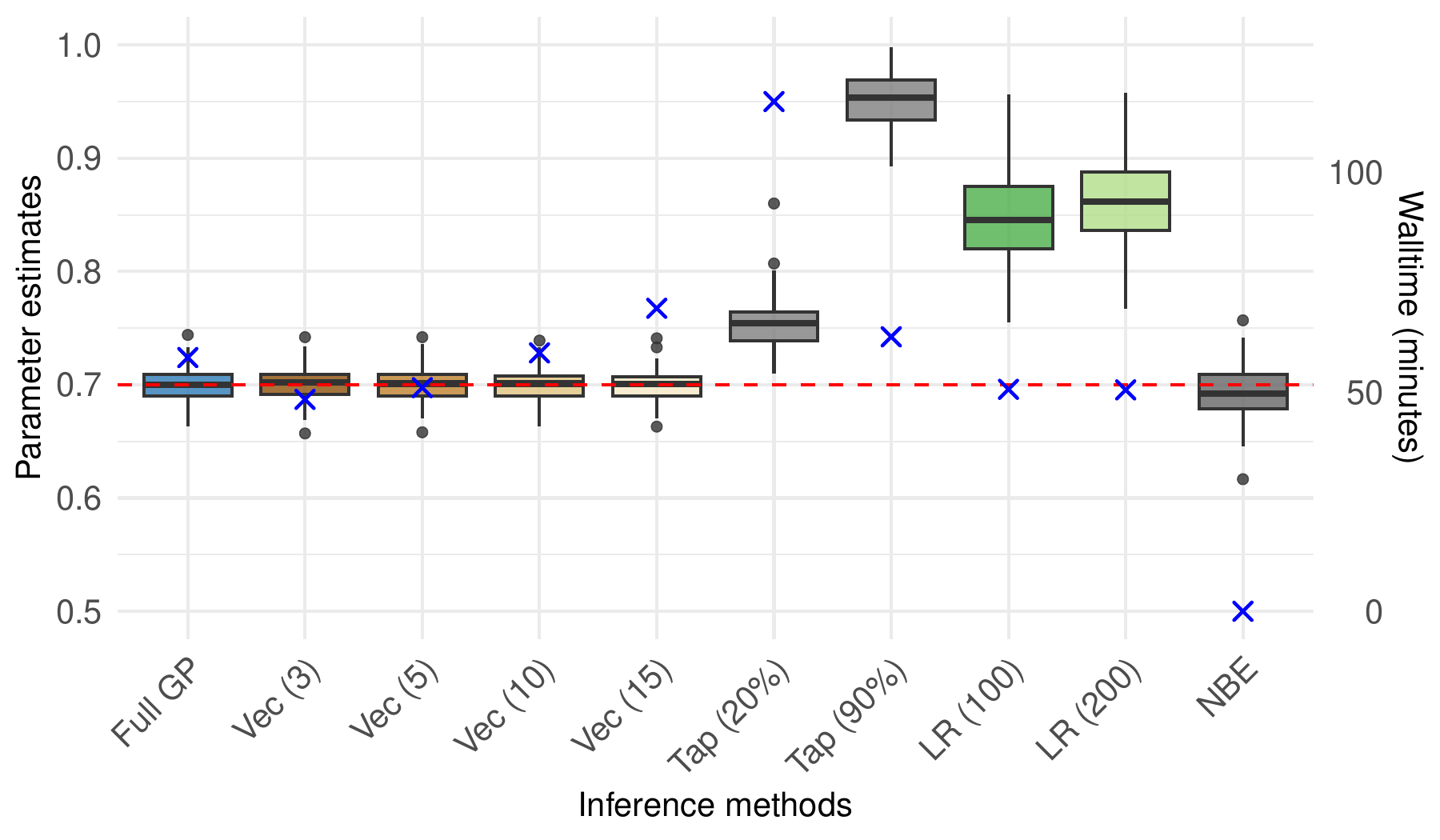}{0.45}{1.5}{0.032}\hfill
\lowcase{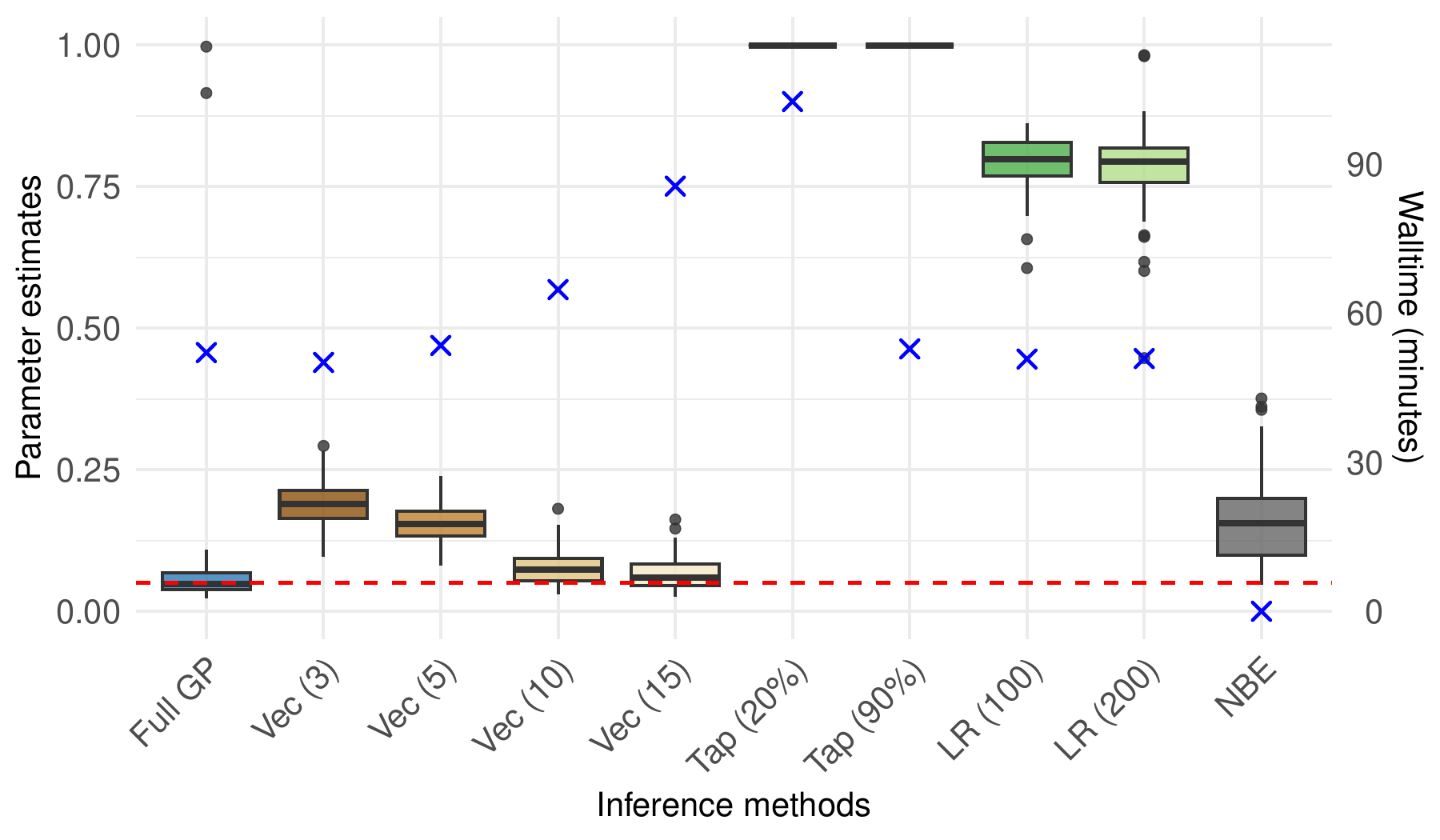}{0.70}{1.5}{0.032}

\caption{Simulation results for Cases 13--16.}
\label{fig:lowdim-4}
\vspace{0.25cm}
\centering

\lowcase{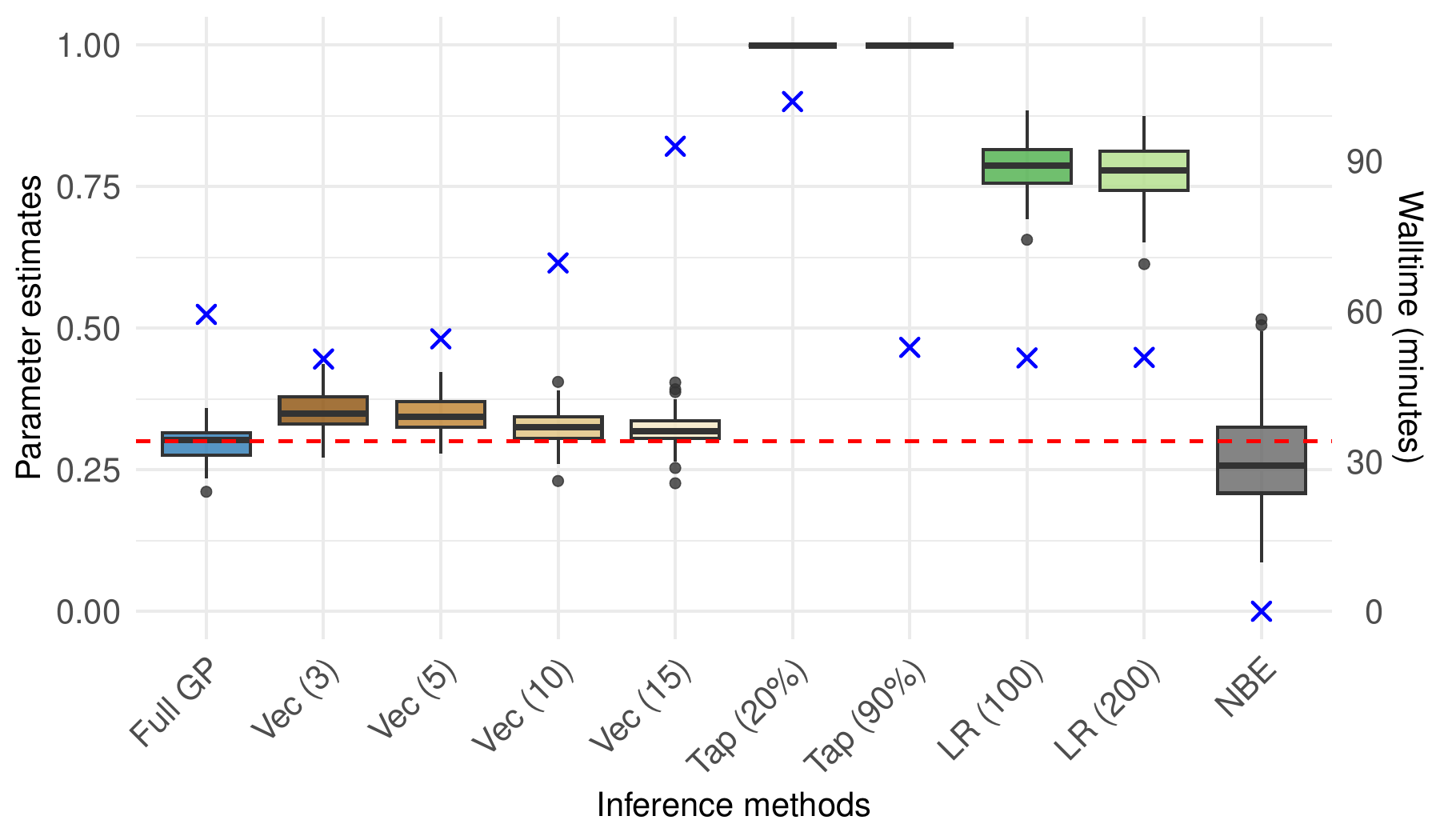}{0.05}{1.5}{0.064}\hfill
\lowcase{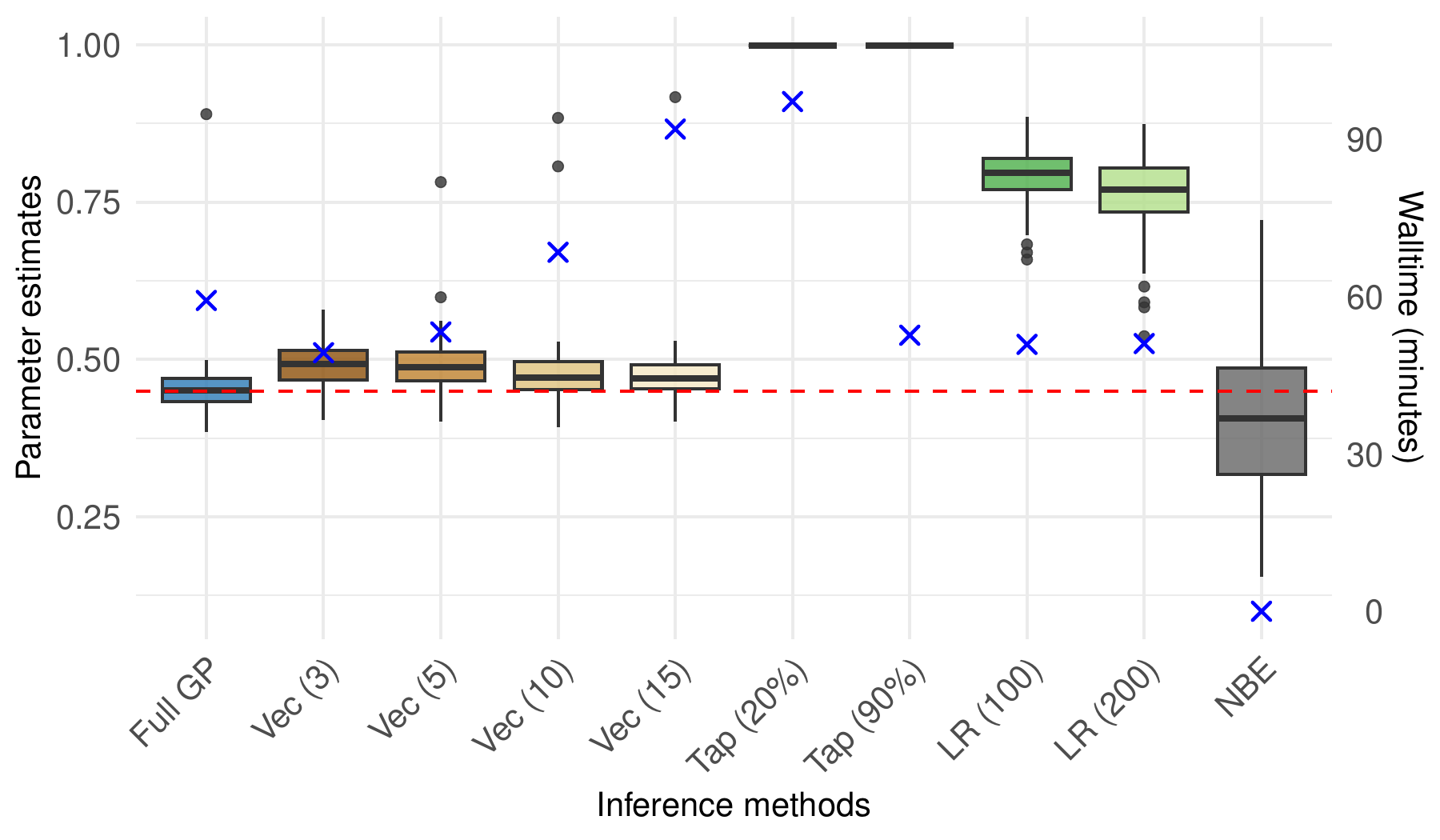}{0.25}{1.5}{0.064}

\vspace{0.35cm}

\lowcase{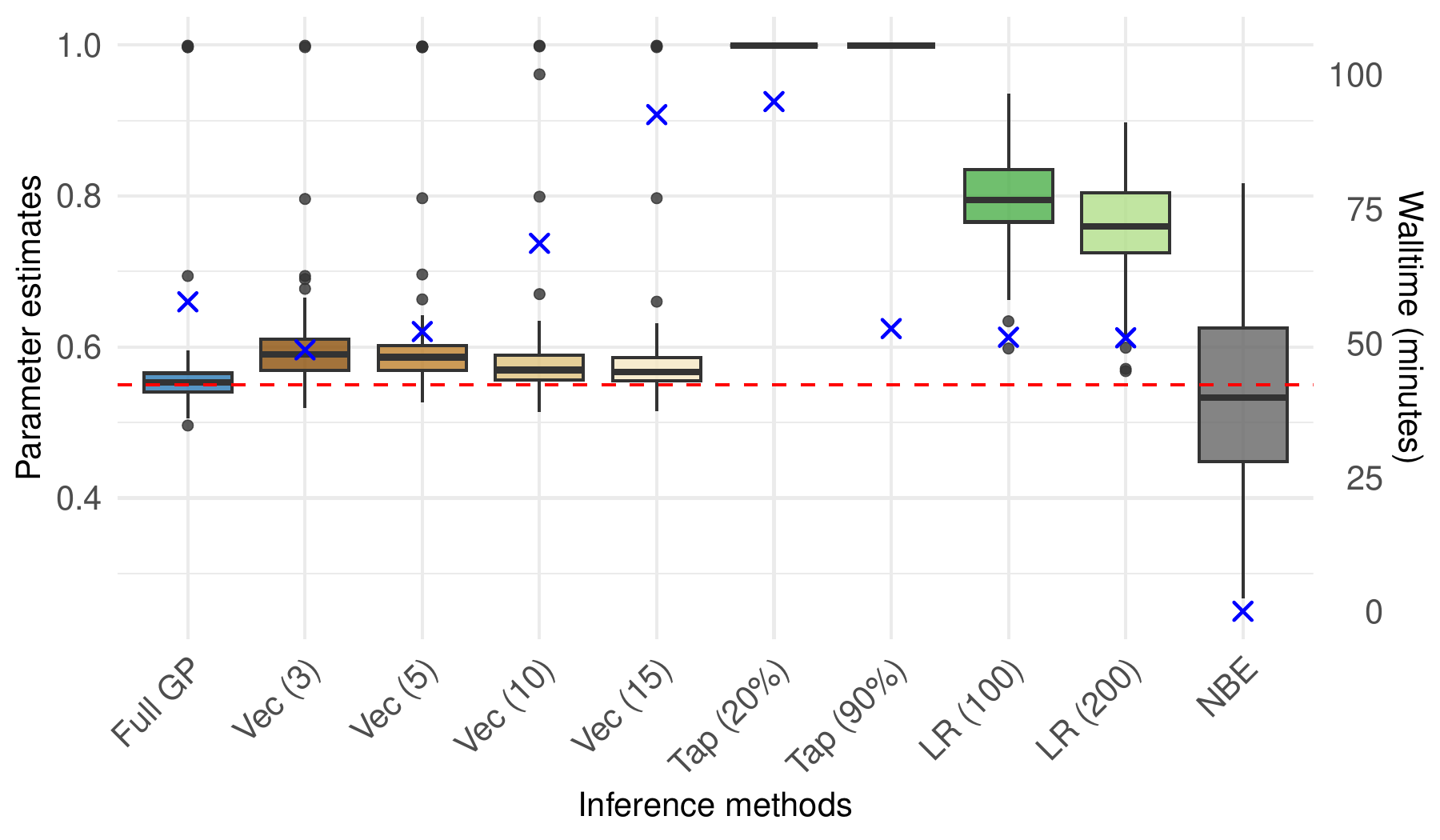}{0.45}{1.5}{0.064}\hfill
\lowcase{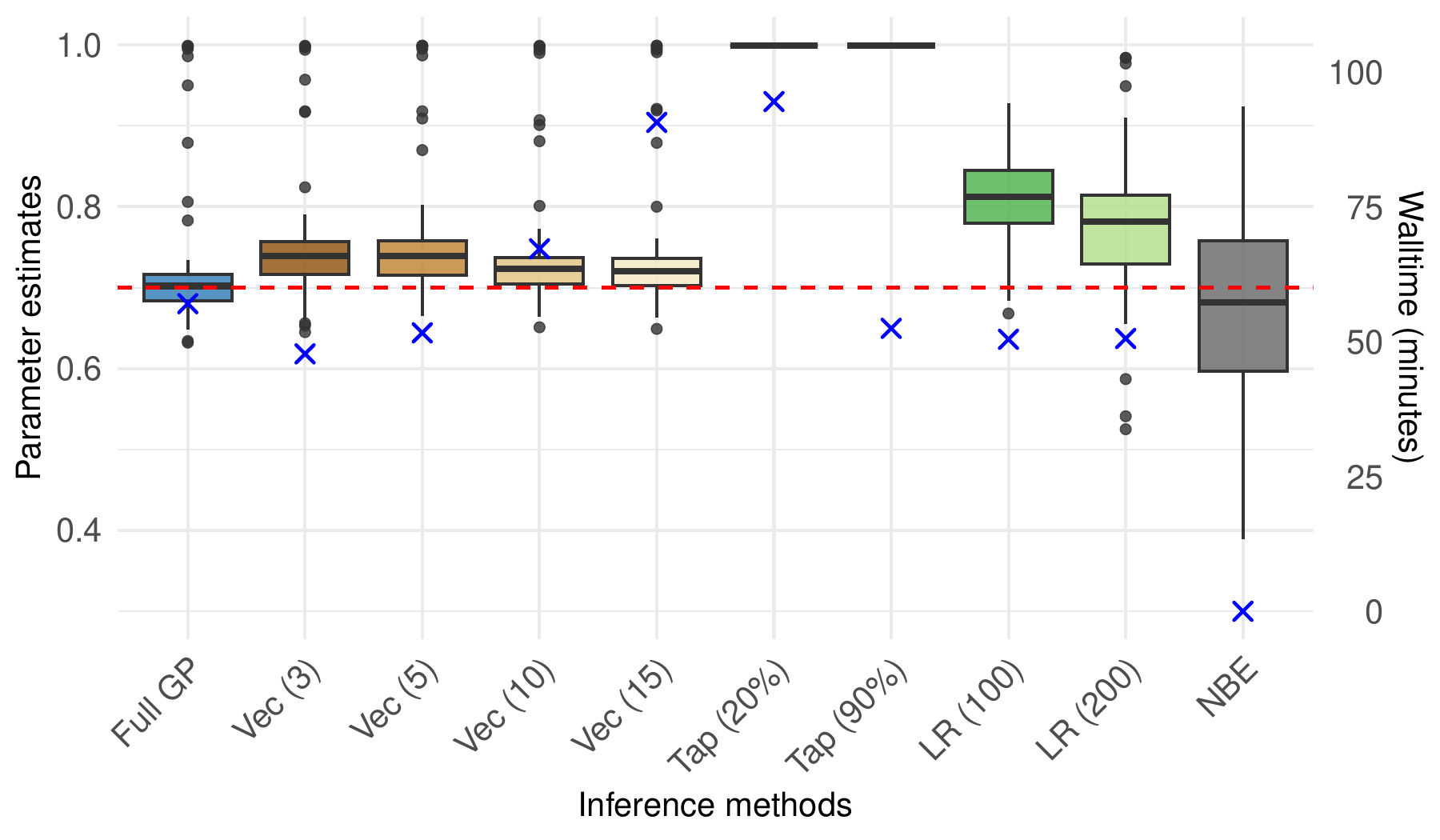}{0.70}{1.5}{0.064}

\caption{Simulation results for Cases 17--20.}
\label{fig:lowdim-5}
\end{figure}

\newpage
\section{High-dimensional Case: Simulation Study Results}
This section presents results from the high-dimensional simulation study, designed to evaluate model performance under large spatial domains and repeated temporal replicates. It supports Section \ref{s:sim_results} of the main text.

\subsection{Tables} \label{tab:highdim_results}
Tables~\ref{tab:sim-results-high1}--\ref{tab:sim-results-high8} summarize results across eight scenarios varying in dependence strength ($\alpha \in \{0.3,0.7\}$), spatial locations ($n \in \{1000, 5000\}$), and temporal replicates($T \in \{10,50\}$). We compare the Full GP (where computationally feasible), Vecchia approximations, covariance tapering, low-rank methods, and the NBE approach. For each scenario, we report empirical coverage and average interval scores for the dependence parameter $\alpha$ and range parameter $\rho$, together with tail-weighted continuous ranked probability scores (twCRPS) and computational walltime. 

\begin{table}[htbp]
\small
\centering
\caption{Simulation study results for the high-dimensional case, Scenario 1, with dependence $\alpha=0.3$, $n=5000$ spatial locations, and $T=50$ temporal replicates. We report the coverage and average interval scores for $\alpha$ and $\rho$ as well as the tail-weighted CRPS (twCRPS).} 
\label{tab:sim-results-high1}
\begin{tabular}{lrrrrrrrr}
   & \multicolumn{2}{c}{Coverage} & \multicolumn{2}{c}{Interval Score} & \multicolumn{3}{c}{twCRPS} & Walltime \\ 
\cline{2-3}\cline{4-5}\cline{6-8}
 & $\alpha$ & Range & $\alpha$ & Range & 1 & 2 & 3 & (min) \\ \hline
Full GP &  &  &  &  &  &  &  &  \\ 
  Vecchia (M=3) & 0.84 & 0.96 & 0.0101 & 0.0004 & 0.0238 & 0.0546 & 0.0052 & 312.9 \\ 
  Vecchia (M=5) & 0.91 & 0.97 & 0.0070 & 0.0003 & 0.0230 & 0.0525 & 0.0050 & 357.6 \\ 
  Vecchia (M=10) & 0.89 & 0.99 & 0.0039 & 0.0002 & 0.0226 & 0.0514 & 0.0049 & 531.6 \\ 
  Vecchia (M=15) & 0.92 & 0.98 & 0.0049 & 0.0002 & 0.0225 & 0.0513 & 0.0049 & 701.3 \\ 
  Taper (90\% sparse) &  &  &  &  &  &  &  &  \\ 
  Taper (20\% sparse) &  &  &  &  &  &  &  &  \\ 
  Low Rank (100) & 0.00 &  & 0.0385 &  & 0.0588 & 0.1275 & 0.0164 & 730.5 \\ 
  Low Rank (200) & 0.00 &  & 0.0471 &  & 0.0513 & 0.1108 & 0.0132 & 730.6 \\ 
   NBE & 0.69 & 0.93 &0.0104 & 0.0006& - & - & - & 0.14\\
   \hline
\end{tabular}
\normalsize
\end{table}

\begin{table}[htbp]
\small
\centering
\caption{Simulation study results for the high-dimensional case, Scenario 2, with dependence $\alpha=0.7$, $n=5000$ spatial locations, and $T=50$ temporal replicates. We report the coverage and average interval scores for $\alpha$ and $\rho$ as well as the tail-weighted CRPS (twCRPS).} 
\label{tab:sim-results-high2}
\begin{tabular}{lrrrrrrrr}
   & \multicolumn{2}{c}{Coverage} & \multicolumn{2}{c}{Interval Score} & \multicolumn{3}{c}{twCRPS} & Walltime \\ 
\cline{2-3}\cline{4-5}\cline{6-8}
 & $\alpha$ & Range & $\alpha$ & Range & 1 & 2 & 3 & (min) \\ \hline
  Full GP &  &  &  &  &  &  &  &  \\ 
  Vecchia (M=3) & 0.89 & 1.00 & 0.0013 & 0.0000 & 0.0304 & 0.0672 & 0.0070 & 315.4 \\ 
  Vecchia (M=5) & 0.94 & 0.99 & 0.0010 & 0.0000 & 0.0294 & 0.0645 & 0.0068 & 360.2 \\ 
  Vecchia (M=10) & 0.96 & 1.00 & 0.0010 & 0.0000 & 0.0288 & 0.0632 & 0.0066 & 534.3 \\ 
  Vecchia (M=15) & 0.94 & 1.00 & 0.0011 & 0.0000 & 0.0288 & 0.0630 & 0.0066 & 701.2 \\ 
  Taper (90\% sparse) &  &  &  &  &  &  &  &  \\ 
  Taper (20\% sparse) &  &  &  &  &  &  &  &  \\ 
  Low Rank (100) & 0.00 &  & 0.4367 &  & 0.0769 & 0.1570 & 0.0239 & 735.4 \\ 
  Low Rank (200) & 0.00 &  & 0.4441 &  & 0.0666 & 0.1359 & 0.0186 & 735.1 \\ 
  NBE & 0.99 & 0.84 &0.0077 & 0.0006& - & - & - & 0.14\\
   \hline
\end{tabular}
\normalsize
\end{table}

\begin{table}[htbp]
\small
\centering
\caption{Simulation study results for the high-dimensional case, Scenario 3, with dependence $\alpha=0.3$, $n=1000$ spatial locations, and $T=50$ temporal replicates. We report the coverage and average interval scores for $\alpha$ and $\rho$ as well as the tail-weighted CRPS (twCRPS).} 
\label{tab:sim-results-high3}
\begin{tabular}{lrrrrrrrr}
   & \multicolumn{2}{c}{Coverage} & \multicolumn{2}{c}{Interval Score} & \multicolumn{3}{c}{twCRPS} & Walltime \\ 
\cline{2-3}\cline{4-5}\cline{6-8}
 & $\alpha$ & Range & $\alpha$ & Range & 1 & 2 & 3 & (min) \\ \hline
 Full GP & 0.89 & 0.95 & 0.0064 & 0.0003 & 0.0331 & 0.0793 & 0.0071 & 129.8 \\ 
  Vecchia (M=3) & 0.88 & 0.96 & 0.0056 & 0.0002 & 0.0343 & 0.0827 & 0.0073 & 97.4 \\ 
  Vecchia (M=5) & 0.88 & 0.96 & 0.0055 & 0.0002 & 0.0335 & 0.0806 & 0.0071 & 105.6 \\ 
  Vecchia (M=10) & 0.93 & 0.98 & 0.0054 & 0.0002 & 0.0332 & 0.0795 & 0.0071 & 135.1 \\ 
  Vecchia (M=15) & 0.91 & 0.96 & 0.0055 & 0.0002 & 0.0331 & 0.0794 & 0.0071 & 171.9 \\ 
  Taper (90\% sparse) & 0.87 & 0.00 & 0.0087 & 0.0286 & 0.0331 & 0.0792 & 0.0070 & 132.3 \\ 
  Taper (20\% sparse) & 0.91 & 0.02 & 0.0064 & 0.0042 & 0.0331 & 0.0793 & 0.0071 & 355.3 \\ 
  Low Rank (100) & 0.00 &  & 0.1230 &  & 0.0614 & 0.1337 & 0.0163 & 147.7 \\ 
  Low Rank (200) & 0.00 &  & 0.1337 &  & 0.0561 & 0.1217 & 0.0135 & 149.8 \\ 
  NBE & 0.79 & 1.00 &0.0092 & 0.0005& - & - & - & 0.14\\
   \hline
\end{tabular}
\normalsize
\end{table}

\begin{table}[htbp]
\small
\centering
\caption{Simulation study results for the high-dimensional case, Scenario 4, with dependence $\alpha=0.7$, $n=1000$ spatial locations, and $T=50$ temporal replicates. We report the coverage and average interval scores for $\alpha$ and $\rho$ as well as the tail-weighted CRPS (twCRPS).} 
\label{tab:sim-results-high4}
\begin{tabular}{lrrrrrrrr}
   & \multicolumn{2}{c}{Coverage} & \multicolumn{2}{c}{Interval Score} & \multicolumn{3}{c}{twCRPS} & Walltime \\ 
\cline{2-3}\cline{4-5}\cline{6-8}
 & $\alpha$ & Range & $\alpha$ & Range & 1 & 2 & 3 & (min) \\ \hline
Full GP & 0.95 & 0.99 & 0.0016 & 0.0001 & 0.0425 & 0.0983 & 0.0096 & 131.6 \\ 
  Vecchia (M=3) & 0.92 & 0.97 & 0.0018 & 0.0001 & 0.0440 & 0.1026 & 0.0099 & 99.3 \\ 
  Vecchia (M=5) & 0.93 & 0.98 & 0.0017 & 0.0001 & 0.0430 & 0.0999 & 0.0097 & 107.7 \\ 
  Vecchia (M=10) & 0.94 & 0.99 & 0.0016 & 0.0001 & 0.0426 & 0.0986 & 0.0096 & 137.1 \\ 
  Vecchia (M=15) & 0.95 & 0.99 & 0.0016 & 0.0001 & 0.0425 & 0.0984 & 0.0096 & 174.8 \\ 
  Taper (90\% sparse) & 0.86 & 0.00 & 0.0021 & 0.0293 & 0.0425 & 0.0982 & 0.0095 & 134.4 \\ 
  Taper (20\% sparse) & 0.94 & 0.00 & 0.0017 & 0.0043 & 0.0425 & 0.0983 & 0.0096 & 357.6 \\ 
  Low Rank (100) & 0.00 &  & 0.5074 &  & 0.0804 & 0.1638 & 0.0235 & 148.6 \\ 
  Low Rank (200) & 0.00 &  & 0.5171 &  & 0.0728 & 0.1483 & 0.0186 & 150.1 \\ 
  NBE & 0.99 & 0.98 &0.0076 & 0.0004& - & - & - & 0.14\\
   \hline
\end{tabular}
\normalsize
\end{table}

\begin{table}[htbp]
\small
\centering
\caption{Simulation study results for the high-dimensional case, Scenario 5, with dependence $\alpha=0.3$, $n=5000$ spatial locations, and $T=10$ temporal replicates. We report the coverage and average interval scores for $\alpha$ and $\rho$ as well as the tail-weighted CRPS (twCRPS).} 
\label{tab:sim-results-high5}
\begin{tabular}{lrrrrrrrr}
   & \multicolumn{2}{c}{Coverage} & \multicolumn{2}{c}{Interval Score} & \multicolumn{3}{c}{twCRPS} & Walltime \\ 
\cline{2-3}\cline{4-5}\cline{6-8}
 & $\alpha$ & Range & $\alpha$ & Range & 1 & 2 & 3 & (min) \\ \hline
 Full GP &  &  &  &  &  &  &  &  \\ 
  Vecchia (M=3) & 0.89 & 0.96 & 0.0028 & 0.0001 & 0.0232 & 0.0547 & 0.0053 & 111.3 \\ 
  Vecchia (M=5) & 0.92 & 0.97 & 0.0026 & 0.0001 & 0.0224 & 0.0526 & 0.0051 & 147.9 \\ 
  Vecchia (M=10) & 0.95 & 0.98 & 0.0024 & 0.0001 & 0.0220 & 0.0515 & 0.0051 & 301.0 \\ 
  Vecchia (M=15) & 0.94 & 0.97 & 0.0053 & 0.0002 & 0.0220 & 0.0513 & 0.0050 & 532.7 \\ 
  Taper (90\% sparse) &  &  &  &  &  &  &  &  \\ 
  Taper (20\% sparse) &  &  &  &  &  &  &  &  \\ 
  Low Rank (100) & 0.00 &  & 0.1142 &  & 0.0644 & 0.2313 & 0.0560 & 292.0 \\ 
  Low Rank (200) & 0.00 &  & 0.0821 &  & 0.0561 & 0.1866 & 0.0292 & 296.5 \\ 
  NBE & 0.92 & 0.94 &0.0088 & 0.0009& - & - & - & 0.14\\
   \hline
\end{tabular}
\normalsize
\end{table}

\begin{table}[htbp]
\small
\centering
\caption{Simulation study results for the high-dimensional case, Scenario 6, with dependence $\alpha=0.7$, $n=5000$ spatial locations, and $T=10$ temporal replicates. We report the coverage and average interval scores for $\alpha$ and $\rho$ as well as the tail-weighted CRPS (twCRPS).} 
\label{tab:sim-results-high6}
\begin{tabular}{lrrrrrrrr}
   & \multicolumn{2}{c}{Coverage} & \multicolumn{2}{c}{Interval Score} & \multicolumn{3}{c}{twCRPS} & Walltime \\ 
\cline{2-3}\cline{4-5}\cline{6-8}
 & $\alpha$ & Range & $\alpha$ & Range & 1 & 2 & 3 & (min) \\ \hline
 Full GP &  &  &  &  &  &  &  &  \\ 
  Vecchia (M=3) & 0.87 & 0.96 & 0.0028 & 0.0001 & 0.0300 & 0.0673 & 0.0070 & 113.3 \\ 
  Vecchia (M=5) & 0.95 & 1.00 & 0.0022 & 0.0001 & 0.0290 & 0.0646 & 0.0068 & 149.3 \\ 
  Vecchia (M=10) & 0.95 & 1.00 & 0.0022 & 0.0001 & 0.0285 & 0.0633 & 0.0066 & 302.2 \\ 
  Vecchia (M=15) & 0.98 & 1.00 & 0.0021 & 0.0001 & 0.0284 & 0.0631 & 0.0066 & 533.9 \\ 
  Taper (90\% sparse) &  &  &  &  &  &  &  &  \\ 
  Taper (20\% sparse) &  &  &  &  &  &  &  &  \\ 
  Low Rank (100) & 0.00 &  & 0.4770 &  & 0.0824 & 0.2312 & 0.0462 & 292.8 \\ 
  Low Rank (200) & 0.00 &  & 0.4553 &  & 0.0729 & 0.2017 & 0.0380 & 295.1 \\ 
  NBE & 1.00 & 0.83 &0.0128 & 0.0016& - & - & - & 0.14\\
   \hline
\end{tabular}
\normalsize
\end{table}

\begin{table}[htbp]
\small
\centering
\caption{Simulation study results for the high-dimensional case, Scenario 7, with dependence $\alpha=0.3$, $n=1000$ spatial locations, and $T=10$ temporal replicates. We report the coverage and average interval scores for $\alpha$ and $\rho$ as well as the tail-weighted CRPS (twCRPS).} 
\label{tab:sim-results-high7}
\begin{tabular}{lrrrrrrrr}
   & \multicolumn{2}{c}{Coverage} & \multicolumn{2}{c}{Interval Score} & \multicolumn{3}{c}{twCRPS} & Walltime \\ 
\cline{2-3}\cline{4-5}\cline{6-8}
 & $\alpha$ & Range & $\alpha$ & Range & 1 & 2 & 3 & (min) \\ \hline
 Full GP & 0.90 & 0.97 & 0.0037 & 0.0002 & 0.0325 & 0.0794 & 0.0073 & 93.1 \\ 
  Vecchia (M=3) & 0.86 & 0.93 & 0.0041 & 0.0002 & 0.0336 & 0.0827 & 0.0075 & 48.9 \\ 
  Vecchia (M=5) & 0.89 & 0.98 & 0.0037 & 0.0002 & 0.0328 & 0.0806 & 0.0074 & 55.8 \\ 
  Vecchia (M=10) & 0.89 & 0.98 & 0.0038 & 0.0002 & 0.0326 & 0.0796 & 0.0073 & 81.7 \\ 
  Vecchia (M=15) & 0.90 & 0.98 & 0.0037 & 0.0002 & 0.0325 & 0.0794 & 0.0073 & 114.7 \\ 
  Taper (90\% sparse) & 0.86 & 0.00 & 0.0042 & 0.0257 & 0.0325 & 0.0793 & 0.0073 & 59.1 \\ 
  Taper (20\% sparse) & 0.89 & 0.18 & 0.0037 & 0.0025 & 0.0325 & 0.0794 & 0.0073 & 173.8 \\ 
  Low Rank (100) & 0.00 &  & 0.1033 &  & 0.0634 & 0.1935 & 0.0194 & 62.3 \\ 
  Low Rank (200) & 0.00 &  & 0.1104 &  & 0.0588 & 0.1786 & 0.0165 & 61.7 \\ 
  NBE & 0.92 & 0.99 &0.0097 & 0.0008& - & - & - & 0.14\\
   \hline
\end{tabular}
\normalsize
\end{table}

\begin{table}[htbp]
\small
\centering
\caption{Simulation study results for the high-dimensional case, Scenario 8, with dependence $\alpha=0.7$, $n=1000$ spatial locations, and $T=10$ temporal replicates. We report the coverage and average interval scores for $\alpha$ and $\rho$ as well as the tail-weighted CRPS (twCRPS).} 
\label{tab:sim-results-high8}
\begin{tabular}{lrrrrrrrr}
   & \multicolumn{2}{c}{Coverage} & \multicolumn{2}{c}{Interval Score} & \multicolumn{3}{c}{twCRPS} & Walltime \\ 
\cline{2-3}\cline{4-5}\cline{6-8}
 & $\alpha$ & Range & $\alpha$ & Range & 1 & 2 & 3 & (min) \\ \hline
 Full GP & 0.90 & 0.99 & 0.0044 & 0.0002 & 0.0420 & 0.0984 & 0.0096 & 94.9 \\ 
  Vecchia (M=3) & 0.87 & 0.96 & 0.0046 & 0.0002 & 0.0435 & 0.1026 & 0.0099 & 50.9 \\ 
  Vecchia (M=5) & 0.89 & 0.98 & 0.0042 & 0.0001 & 0.0426 & 0.0999 & 0.0097 & 58.0 \\ 
  Vecchia (M=10) & 0.90 & 0.98 & 0.0042 & 0.0002 & 0.0421 & 0.0986 & 0.0096 & 83.9 \\ 
  Vecchia (M=15) & 0.90 & 0.97 & 0.0044 & 0.0002 & 0.0421 & 0.0985 & 0.0096 & 117.5 \\ 
  Taper (90\% sparse) & 0.88 & 0.00 & 0.0045 & 0.0262 & 0.0420 & 0.0982 & 0.0096 & 60.6 \\ 
  Taper (20\% sparse) & 0.91 & 0.10 & 0.0043 & 0.0027 & 0.0420 & 0.0984 & 0.0096 & 175.6 \\ 
  Low Rank (100) & 0.00 &  & 0.4977 &  & 0.0830 & 0.2090 & 0.0279 & 62.5 \\ 
  Low Rank (200) & 0.00 &  & 0.5100 &  & 0.0760 & 0.1905 & 0.0226 & 61.8 \\ 
  NBE & 0.99 & 0.88 &0.0128 & 0.0012& - & - & - & 0.14\\
   \hline
\end{tabular}
\normalsize
\end{table}

\newpage
\subsection{Figures}
Figures~\ref{fig:highdim-1} and \ref{fig:highdim-2} report performance comparisons among the Full GP, Vecchia approximations, covariance tapering, low-rank methods, and the NBE approach. The figures report the posterior medians of $\alpha$ across methods as well as the computational walltimes.

\begin{figure}[!htbp]
\centering

\highcase{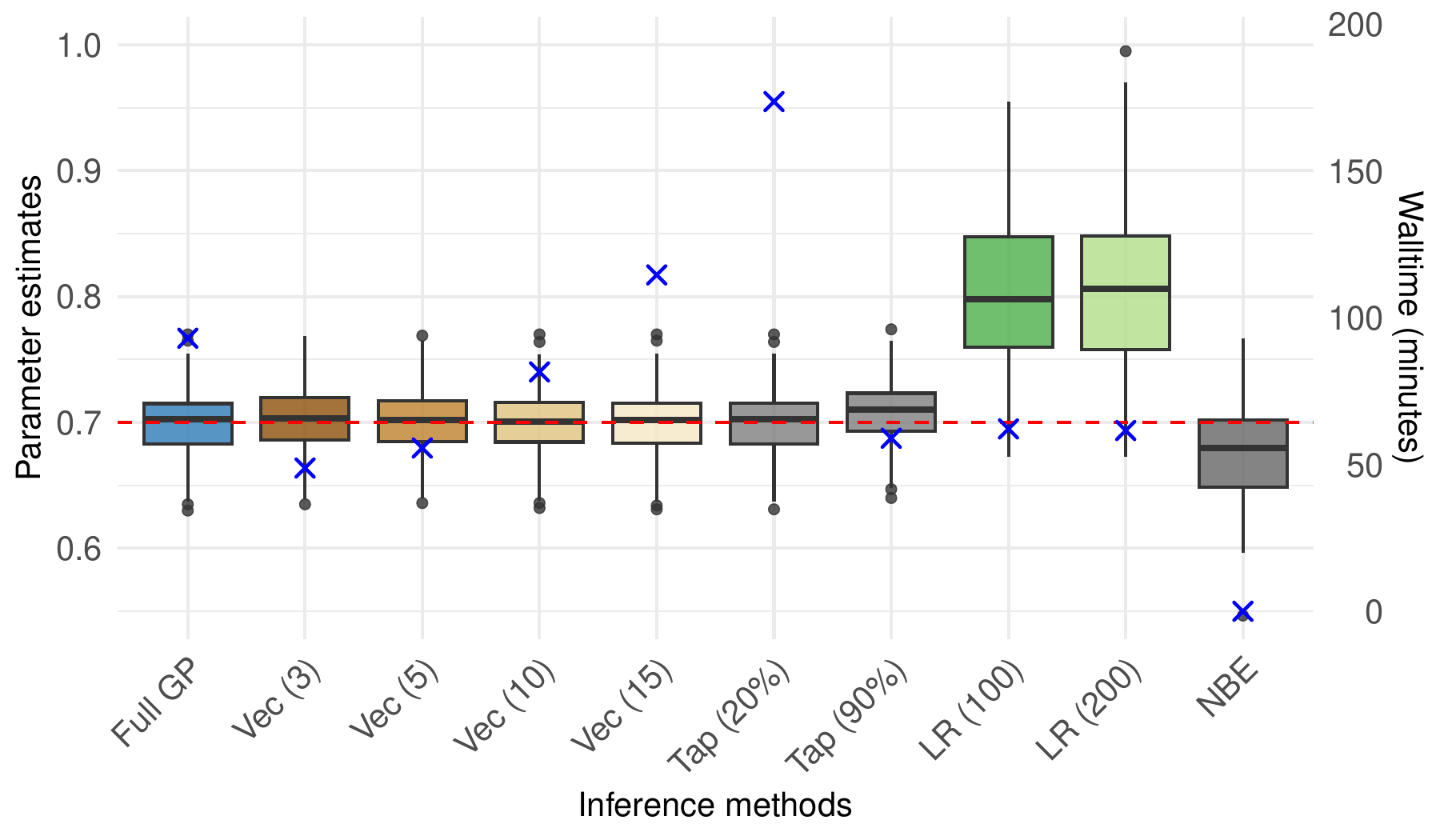}{0.3}{5000}{50}\hfill
\highcase{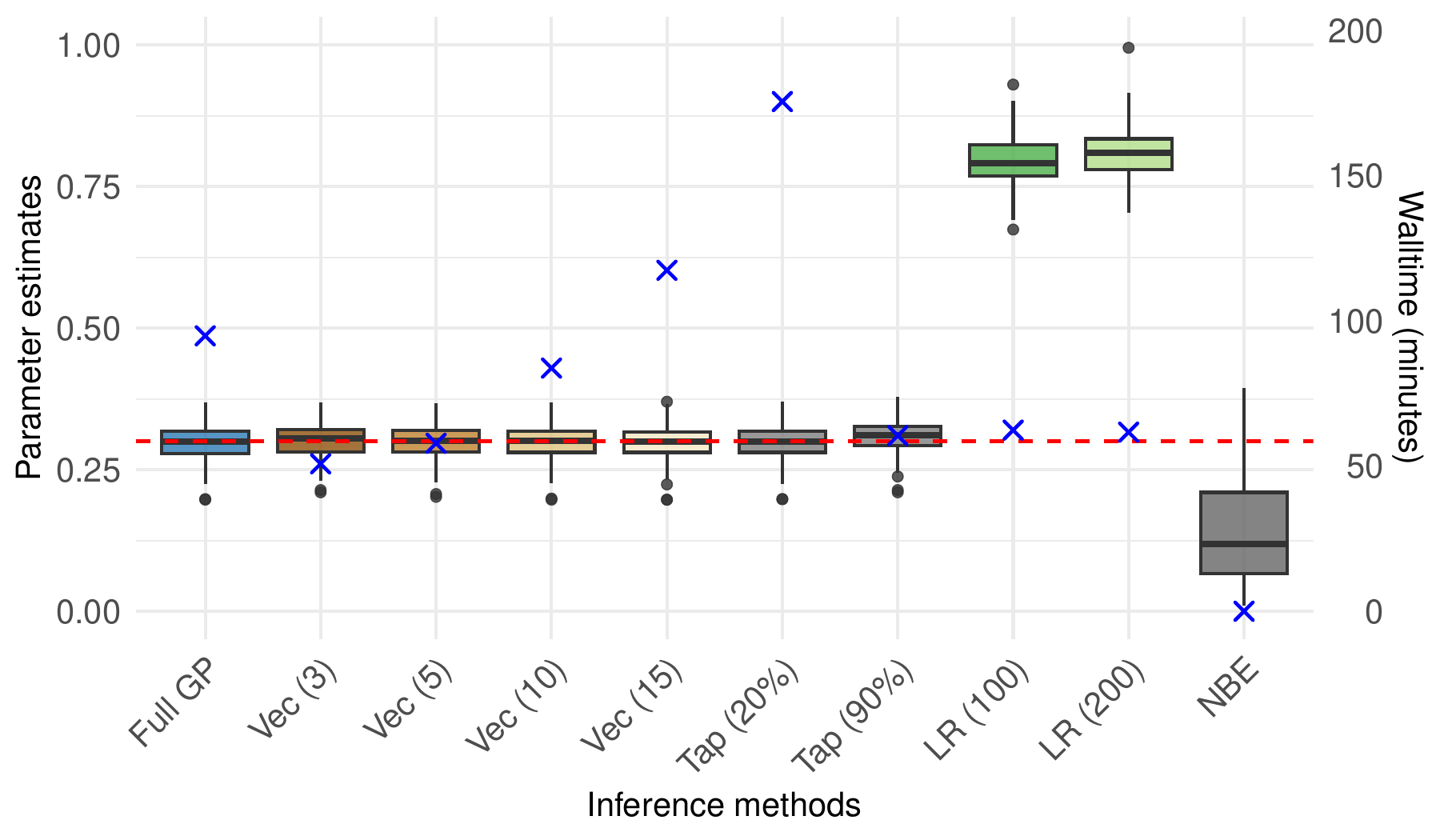}{0.7}{5000}{50}

\vspace{0.35cm}

\highcase{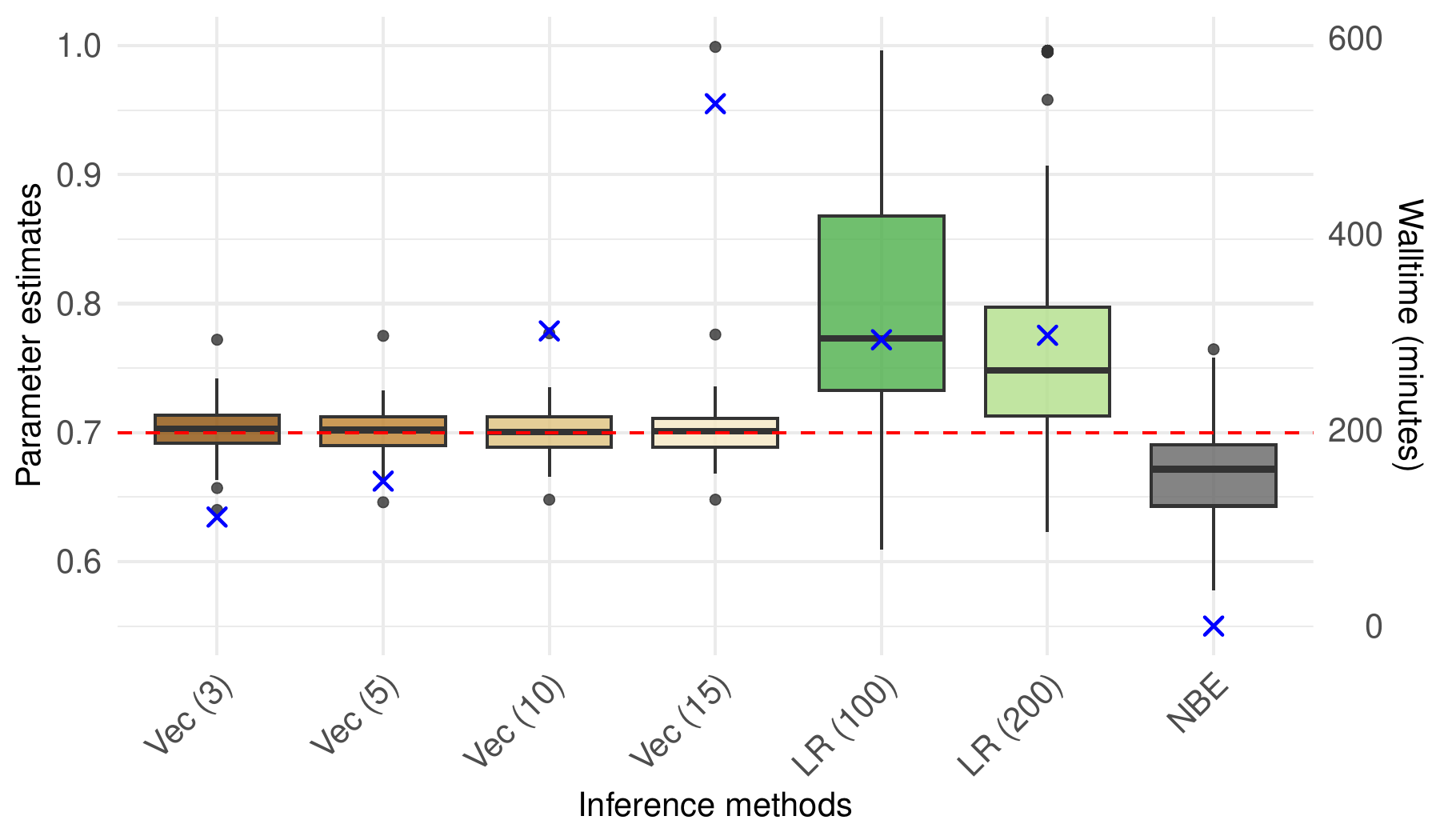}{0.3}{1000}{50}\hfill
\highcase{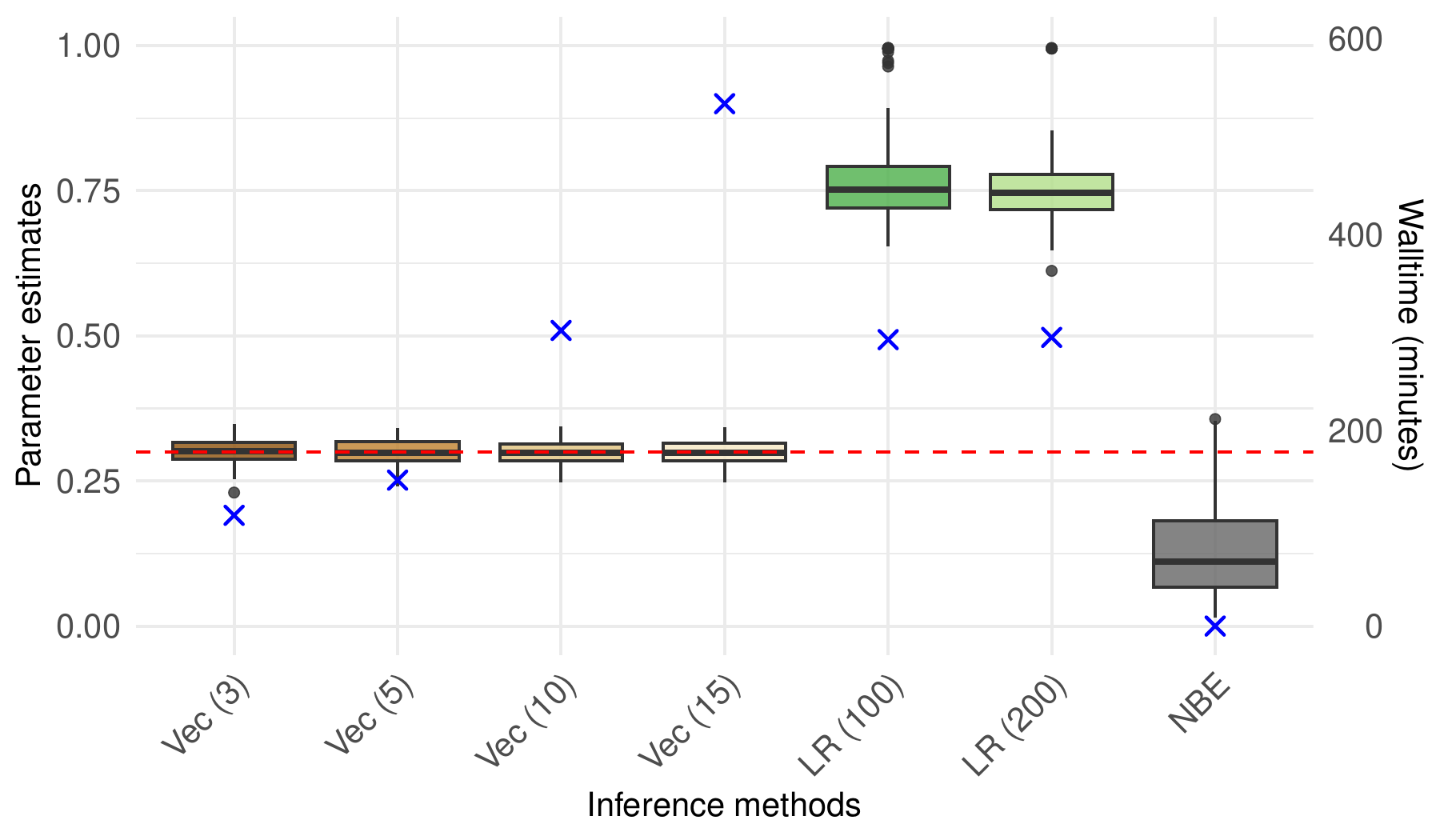}{0.7}{1000}{50}

\caption{High-dimensional simulation results for Cases 1--4.}
\label{fig:highdim-1}
\vspace{0.25cm}
\highcase{5}{0.3}{5000}{10}\hfill
\highcase{6}{0.7}{5000}{10}

\vspace{0.35cm}

\highcase{7}{0.3}{1000}{10}\hfill
\highcase{8}{0.7}{1000}{10}

\caption{High-dimensional simulation results for Cases 5--8.}
\label{fig:highdim-2}
\end{figure}

\newpage
\section{XGBoost details}\label{sec:XGBoost}

We use an XGBoost model to predict air temperature $T_{2m}(\bs)$ from surface skin temperature $T_{skin}(\bs)$ and aligned environmental covariates extracted at the timing of seasonal extremes. This approach captures nonlinear relationships while preserving the physical characteristics associated with extreme events.

\paragraph{Covariate acquisition}
We download covariates (e.g., wind speed, NDVI, and other NLDAS-derived fields) using Google Earth Engine. A key implementation detail is temporal alignment: for each location and year, we identify the time (and hour) at which the \emph{seasonal surface skin temperature extreme} occurs, and we extract all covariates at that same timestamp. This ensures the feature vector represents the physical state associated with the surface-temperature extreme rather than a daily or monthly average.

\subsection{Feature construction and training dataset}
Let $\bs$ denote a spatial location (grid cell center), and let $t$ index time. For each sampled event time $t^\star(\bs)$ (the time of the seasonal surface skin temperature extreme at $\bs$), we construct a feature vector
\begin{small}
    \begin{equation}
\begin{split}
    \mathbf{x}(\bs,t^\star) =
&\left(\text{lon}(\bs),\ \text{lat}(\bs),\ \text{elev}(\bs),\ \text{year}(t^\star),\ \sin\!\big(2\pi\,\text{doy}(t^\star)/365.25\big),\right.\\
&\left.\cos\!\big(2\pi\,\text{doy}(t^\star)/365.25\big),\ T_{skin}(\bs,t^\star),\ \text{wind}(\bs,t^\star),\ \text{NDVI}(\bs,t^\star)\right)^\top,
\end{split}
\end{equation}
\end{small}
in which $\text{doy}(t^\star)$ denotes the day of year at time $t^\star$. The response is the coincident 2-m air temperature
\begin{equation}
y(\bs,t^\star) = T_{2m}(\bs,t^\star).
\end{equation}
Because our training data are defined on NLDAS grid-cell centers while validation is conducted at off-grid GHCN stations, we implemented several practical steps to ensure that the covariates used for prediction are commensurate with the training inputs: (1) In training, elevation is treated as a grid-cell attribute. For each NLDAS grid cell, we therefore computed an \emph{areal-averaged} elevation by generating a dense $0.0125^\circ$ subgrid centered at the NLDAS location (5 points to the west/east and 5 points to the south/north; $11\times 11$ points total), querying elevations on this dense grid (with \texttt{elevatr}), and averaging them to obtain a representative cell-mean elevation. This mitigates sensitivity to local topographic roughness within a $\sim$12\,km NLDAS cell. (2) The four GHCN validation sites are airport stations, where the local land cover is predominantly paved with minimal vegetation. To reflect this micro-environment, we set $\text{NDVI}=0$ at these station locations when constructing the feature vectors for validation. This prevents spurious vegetation cooling effects from being introduced via nearby non-airport pixels and better matches the physical expectation that evapotranspiration is suppressed over impervious surfaces.

\paragraph{XGBoost model and training procedure}
We fit a gradient-boosted regression tree model using XGBoost with squared-error loss:
\begin{equation}
\widehat{f} = \arg\min_{f \in \mathcal{F}} \sum_{i=1}^n \big(y_i - f(\mathbf{x}_i)\big)^2 + \Omega(f),
\end{equation}
where $\Omega(\cdot)$ denotes the standard XGBoost complexity penalty on the ensemble of trees. The model naturally captures nonlinearities and interactions among $T_{skin}$, vegetation, wind, seasonality, and geography that govern the $T_{skin}$--$T_{2m}$ gradient.

To avoid temporal leakage, we split the dataset by time index rather than random sampling: early times are used for training, an intermediate block for validation (early stopping), and the most recent block for testing. This provides a conservative assessment of generalization under nonstationarity.

We use standard regularization controls (tree depth, minimum child weight, subsampling, and column subsampling) and early stopping on validation RMSE/MAE. Because physically, $T_{2m}$ should increase with surface temperature conditional on the other covariates, we optionally impose a monotonicity constraint on the $T_{skin}$ feature. Importantly, XGBoost produces a piecewise-constant function in each covariate: sensitivity to $T_{skin}$ appears as a set of threshold effects (``stair-step'' response) rather than a smooth slope. Feature importance (gain/cover/frequency) therefore reflects global loss reduction attributable to $T_{skin}$ splits, not a local derivative at a specific point.

\section{Case Study: Surface Temperature Analysis — Additional Results}
We provide additional results from the surface skin temperature case study described in Section~\ref{s:application} of the main text. Table~\ref{tab:data_results_alpha2} reports posterior mean estimates and 95\% credible intervals for $\nu=1.5$, while Tables~\ref{tab:data_results_rho1} and \ref{tab:data_results_rho2} summarize posterior estimates of the spatial range parameter $\rho$ for $\nu=0.5$ and $\nu=1.5$, respectively. Tail-weighted CRPS results are presented in Tables~\ref{tab:twcrps_nu1} and \ref{tab:twcrps_nu2}, and corresponding walltimes are reported in Tables~\ref{tab:walltime_nu1} and \ref{tab:walltime_nu2} for the two values of $\nu$. We also report the posterior standard deviations for the metrics reported in Table~\ref{fig:posterior_summary_2000} of the main text.  

\begin{table}[ht]
\small
\centering
\resizebox{\textwidth}{!}{\begin{tabular}{lllll}
  \hline
Method & Spring & Summer & Fall & Winter \\ 
  \hline
Full GP & 0.708 (0.689,0.726) & 0.324 (0.305,0.34) & 0.514 (0.497,0.532) & 0.9 (0.882,0.918) \\ 
  Low Rank - 100 & 0.785 (0.785,0.785) & 0.854 (0.854,0.854) & 0.813 (0.813,0.813) & 0.763 (0.762,0.762) \\ 
  Low Rank - 500 & 0.788 (0.788,0.788) & 0.839 (0.839,0.839) & 0.81 (0.809,0.809) & 0.808 (0.808,0.808) \\ 
  Taper - 90\% & 0.462 (0.446,0.478) & 0.288 (0.269,0.305) & 0.323 (0.303,0.341) & 0.381 (0.361,0.401) \\ 
  Taper - 20\% & 0.665 (0.646,0.685) & 0.314 (0.296,0.331) & 0.475 (0.456,0.495) & 0.878 (0.861,0.899) \\ 
  Vecchia - M=3 & 0.756 (0.737,0.776) & 0.292 (0.272,0.311) & 0.552 (0.533,0.572) & 0.947 (0.923,0.973) \\ 
  Vecchia - M=5 & 0.74 (0.72,0.759) & 0.32 (0.301,0.338) & 0.538 (0.519,0.556) & 0.934 (0.913,0.959) \\ 
  Vecchia - M=10 & 0.687 (0.669,0.706) & 0.32 (0.302,0.337) & 0.497 (0.476,0.514) & 0.886 (0.869,0.906) \\ 
  Vecchia - M=20 & 0.708 (0.691,0.727) & 0.323 (0.306,0.341) & 0.514 (0.496,0.532) & 0.899 (0.882,0.918) \\ 
     Neural Bayes&0.71 (0.416, 0.867)  & 0.31 (0.045, 0.339)  & 0.63 (0.266, 0.771) & 0.75 (0.530, 0.911) \\
   \hline
\end{tabular}}
\caption{Posterior mean estimates and 95\% credible intervals for the $\alpha$ parameter under the Mat\'ern covariance function with smoothness parameter $\nu = 1.5$, for the NLDAS surface skin temperature application.}
\label{tab:data_results_alpha2}
\normalsize
\end{table}

\begin{table}[ht]
\small
\centering
\resizebox{\textwidth}{!}{\begin{tabular}{lllll}
  \hline
Method & Spring & Summer & Fall & Winter \\ 
  \hline
Full GP & 0.097 (0.095,0.099) & 0.052 (0.051,0.053) & 0.092 (0.09,0.094) & 0.078 (0.077,0.08) \\ 
  Low Rank - 100 & - & - & - & - \\ 
  Low Rank - 500 & - & - & - & - \\ 
  Taper - 90\% & 0.308 (0.288,0.329) & 0.086 (0.083,0.088) & 0.289 (0.27,0.307) & 0.161 (0.152,0.169) \\ 
  Taper - 20\% & 0.12 (0.117,0.123) & 0.059 (0.057,0.06) & 0.113 (0.11,0.116) & 0.092 (0.09,0.095) \\
  Vecchia - M=3 & 0.099 (0.097,0.101) & 0.054 (0.053,0.055) & 0.096 (0.094,0.098) & 0.079 (0.077,0.081) \\ 
  Vecchia - M=5 & 0.097 (0.095,0.099) & 0.053 (0.052,0.054) & 0.093 (0.091,0.095) & 0.078 (0.075,0.08) \\ 
  Vecchia - M=10 & 0.096 (0.094,0.098) & 0.052 (0.051,0.053) & 0.091 (0.09,0.093) & 0.077 (0.075,0.079) \\ 
  Vecchia - M=20 & 0.096 (0.094,0.098) & 0.052 (0.051,0.053) & 0.092 (0.09,0.094) & 0.078 (0.076,0.08) \\ 
       Neural Bayes&0.069 (0.079, 0.10)  & 0.048 (0.046, 0.072)  & 0.067 (0.074, 0.097) & 0.065 (0.077, 0.104) \\ 
   \hline
\end{tabular}}
\caption{Posterior mean estimates and 95\% credible intervals for the spatial range parameter $\rho$  under the Mat\'ern covariance function with smoothness parameter $\nu = 0.5$, for the NLDAS surface skin temperature application.}
\label{tab:data_results_rho1}
\normalsize
\vspace{1cm}
\small
\centering
\resizebox{\textwidth}{!}{\begin{tabular}{lllll}
  \hline
Method & Spring & Summer & Fall & Winter \\ 
  \hline
Full GP & 0.022 (0.022,0.023) & 0.017 (0.017,0.017) & 0.023 (0.022,0.023) & 0.02 (0.02,0.02) \\ 
  Low Rank - 100 & - & - & - & - \\ 
  Low Rank - 500 & - & - & - & - \\ 
  Taper - 90\% & 0.12 (0.11,0.13) & 0.024 (0.023,0.024) & 0.073 (0.068,0.078) & 0.19 (0.162,0.227) \\ 
  Taper - 20\% & 0.026 (0.026,0.026) & 0.018 (0.018,0.018) & 0.026 (0.026,0.026) & 0.022 (0.022,0.022) \\ 
  Vecchia - M=3 & 0.025 (0.024,0.025) & 0.018 (0.018,0.019) & 0.025 (0.024,0.025) & 0.022 (0.022,0.022) \\ 
  Vecchia - M=5 & 0.023 (0.022,0.023) & 0.017 (0.017,0.017) & 0.023 (0.023,0.023) & 0.02 (0.02,0.02) \\ 
  Vecchia - M=10 & 0.023 (0.022,0.023) & 0.017 (0.017,0.017) & 0.023 (0.023,0.023) & 0.02 (0.02,0.02) \\ 
  Vecchia - M=20 & 0.022 (0.022,0.023) & 0.017 (0.017,0.017) & 0.023 (0.022,0.023) & 0.02 (0.02,0.02) \\ 
       Neural Bayes&0.028 (0.027, 0.036)  & 0.023 (0.020, 0.033)  & 0.025 (0.026, 0.035) & 0.029 (0.028, 0.036) \\ 
   \hline
\end{tabular}}
\caption{Posterior mean estimates and 95\% credible intervals for the spatial range parameter $\rho$  under the Mat\'ern covariance function with smoothness parameter $\nu = 1.5$, for the NLDAS surface skin temperature application.}
\normalsize
\label{tab:data_results_rho2}
\end{table}

\begin{table}[!htbp]
\vspace{0.5cm}
\small
\centering
\begin{tabular}{lllll}
  \hline
Method & Spring & Summer & Fall & Winter \\ 
  \hline
Full GP & 0.0341 & 0.0564 & 0.0373 & 0.0301 \\ 
  Low Rank - 100 & 0.0563 & 0.0890 & 0.0618 & 0.0499 \\ 
  Low Rank - 500 & 0.0509 & 0.0906 & 0.0568 & 0.0454 \\ 
  Taper - 90\% & 0.0341 & 0.0564 & 0.0373 & 0.0301 \\ 
  Taper - 20\% & 0.0342 & 0.0565 & 0.0374 & 0.0303 \\ 
  Vecchia - M=3 & 0.0352 & 0.0576 & 0.0384 & 0.0311 \\ 
  Vecchia - M=5 & 0.0342 & 0.0562 & 0.0373 & 0.0303 \\ 
  Vecchia - M=10 & 0.0342 & 0.0564 & 0.0375 & 0.0303 \\ 
  Vecchia - M=20 & 0.0341 & 0.0564 & 0.0374 & 0.0301 \\ 
       Neural Bayes& &   & & \\ 
   \hline
\end{tabular}
\caption{Tail-weighted CRPS for the NLDAS surface skin temperature application using the Mat\'ern covariance with smoothness $\nu = 0.5$.}
\normalsize
\label{tab:twcrps_nu1}
\end{table}

\begin{table}[!htbp]
\small
\centering
\begin{tabular}{lllll}
  \hline
Method & Spring & Summer & Fall & Winter \\ 
  \hline
Full GP & 0.0345 & 0.0577 & 0.0379 & 0.0308 \\ 
  Low Rank - 100 & 0.0564 & 0.0895 & 0.0621 & 0.0494 \\ 
  Low Rank - 500 & 0.0462 & 0.0826 & 0.0517 & 0.0406 \\ 
  Taper - 90\% & 0.0341 & 0.0574 & 0.0374 & 0.0304 \\ 
  Taper - 20\% & 0.0344 & 0.0576 & 0.0378 & 0.0307 \\ 
  Vecchia - M=3 & 0.0358 & 0.0582 & 0.0395 & 0.0313 \\ 
  Vecchia - M=5 & 0.0351 & 0.0564 & 0.0382 & 0.0307 \\ 
  Vecchia - M=10 & 0.0346 & 0.0576 & 0.0379 & 0.0309 \\ 
  Vecchia - M=20 & 0.0345 & 0.0577 & 0.0378 & 0.0309 \\ 
       Neural Bayes& &   & & \\ 
   \hline
\end{tabular}
\caption{Tail-weighted CRPS for the NLDAS surface skin temperature application using the Mat\'ern covariance with smoothness $\nu = 1.5$.}
\normalsize
\label{tab:twcrps_nu2}
\end{table}

\begin{table}[!htbp]
\small
\centering
\begin{tabular}{lrrrr}
  \hline
Time & Spring & Summer & Fall & Winter \\ 
  \hline
Full GP & 136.12 & 138.20 & 137.52 & 131.84 \\ 
  Low Rank - 100 & 79.71 & 84.18 & 84.40 & 92.48 \\ 
  Low Rank - 500 & 82.57 & 84.93 & 85.59 & 89.86 \\ 
  Taper - 90\% & 140.47 & 139.08 & 139.10 & 138.47 \\ 
  Taper - 20\% & 366.89 & 362.20 & 361.73 & 359.20 \\ 
  Vecchia - M=3 & 57.05 & 57.94 & 57.42 & 53.40 \\ 
  Vecchia - M=5 & 64.06 & 66.12 & 65.15 & 59.60 \\ 
  Vecchia - M=10 & 96.28 & 94.71 & 97.27 & 89.20 \\ 
  Vecchia - M=20 & 238.62 & 181.00 & 207.23 & 190.15 \\ 
          Neural Bayes (training) &\multicolumn{4}{c}{$2\times 4320$} \\ 
          Neural Bayes (posterior median) &0.0003 & 0.0003  & 0.0003&0.0003 \\ 
          Neural Bayes (credible interval) &0.0006 & 0.0006  & 0.0006&0.0006 \\ 
   \hline
\end{tabular}
\caption{Model-fitting runtimes (in minutes) for the NLDAS surface skin temperature application using the Mat\'ern covariance with smoothness $\nu = 0.5$.}
\normalsize
\label{tab:walltime_nu1}
\end{table}

\begin{table}[!htbp]
\small
\centering
\begin{tabular}{lrrrr}
  \hline
Time & Spring & Summer & Fall & Winter \\ 
  \hline
Full GP & 139.90 & 141.80 & 139.72 & 139.97 \\ 
  Low Rank - 100 & 91.85 & 82.79 & 83.42 & 82.09 \\ 
  Low Rank - 500 & 83.76 & 83.44 & 84.01 & 81.58 \\ 
  Taper - 90\% & 137.27 & 139.16 & 139.14 & 138.11 \\ 
  Taper - 20\% & 357.34 & 361.93 & 360.72 & 357.18 \\ 
  Vecchia - M=3 & 54.87 & 55.99 & 55.79 & 53.32 \\ 
  Vecchia - M=5 & 60.07 & 59.93 & 60.73 & 57.80 \\ 
  Vecchia - M=10 & 76.52 & 73.12 & 76.97 & 71.98 \\ 
  Vecchia - M=20 & 122.65 & 112.63 & 123.30 & 112.34 \\ 
     Neural Bayes (training) &\multicolumn{4}{c}{$2\times 4320$} \\ 
          Neural Bayes (posterior median) &0.0003 & 0.0003  & 0.0003&0.0003 \\ 
          Neural Bayes (credible interval) &0.0006 & 0.0006  & 0.0006&0.0006 \\ 
   \hline
\end{tabular}
\caption{Model-fitting runtimes (in minutes) for the NLDAS surface skin temperature application using the Mat\'ern covariance with smoothness $\nu = 1.5$.}
\normalsize
\label{tab:walltime_nu2}
\end{table}

\begin{figure}[!t]
\centering
\includegraphics[width=0.9\textwidth]{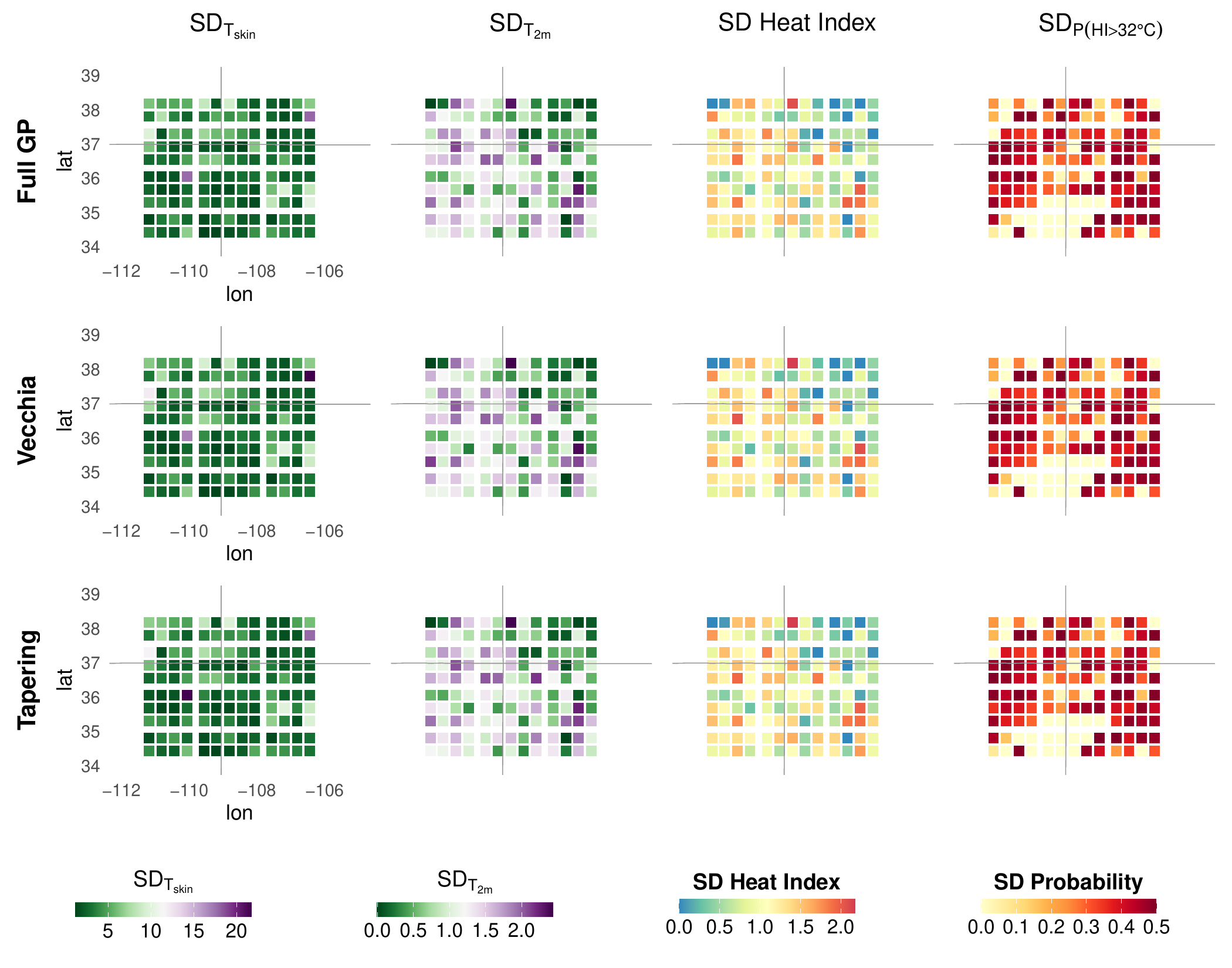}
\caption{Posterior standard deviation fields of surface skin temperature ($T_{skin}$), 2-m air temperature ($T_{2m}$), and heat index (HI), along with exceedance probabilities $\Pr(\mathrm{HI} > 32^\circ\mathrm{C})$ for the year 2000. Rows correspond to inference methods: Full GP {\bf (top)}, Vecchia approximation with $m=5$ {\bf (middle)}, and covariance tapering with $20\%$ sparsity {\bf (bottom)}.}
\label{fig:sd_posterior_summary_2000}
\end{figure}

\end{document}